\documentclass[12pt]{article}

\newcommand{\blind}{1}

\usepackage{graphicx}
\usepackage{enumerate}
\usepackage{url} 
\usepackage{longtable}
\usepackage[section]{placeins}
\usepackage[top=1in, left=1in, right=1in, bottom=1in]{geometry}
\usepackage[parfill]{parskip}	
\usepackage{caption}
\usepackage[normalem]{ulem}
\usepackage{subfigure}
\usepackage{subcaption} 
\usepackage{mwe} 
\usepackage{cancel} 
\usepackage{soul} 
\usepackage[export]{adjustbox}
\usepackage{amssymb}
\usepackage{epstopdf}
\DeclareSymbolFont{rsfs}{U}{rsfs}{m}{n}
\DeclareSymbolFontAlphabet{\mathscrsfs}{rsfs}
\usepackage{bbm}
\usepackage{bm}
\usepackage{amssymb}
\usepackage{mathtools}
\usepackage{xcolor}
\usepackage{amsfonts}
\usepackage{upgreek}

\usepackage{lscape}
\usepackage{enumitem}
\usepackage{rotating}
\usepackage{setspace}
\usepackage{threeparttable}
\usepackage{booktabs}
\usepackage[compact]{titlesec}
\usepackage{lmodern}
\fontfamily{lmtt}\selectfont
\usepackage[T1]{fontenc}

\usepackage[nocitation]{apacite}
\usepackage{natbib}
\bibpunct{(}{)}{;}{a}{}{,}
\usepackage{hyperref}
\usepackage{titling}
\usepackage{pifont}
\usepackage[capposition=top]{floatrow}
\newcommand{\subtitle}[1]{%
  \posttitle{%
    \par\end{center}
    \begin{center}\large#1\end{center}
    \vskip0.5em}%
}
\usepackage{amsthm}
\usepackage{amsmath} 
\usepackage{multirow}
\usepackage{mathrsfs}
\usepackage{float}
\usepackage[toc,page]{appendix}

\newcommand\independent{\protect\mathpalette{\protect\independenT}{\perp}}
\def\independenT#1#2{\mathrel{\rlap{$#1#2$}\mkern2mu{#1#2}}}

\newcommand{\Var}{\textrm{Var}}
\newtheorem{theorem}{Theorem}[section]
\newtheorem{lemma}[theorem]{Lemma}

\newtheorem{corollary}[theorem]{Corollary}

\newtheorem{assumption}{Assumption}
\newtheorem{definition}{Definition}[section]
\newtheorem{remark}{Remark}[section]

\graphicspath{ {./graphics/} }

\usepackage{mathrsfs}

\usepackage{color}

\newcommand{\rn}[1]{%
  \textup{\expandafter{\romannumeral#1}}%
}
\newcommand{\RN}[1]{%
  \textup{\uppercase\expandafter{\romannumeral#1}}%
}

\newcommand{\E}{\mathbb{E}}
\newcommand{\Pb}{\mathbb{P}}

\newcommand{\var}{\text{var}}

\newcommand{\Pn}{\mathbb{P}_n}
\newcommand{\Gn}{\mathbb{G}_n}
\newcommand{\R}{\mathbb{R}}

\makeatletter
\newcommand{\customlabel}[2]{%
   \protected@write \@auxout {}{\string \newlabel {#1}{{#2}{\thepage}{#2}{#1}{}} }%
   \hypertarget{#1}{#2}
}
\makeatother

\newtheorem{assumptionalt}{Assumption}[assumption]
\newenvironment{assumptionp}[1]{
  \renewcommand\theassumptionalt{#1}
  \assumptionalt
}{\endassumptionalt}

\newcommand{\argmin}{\mathop{\mathrm{argmin}}}
\makeatletter

\begin{document}

\def\spacingset#1{\renewcommand{\baselinestretch}%
{#1}\small\normalsize} \spacingset{1}


\if1\blind
{
  \title{\bf Causal K-Means Clustering}
  \author{Kwangho Kim\thanks{Department of Statistics, Korea University, 145 Anam-ro, Seongbuk-gu, Seoul 02841, Korea; email: \url{kwanghk@korea.ac.kr}.}  \hspace{2cm}
    Jisu Kim\thanks{Department of Statistics, Seoul National University, 1 Gwanak-ro, Gwanak-gu, Seoul 08826, Korea; email: \url{jkim82133@snu.ac.kr}.} \hspace{2cm}
    Edward H. Kennedy\thanks{Department of Statistics and Data Science, Carnegie Mellon University, 5000 Forbes Ave, Pittsburgh, PA 15213, USA; email: \url{edward@stat.cmu.edu}.}\\
    }
  \date{}
  \maketitle
} \fi

\if0\blind
{
  \bigskip
  \bigskip
  \bigskip
  \begin{center}
    {\LARGE\bf Causal k-Means Clustering}
\end{center}
  \medskip
} \fi

\bigskip
\begin{abstract}
Causal effects are often characterized with population summaries. These might provide an incomplete picture when there are heterogeneous treatment effects across subgroups. Since the subgroup structure is typically unknown, it is more challenging to identify and evaluate subgroup effects than population effects. We propose a new solution to this problem: \emph{Causal k-Means Clustering}, which leverages the k-means clustering algorithm to uncover the unknown subgroup structure. Our problem differs significantly from the conventional clustering setup since the variables to be clustered are unknown counterfactual functions. We present a plug-in estimator which is simple and readily implementable using off-the-shelf algorithms, and study its rate of convergence. We also develop a new bias-corrected estimator based on nonparametric efficiency theory and double machine learning, and show that this estimator achieves fast root-n rates and asymptotic normality in large nonparametric models. Our proposed methods are especially useful for modern outcome-wide studies with multiple treatment levels. Further, our framework is extensible to clustering with generic pseudo-outcomes, such as partially observed outcomes or otherwise unknown functions. Finally, we explore finite sample properties via simulation, and illustrate the proposed methods using a study of mobile-supported self-management for chronic low back pain.
\end{abstract}

\noindent%
{\it Keywords:} Causal inference; Heterogeneous treatment effect; Personalization; Subgroup analysis; Observational studies
\vfill

\newpage
\spacingset{1.25} 

\section{Introduction} \label{sec:introduction}

\subsection{Heterogeneity in Treatment Effects}
Statistical causal inference is concerned with how an outcome would change under an intervention on a cause of interest. Among causal estimands, the average treatment effect (ATE) is one of the most fundamental and extensively studied. For a binary treatment $A \in \{0,1\}$, the ATE is defined by
\begin{align} \label{def:ATE}
    \E(Y^1 - Y^0),
\end{align}
where $Y^a$ is the potential outcome that would have been observed under treatment $A = a$ \citep{rubin1974estimating}. For each unit, only one of $Y^0$ or $Y^1$ is observed, while the other remains unobserved and counterfactual. There has been lots of work concerning efficient and flexible estimation of the ATE and its analogs \citep[][]{kennedy2022semiparametric}. 

However, treatment effects often vary across subgroups in both magnitude and direction; some subgroups may experience larger effects than others, and a treatment may even benefit some subgroups while harming others. A potential shortcoming of the ATE is that it can mask this effect heterogeneity. Identifying treatment effect heterogeneity and the subgroups in which it arises plays an essential role in fields such as policy evaluation, drug development, and health care, and has received growing attention. For example, patients with different subtypes of cancer often react differently to the same treatment; however, our understanding of cancer subtypes at the molecular level is limited, and there is little consensus about which treatments are most effective for which patients \citep{kravitz2004evidence, hayden2009personalized}. Typically, the functional relationship between treatment effects and unit attributes is unknown a priori, so such heterogeneity must be explored using data-driven methods. Despite a growing body of recent work, this area remains relatively underexplored compared to other branches of causal inference \citep{kennedy2020optimal}, with several key challenges yet to be addressed.

To better understand treatment effect heterogeneity, investigators often target to estimate the conditional average treatment effect (CATE):
\begin{equation} \label{eqn:CATE}
\tau(X) = \E[Y^1-Y^0 \mid X],
\end{equation}
where $X\in\mathcal{X}$ is a vector of observed covariates. The CATE provides an individualized map of treatment effects over the covariate space, thereby enabling causal effect estimates tailored to each individual’s characteristics. Many methods have been proposed for CATE estimation, with recent work focusing on nonparametric and machine learning approaches that allow smooth, complex variation in effects across $X$. For example, existing approaches include loss-based super learning \citep{van2015targeted}, recursive partitioning \citep{athey2016recursive,zhang2017mining}, random forests \citep{foster2011subgroup,wager2018estimation}, support vector machine \citep{imai2013estimating}, weighted ensembles \citep{grimmer2017estimating}, neural network based on integral probability metrics \citep{shalit2017estimating}, meta-learning for unbalanced designs \citep{kunzel2017meta}, and reproducing kernel Hilbert space methods with oracle-efficiency guarantees \citep{nie2017quasi}. More recently, \citet{kennedy2020optimal} derived model-free error bounds and proposed an estimator attaining minimax-optimal convergence rates under smoothness conditions.

Some studies instead adopt a more structured framework for directly modeling effect heterogeneity \citep[e.g.,][]{shahn2017latent,suk2021hybridizing}. Rather than seeking to recover a fully individualized effect surface over $\mathcal{X}$, these approaches represent heterogeneity through a low-dimensional latent subgroup structure with subgroup-specific treatment effects. Accordingly, their primary objective is not fine-grained prediction for each covariate profile, but parsimonious identification of latent classes that account for systematic variation in treatment response.

\subsection{Understanding Heterogeneity via Cluster Analysis}

Existing work on treatment effect heterogeneity has largely emphasized supervised learning methods designed for accurate CATE estimation. While powerful for individualized prediction, such approaches are not primarily designed to reveal the underlying subgroup structure of treatment response. At the same time, existing parametric latent class approaches impose restrictive assumptions on the underlying structure of heterogeneity. In contrast, we study treatment effect heterogeneity from an unsupervised learning perspective. We propose \emph{causal clustering}, a new approach that uses tools from cluster analysis to identify subgroups with similar treatment effects in a flexible, nonparametric way. Our framework is therefore primarily descriptive and discovery-oriented, filling an important gap in the literature. 

\begin{figure}[t!]
\centering
\subfigure[]{\includegraphics[width=0.32\textwidth]{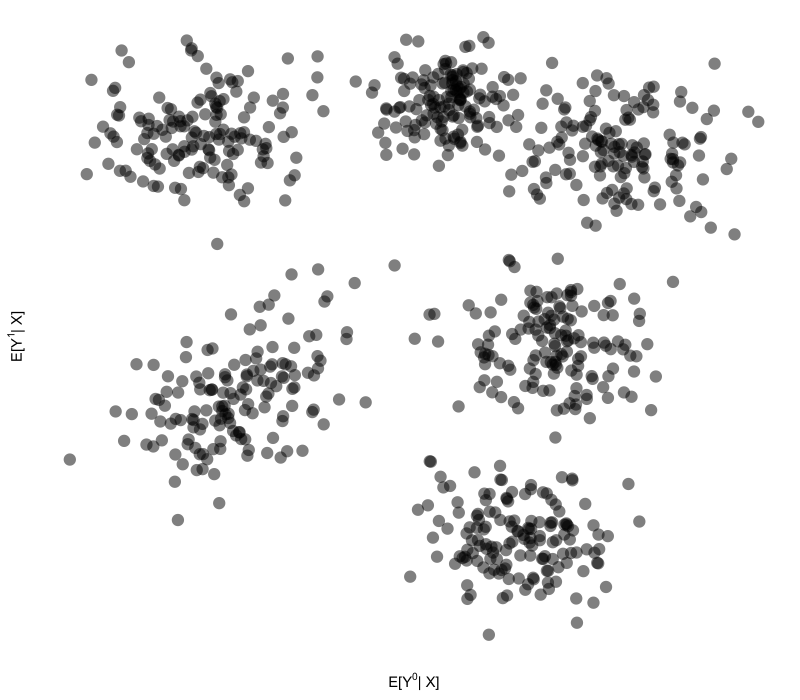}} 
\hfill%
\subfigure[]{\includegraphics[width=0.32\textwidth]{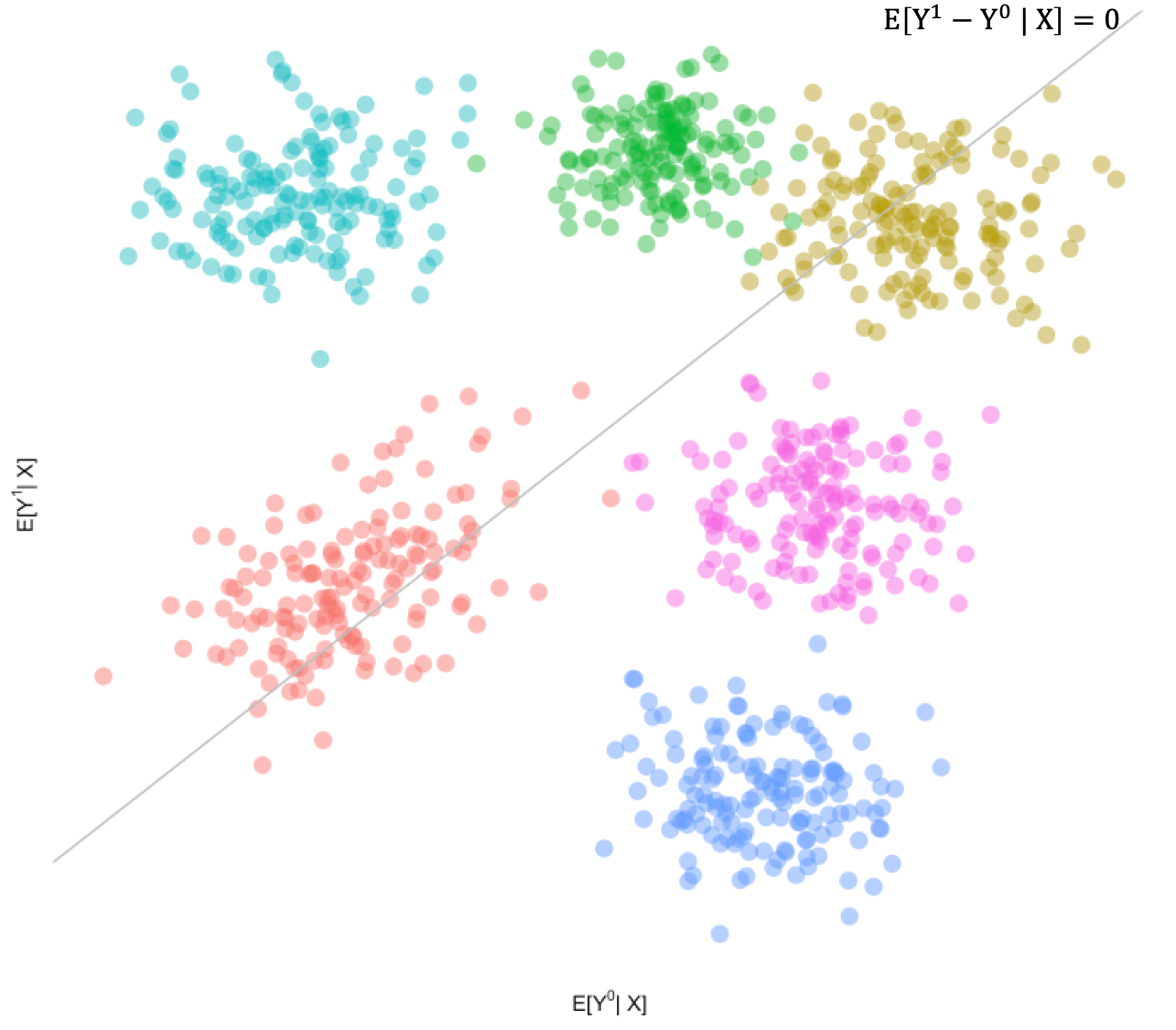}}
\hfill%
\subfigure[]{\includegraphics[width=0.32\textwidth]{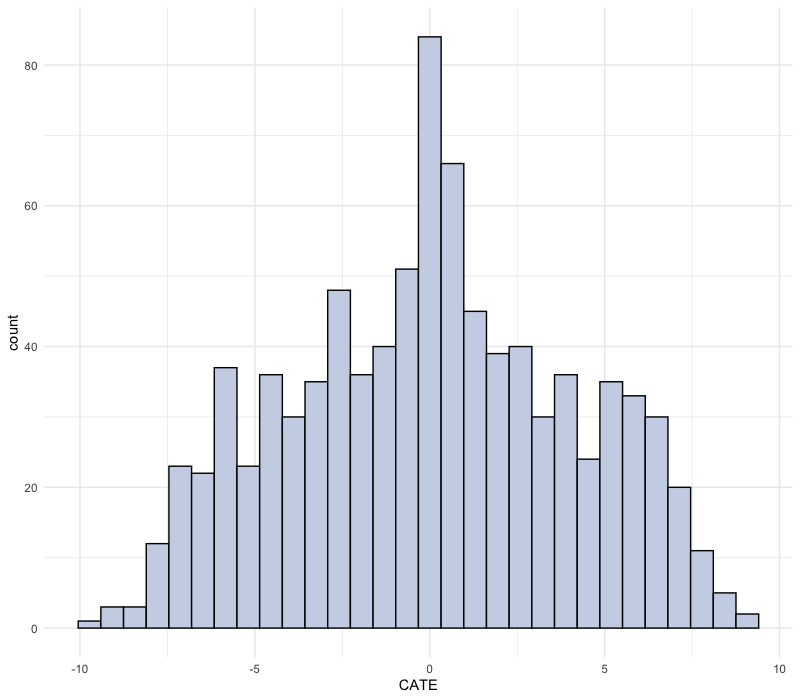}}
\hfill%
\caption{Illustration of causal clustering under binary treatment. (a) 600 points representing $(\mathbb{E}[Y^0 \mid X], \mathbb{E}[Y^1 \mid X])$, each corresponding to a unique covariate $X$, are generated with zero ATE; (b) We aim to uncover true subgroup structure with six clusters, with units within each cluster being more homogeneous in terms of the CATE; (c) The histogram fails to reveal the details about the true subgroup structure.
}
\label{fig:causal-cluster-illustration}
\end{figure}

We illustrate the idea of causal clustering through the case of binary treatments in Figure \ref{fig:causal-cluster-illustration}. We generate a sample where a projection $(\E[Y^0 \mid X], \E[Y^1 \mid X])$ of each observation is drawn from a mixture of six Gaussian distributions with different means and covariance functions, with the overall ATE set to zero. By construction, there are six clusters, with units within each cluster being more homogeneous in terms of the CATE. When it comes to analyzing the heterogeneity of treatment effects, people often rely on the histogram of the CATE as in Figure \ref{fig:causal-cluster-illustration}-(c). However, in this case, the histogram fails to reveal the details about the true subgroup structure. Adapting the idea of cluster analysis, we seek to identify subgroups whose treatment responses differ substantially from those of other subgroups while remaining relatively homogeneous within each subgroup, as illustrated in Figure \ref{fig:causal-cluster-illustration}-(b). This provides a new way to study the subgroup structure of treatment effect heterogeneity. To our knowledge, clustering methods have not been explicitly developed for this purpose in the causal inference literature.

Our problem differs significantly from the conventional clustering setup since the variable to be clustered consists of unknown functions (i.e., potential outcome regression functions) that must be estimated. Clustering with these unknown ``pseudo-outcomes'' has not received as much attention as clustering on standard fully observed data. Prior work has considered clustering with partially observed or noisy data, but still in fixed-dimensional vector settings. For example, \citet{serafini2020handling} studied clustering with missing data, \citet{haviland2011group} considered group-based trajectory modeling under nonrandom dropout, and \citet{su2018clustering} examined clustering with measurement error. In a related vein, \citet{kumar2007clustering} studied clustering based on unknown model parameters, albeit without theoretical guarantees. To the best of our knowledge, however, existing clustering methods have not addressed nonparametric clustering based on unknown functions. In our setting, we show that when the nuisance estimation error for these unknown functions is sufficiently small, the excess clustering risk is correspondingly small. In this sense, our work is conceptually analogous to the classification-versus-regression distinction in statistical learning \citep[][Theorem 2.2]{devroye2013probabilistic}.

In addition to existing supervised learning approaches, our framework provides a complementary tool for identifying subgroups with substantially different treatment responses. Our proposed methods are particularly useful in outcome-wide studies with multiple treatment levels \citep{vanderweele2017outcome, vanderweele2016association}; instead of probing a high-dimensional CATE surface, one may attempt to uncover lower-dimensional clusters with similar responses to a given treatment set. In addition, examining empirical covariate distributions within and across clusters can further clarify effect heterogeneity by highlighting baseline covariates most strongly associated with variation in treatment effects.

The remainder of the paper is structured as follows. In Section 2, we formalize the idea of causal clustering based on the k-means algorithm. In Section 3, we present a plug-in estimator, which is simple and readily implementable yet will in general not be $\sqrt{n}$-consistent. In Section 4, we develop an efficient bias-corrected estimator for k-means causal clustering under a margin condition, which attains fast $\sqrt{n}$ rates and asymptotic normality under weak nonparametric conditions. In section 5, we illustrate our approach using simulations and real  data on effects of treatment programs for substance abuse. Section 6 concludes with a discussion.\\

\section{Setup and estimands}
\label{sec:framework}

Consider a random sample $(Z_{1}, ... , Z_{n})$ of $n$ tuples $Z=(Y,A,X) \sim \Pb$, where $Y \in \R$ represents the outcome, $A \in \mathcal{A}= \{1,...,p\}$ denotes an intervention, and $X \in \mathcal{X} \subseteq \R^d$ comprises observed covariates. For simplicity, we focus on univariate outcomes, although the proposed methodology extends naturally to multivariate outcomes.
Throughout, we rely on the following widely-used identification assumptions \citep[e.g.,][Chapter 12]{imbens2015causal}:
\begin{assumptionp}{C1}[consistency] \label{assumption:A1-consistency}
$Y = Y^a$ if $A=a$. 
\end{assumptionp}
\begin{assumptionp}{C2}[no unmeasured confounding] \label{assumption:A2-no-unmeasured-confounding} 
 $A \independent Y^{a} \mid X$.
\end{assumptionp}
\begin{assumptionp}{C3}[positivity] \label{assumption:A3-positivity}
$\Pb(A=a \mid X)$ is bounded away from 0 a.s. $[\Pb]$.
\end{assumptionp}
For $a \in \mathcal{A}$, let the outcome regression function be denoted by
\begin{align*} 
\mu_a(X) &\equiv \E(Y^a \mid X) =  \E(Y \mid X, A=a).
\end{align*}
For $\forall a,a' \in \mathcal{A}$, one may define the pairwise CATE by
\begin{equation} \label{eqn::est-CATE}
\begin{aligned}
\tau_{aa'}(X) \equiv \E(Y \mid X, A=a) - \E(Y \mid X, A=a') 
= \mu_a(X) - \mu_{a'}(X)
\end{aligned}
\end{equation}

Then, we define the \textit{conditional counterfactual mean vector} $\mu: \mathcal{X} \to \R^p$ as
\begin{align} \label{eqn:CCMV-map}
    \mu(X) = \left[\E(Y^1 \mid X), \ldots ,\E(Y^{p} \mid X)\right]^\top.
\end{align}
If all coordinates of a point $\mu(X)$ were the same, there would be no treatment effect on the conditional mean scale. Also, adjacent units in the conditional counterfactual mean vector space would have similar responses to a given set of treatments, since for two units $i, j$, 
\[
\mu(X_{i}) \,\text{ close to } \, \mu(X_{j}) \Rightarrow \tau_{aa'}(X_{i}) \,\text{ close to } \, \tau_{aa'}(X_{j}) \quad \text{for all }\ a,a'\in \mathcal{A}.
\]
This provides vital motivation for uncovering subgroup structure via cluster analysis on projections of a sample onto the conditional counterfactual mean vector space via \eqref{eqn:CCMV-map}, whereby each cluster corresponds to a subgroup exhibiting similar treatment effects. Crucially, standard clustering theory is limited here since the variable to be clustered is $\mu(X)$, a vector of unknown regression functions, which themselves have to be estimated. Throughout, we suppress the explicit dependence of $\mu(X)$ on $X$ and denote it simply as $\mu(X) \equiv \mu$ when referring to it as a regression function.

\begin{figure}[t!]
\centering
\subfigure[]{\includegraphics[width=0.495\textwidth]{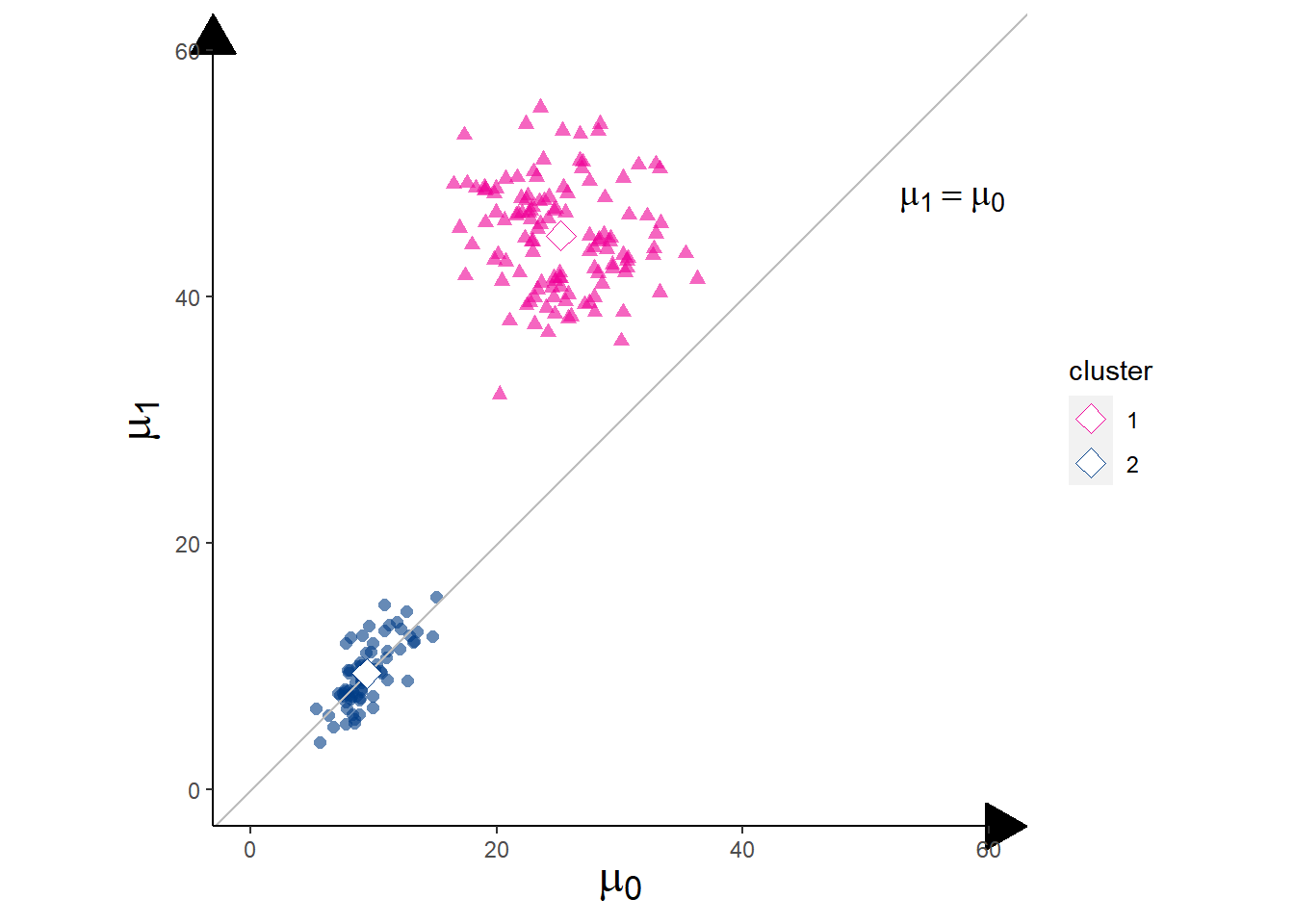}} 
\hfill%
\subfigure[]{\includegraphics[width=0.495\textwidth]{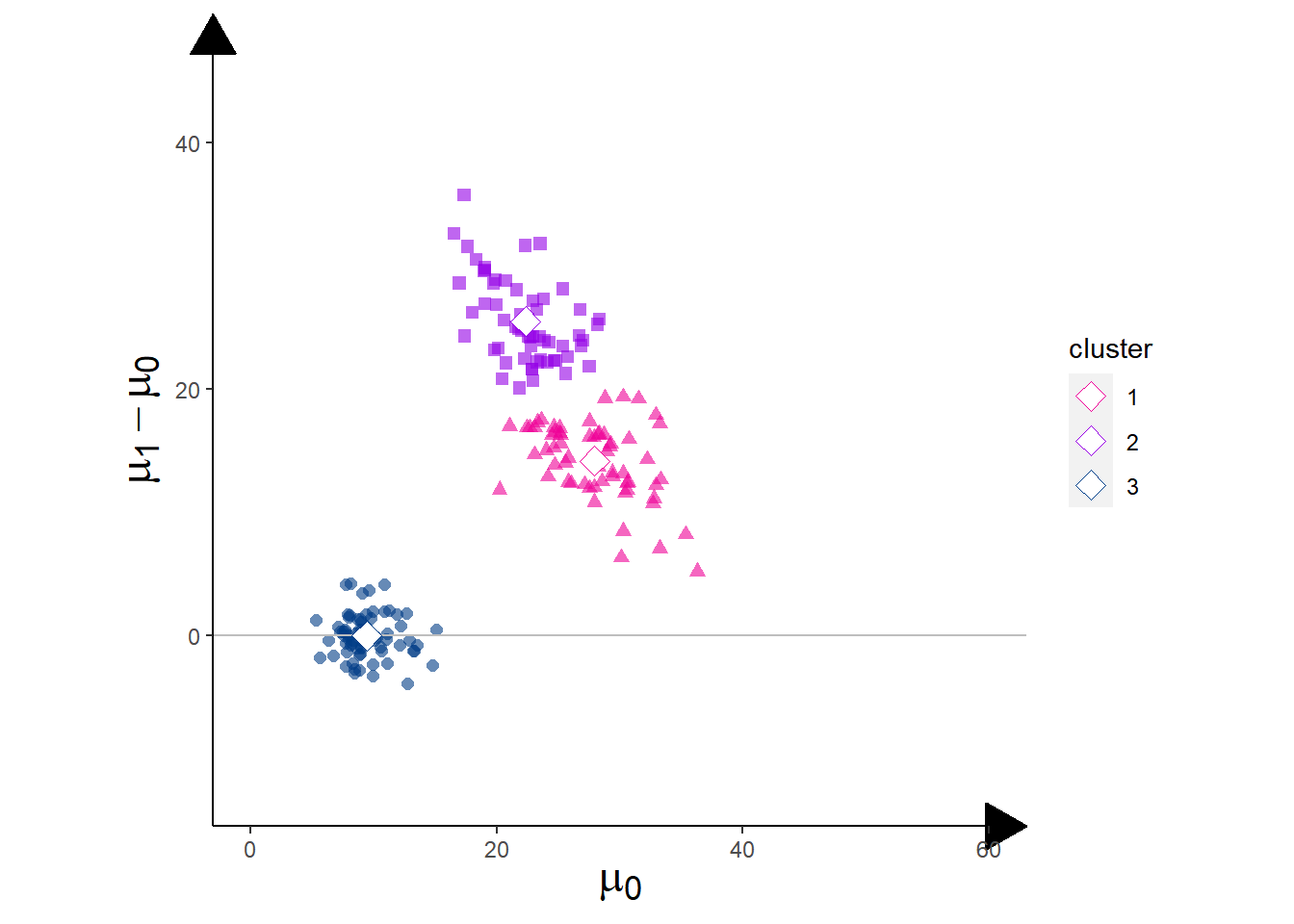}}
\caption{Consider a scenario where a treatment is ineffective for the low-risk patients but beneficial for those whose baseline risk $\mu_0$ exceeds a certain threshold. For example, the treatment effect could be near-zero for a group with $\mu_0\approx 10$, highly beneficial for $\mu_0\approx 20$, and moderately beneficial for $\mu_0\approx 30$. In this case, given the same data, cluster analysis with the parametrization $\mu = (\mu_0, \mu_1 - \mu_0)$ in (b) makes it easier to understand how treatment effects vary with the baseline risk than $\mu = (\mu_0, \mu_1)$ in (a).
}
\label{fig:alternative-parametrization}
\end{figure}

\begin{remark}
    The conditional counterfactual mean vector in \eqref{eqn:CCMV-map} can be readily reparametrized for a specific purpose without affecting our subsequent results. With $\mathcal{A} = \{0,1\}$, for instance, one may consider $\mu = (\mu_0, \mu_1 - \mu_0)$ with $A=0$ untreated and $\mu_0$ as a baseline risk instead of $\mu = (\mu_0, \mu_1)$. As illustrated in Figure \ref{fig:alternative-parametrization}, this may be more useful for exploring the relationship between the baseline risk and the treatment effect. As shown in the heterogeneous treatment effect literature, contrasts of regression functions are often structurally simpler than the regression functions themselves \citep[e.g.,][]{chernozhukov2018generic,kennedy2020optimal}. Accordingly, alternative parametrizations may better exploit meaningful structure in the CATE functions, such as smoothness or sparsity. For instance, when baseline risk is not of primary interest, clustering based on treatment contrasts, $\mu = (\mu_1-\mu_0,\mu_2-\mu_0,\ldots)$, may be more effective than clustering on $\mu = (\mu_0,\mu_1,\mu_2,\ldots)$. If we are interested in how a treatment shifts the quantiles \citep[e.g.][]{chernozhukov2005iv, zhang2012causal}, we can redefine our conditional counterfactual mean vector by $\mu=(Q_0(q), Q_1(q))$ for some prespecified $q \in (0,1)$ (for median, $q=1/2$), where $Q_a(q)$ is the quantile function of our potential outcome $Y^a$, i.e., $Q_a(q) = \inf\left\{ y\in \mathbb{R} : q \le F_{Y^a}(y) \right\}$ for $F_{Y^a} = \Pb(Y^a \leq y \mid X)$. 
\end{remark}

In this work, we develop a causal clustering framework based on $k$-means. The $k$-means algorithm, also known as vector quantization, is one of the oldest and most widely used clustering methods. It identifies $k$ representative points, or cluster centers, that induce a Voronoi partition of the space. The method has been extensively studied in the clustering literature; see \citet{jain2010data} for a review and \citet{graf2007foundations} for a thorough account. Among algorithmic clustering methods, $k$-means has received especially strong theoretical attention, due in part to its close connection to principal components analysis \citep{ding2004k}.
	
We call a set of $k$ representative points a \emph{codebook} $C= \{c_1, ... , c_k\}$ where each $c_j \in \mathbb{R}^p$. Let $\Pi_C(x)$ be the projection of $x \in \mathbb{R}^p$ onto $C$:
$$
\Pi_C(x) = \underset{c\in C}{\argmin} \Vert c - x \Vert_2^2.
$$
Then we define the \textit{population clustering risk} $R(C)$ with respect to $\mu$ by
\begin{align} \label{eqn:population-clustering-risk}
R(C) = \E \Vert \mu - \Pi_C(\mu)  \Vert_2^2,
\end{align}
and the corresponding optimal codebook $C^*$ by
\begin{align} \label{eqn:population-optimal-codebook}
C^* = \underset{C\in \mathcal{C}_k}{\argmin} R(C),
\end{align}
where $\mathcal{C}_k$ denotes all codebooks of length $k$ in the image of $\mu$. When $C$ is fixed, the population clustering risk \eqref{eqn:population-clustering-risk} can be viewed as a real-valued functional on a nonparametric model. Each cluster center $c_j \in C$ corresponds to the vector of subgroup average potential outcomes $\{\mathbb E(Y^a \mid X \in R_j)\}_{a\in\mathcal{A}}$, where $R_j \subset \mathcal{X}$ denotes the covariate region assigned to cluster $j$. Importantly, $R(C)$ is a nonsmooth functional of the observed data distribution, and thus standard semiparametric efficiency theory is not directly applicable. In Section \ref{sec:efficient-causal-clustering}, we develop an efficient estimator of $R(C^*)$ under a margin condition.

\textbf{Notation.} 
In the sequel, we use the shorthand $\mu_{(i)} \equiv \mu(X_i) = \left[\mu_1(X_i), ... , \mu_p(X_i)\right]^\top$ and $\widehat{\mu}_{(i)} \equiv \widehat{\mu}(X_i) = \left[\widehat{\mu}_1(X_i), ... , \widehat{\mu}_p(X_i)\right]^\top$. We let $\Vert x \Vert_q$ denote $L_q$ norm for any fixed vector $x$. For a given function $f$, we use the notation
$\Vert f \Vert_{\Pb,q} = \left[\Pb (\vert f \vert^q) \right]^{1/q} = \left[\int \vert f(z)\vert^q d\Pb(z)\right]^{1/q}$
as the $L_{q}(\mathbb{P})$-norm of $f$. Also, we let ${\Pb}$ denote the conditional expectation given the sample operator $\hat{f}$, as in $\mathbb{P}(\hat{f})=\int\hat{f}(z)d\mathbb{P}(z)$. Notice that $\Pb(\hat{f})$ is random only if $\hat{f}$ depends on samples, in which case $\Pb(\hat{f}) \neq \E(\hat{f})$. Otherwise $\Pb$ and $\E$ can be used exchangeably. For example, if $\hat{f}$ is constructed on a separate (training) sample $\mathsf{D}^n = (Z_1,...,Z_n)$, then ${\Pb}\left\{\hat{f}(Z)\right\} = \E\left\{\hat{f}(Z) \mid \mathsf{D}^n \right\}$ for a new observation $Z \sim \Pb$. We let $\Pn$ denote the empirical measure as in $\Pn(f)=\Pn\{(f(Z)\}=\frac{1}{n}\sum_{i=1}^n f(Z_i)$. Lastly, we use the shorthand $a_n \lesssim b_n$ to denote $a_n \leq \mathsf{c} b_n$ for some universal constant $\mathsf{c} > 0$. \\

\section{Plug-in Estimator}
\label{sec:plug-in-estimator}

If the $\{\mu_{(i)}\}$ were all known, then for a fixed number of clusters $k$, the optimal codebook $C^*$ could be estimated by minimizing the empirical clustering risk, just as in standard $k$-means clustering:
\begin{equation} \label{eqn:conventional-codebook-estimator}
\begin{gathered} 
\widehat{C}^* = \underset{C\in \mathcal{C}_k}{\argmin} R_n(C), \\
\text{where } \quad R_n(C) = \frac{1}{n} \sum_{i=1}^{n} \Vert \mu_{(i)} - \Pi_C(\mu_{(i)})  \Vert_2^2.
\end{gathered}
\end{equation}

The common method used to compute $\widehat{C}^*$ is Lloyd's algorithm \citep{lloyd1982least,kanungo2002efficient}, yet there are other recent developments as well \citep{leskovec2020mining}. A solution of such algorithms normally depends on the starting values. Some popular methods for choosing good starting values are discussed in, for example, \citet{tseng2005tight, arthur2007k}.

The problem of assessing how well $\widehat{C}^*$ approximates the true $C^*$ has been extensively studied. \citet{pollard1981strong} proved strong consistency of k-means clustering in the sense that $\widehat{C}^* \xrightarrow{a.s.} C^*$ as well as $R(\widehat{C}^*) - R(C^*) \xrightarrow{a.s.} 0$. Borrowing techniques from statistical learning theory, \citet{linder1994rates} and \citet{biau2008performance} showed that when an input vector is almost surely bounded, the expected excess risk may decay at $O(\sqrt{\log n/n})$ and $O(1/\sqrt{n})$ rates, respectively. More recently, it has been shown that faster $O(\log n/n)$ or $O(1/n)$ rates can be attained under a margin condition on the source distribution \citep[][]{levrard2015nonasymptotic, levrard2018quantization}; we shall go over this margin condition in detail shortly.

However, in our setting, the estimator $\widehat{C}^*$ in \eqref{eqn:conventional-codebook-estimator} is not available because  $\{\mu_{(i)}\}$ are unobserved. Instead, we propose the following plug-in estimator 
\begin{equation} \label{def:plug-ins-k-means-codebook}
\begin{gathered} 
\widehat{C} = \underset{C\in \mathcal{C}_k}{\argmin} \widehat{R}_n(C), \\
\text{where } \quad \widehat{R}_n(C) = \frac{1}{n} \sum_{i=1}^{n} \Vert \widehat{\mu}_{(i)} - \Pi_C(\widehat{\mu}_{(i)})  \Vert_2^2,
\end{gathered}
\end{equation}
where $\widehat{\mu}$ is some initial estimator of the outcome regression functions in \eqref{eqn:CCMV-map}. We will use sample splitting to avoid imposing empirical process conditions on the function class of $\mu$ \citep[e.g.,][]{kennedy2016semiparametric,Chernozhukov17}. For now, we suppose that $\widehat{\mu}$ are constructed on a separate, independent sample; this will be discussed in more detail in the following section. 

Due to the non-smoothness of the projection function $\Pi_C(\cdot)$, in general we would not expect the proposed plug-in estimator \eqref{def:plug-ins-k-means-codebook} to inherit the rate of convergence of $\widehat{\mu}$. To resolve this, we shall assume that the source distribution $\Pb$ is concentrated around $C^*$, in the spirit of \citet{levrard2015nonasymptotic, le2018notion, levrard2018quantization}, as made precise shortly.

In the sequel, the set of minimizers of the clustering risk will be denoted by $\mathcal{C}^*_k$, i.e., $\mathcal{C}^*_k = \{C^* \in \mathcal{C}_k : R(C^*)=\underset{C\in \mathcal{C}_k}{\min}R(C) \}$. For $C^* \in \mathcal{C}^*_k$, we define the \emph{Voronoi cell} associated with the cluster center $c^*_i$ as the closed set by
$$
V_i(C^*) = \left\{ \mu \mid  \Vert \mu - c^*_i \Vert_2 \leq \Vert \mu - c^*_j \Vert_2, \forall j \neq i \right\},
$$
and its boundary by
\[
\partial V_i(C^*) = \left\{ \mu \mid \Vert \mu - c^*_i \Vert_2 = \Vert \mu - c^*_j \Vert_2, \forall j \neq i \right\}.
\]
And we write the entire boundaries induced from $C^*$ as
$$
\partial C^* = \underset{i}{\bigcup} \partial V_i(C^*).
$$
Next, for any $C^*$ and some $t>0$, define
\[
N_{C^*}(t) = \underset{j}{\bigcup} \left\{ \mu \in  V_j(C^*) \Bigm\vert \left\vert \Vert \mu - c^*_j \Vert_2 - \underset{i \neq j}{\min}  \Vert \mu - c^*_i \Vert_2 \right\vert  \leq t \right\}.
\]
\begin{figure}[t!]
    \centering
    \includegraphics[width=.5\linewidth]{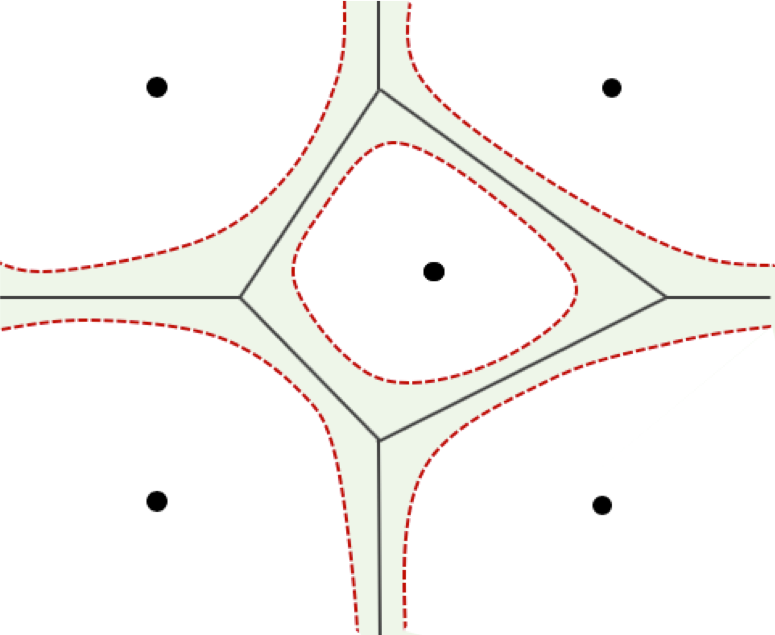}
    \caption{Illustration of the margin condition in Definition \ref{def:margin-condition}, where we control the probability mass in the shaded area within the red-dashed lines specified by $\kappa$.}
    \label{fig:two-margin-conditions}
\end{figure}
The set $N_{C^*}(t)$ may be viewed as a neighborhood of $\partial C^*$ consisting of points $\mu$ for which the distances to the two nearest cluster centers differ by at most $t$. For example, in two-dimensional Euclidean space, that is, when $p=2$, the set $N_{C^*}(t)$ consists of regions bounded by hyperbolas symmetric about the common Voronoi boundaries $\{\partial V_i(C^*) \cap \partial V_j(C^*) \mid i \neq j,\; i,j \in \{1,\ldots,k\}\}$, as illustrated in Figure \ref{fig:two-margin-conditions}. Now we introduce the following \textit{margin condition}.
\begin{definition}[Margin condition] \label{def:margin-condition} 
	A distribution $\Pb$ satisfies a margin condition with radius $\kappa>0$ and rate $\alpha>0$ if and only if for all $0 \leq t \leq \kappa$,
	\begin{align*}
	\underset{C^*  \in \mathcal{C}^*_k}{\sup} \Pb(\mu  \in N_{C^*}(t)) \lesssim  t^\alpha.
	\end{align*}
\end{definition}
The margin condition above requires local control of the probability mass near $\partial C^*$ for each $C^* \in \mathcal{C}_k^*$, and thus ensures that every optimal codebook induces a natural classification. A larger $\alpha$ indicates that $\Pb$ is ``more structured", facilitating the formation of such a natural classifier, whereas a smaller $\alpha$ suggests that a natural classifier is less likely to exist. When $\alpha<1$, probability mass concentrates near $\partial C^*$ at a rate faster than would be implied by a bounded density, implying that the distribution is not well separated there. The strong margin condition with exponent $\alpha = 1$ has been employed in standard k-means clustering to obtain fast $O(1/n)$ rates of convergence for the excess risk \citep{levrard2015nonasymptotic, levrard2018quantization}, or to establish strong stability for Lloyd’s algorithm \citep{le2018notion}. This type of margin condition, which controls the probability mass near the critical region, is common in causal inference problems involving nonsmooth target parameters \citep[e.g.,][]{van2015targeted,luedtke2016statistical,kennedy2018sharp,levis2023covariate,pmlr-v202-kim23ab}. We next introduce the following mild boundedness and consistency assumptions as well.

\begin{assumptionp}{A1} 
\label{assumption:A1-boundedness}
$\Vert \mu_a \Vert_\infty , \Vert \widehat{\mu}_a \Vert_\infty \leq B < \infty$ a.s.
\end{assumptionp}
\begin{assumptionp}{A2} 
\label{assumption:A2-consistency}
$\max_a \left\Vert \widehat{\mu}_{a}-\mu_{a}\right\Vert_{\infty} = o_\Pb(1)$.
\end{assumptionp}

In the next theorem, we give upper bounds of the excess risk, showing that the proposed plug-in estimator \eqref{def:plug-ins-k-means-codebook} is risk consistent. 

\begin{theorem}
    \label{thm:k-means} 
    Suppose $\Pb$ satisfies the margin condition with some $\kappa > 0$, $\alpha > 0$. Let
    \[
    R_{1,n} = \max_a \left\Vert  \widehat{\mu}_a - \mu_a \right\Vert_{\Pb,1} + \max_a \Vert \widehat{\mu}_a - \mu_a \Vert_{\infty}^{\alpha + 1}  + \frac{1}{\kappa}\max_a \left(\Vert \widehat{\mu}_a - \mu_a \Vert_{\infty} \left\Vert \widehat{\mu}_a - \mu_a \right\Vert_{\Pb,1}\right).
    \]    
    Then under Assumptions \ref{assumption:A1-boundedness}, \ref{assumption:A2-consistency}, we have
    \begin{align*}    
    \Pb\left\{R(\widehat{C})-R(C^{*})\right\} = O\left( \frac{1}{\sqrt{n}} + R_{1,n} \right) \quad \text{and} \quad
    R(\widehat{C})-R(C^{*}) = O_\Pb\left( \sqrt{\frac{\log n}{n}} + R_{1,n} \right), 
    \end{align*}
    whenever $\widehat{\mu}$ is constructed from a separate independent sample. 
\end{theorem}

A proof of the above theorem and all subsequent proofs can be found in Web Appendix \ref{appendix-proofs}. The term $\Vert \widehat{\mu}_a - \mu_a \Vert_{\infty}^{\alpha + 1}$ in $R_{1,n}$ is standard in the literature on estimation of nonsmooth functionals under margin conditions, including the works cited above. The term $\frac{1}{\kappa} \Vert \widehat{\mu}_a - \mu_a \Vert_{\infty} \left\Vert \widehat{\mu}_a - \mu_a \right\Vert_{\Pb,1}$ arises because the margin condition in Definition \ref{def:margin-condition} imposes only local control over the neighborhood $N_{C^*}(\kappa)$; this term disappears when $\kappa \to \infty$. Theorem \ref{thm:k-means} essentially states that the extra price we pay for excess risk is the estimation error of the outcome regression functions.

Risk consistency of $\widehat{C}$ does not by itself imply that $\widehat{C}$ is close to the true codebook $C^*$. To establish consistency of $\widehat{C}$, we require the following additional condition.

\begin{assumptionp}{A3} 
\label{assumption:A3-uniqueness-of-codebook}
        $C^*$ is unique up to relabeling of its coordinates.
\end{assumptionp}


The uniqueness condition in Assumption \ref{assumption:A3-uniqueness-of-codebook} has also been used in earlier work by \citet{pollard1981strong, pollard1982central}. The next theorem states that the proposed plug-in estimator is consistent. 

\begin{theorem}
    \label{thm:plug-in-consistency}
    Under Assumptions \ref{assumption:A1-boundedness} - \ref{assumption:A3-uniqueness-of-codebook}, $\widehat{C}$ computed by the plug-in estimator \eqref{def:plug-ins-k-means-codebook} converges in probability to $C^*$.
\end{theorem}

The map $C \mapsto R(C)$ from $\R^{kp}$ into $\R$ is differentiable if $\Vert \mu \Vert_{\Pb,2} < \infty$ \citep{pollard1982central}. Based on Theorems \ref{thm:k-means} and \ref{thm:plug-in-consistency}, one may thus characterize the rate of convergence of $\widehat{C}$, as stated in the next corollary.

\begin{corollary}\label{cor:plug-in-convergence-rate}
    Suppose that $\Pb$ satisfies the margin condition with some $\kappa > 0$, $\alpha > 0$, and that Assumptions \ref{assumption:A1-boundedness} - \ref{assumption:A3-uniqueness-of-codebook} hold. Also assume that $\widehat{\mu}$ is constructed from a separate independent sample. Then
    $
        \Vert \widehat{C} - C^* \Vert_1 = \sum_{j=1}^k \Vert \widehat{c}_j - c^*_j \Vert_1 = O_\Pb\left( \sqrt{\frac{\log n}{n}} + R_{1,n} \right).
    $
\end{corollary}

The plug-in estimator is simple and intuitive. When an initial estimator is available or $\widehat{\mu}$ is fitted in a separate independent sample, \eqref{def:plug-ins-k-means-codebook} is readily implementable using the standard, off-the-shelf algorithms including Lloyd's algorithm. Alternatively, we may estimate the risk via cross-fitting, where we swap the samples, repeat the procedure, and average the results to regain full sample size efficiency. Then we compute the optimal codebook that minimizes the estimated risk. We shall address this in further detail in the following section.

Note that the convergence rate in Theorem \ref{thm:k-means} essentially inherits from $\widehat{\mu}$. Hence, for either the risk or the codebook, rates of convergence would be expected to be slower than $\sqrt{n}$ with non-normal limiting distributions not centered at the true parameter, unless careful undersmoothing of particular estimators (e.g., splines) is used. Consequently, valid confidence intervals, even via bootstrap, may not be constructed. Substituting doubly robust scores for $\mu$ in the risk function may appear promising for improving convergence rates. However, unlike standard smooth causal parameters, this does not guarantee recovery of the efficient influence function because of the intrinsic nonsmoothness of the risk functional. In the following section, we will develop an estimator that can be $\sqrt{n}$ consistent and asymptotically normal even if the nuisance functions are estimated flexibly at slower than $\sqrt{n}$ rates, in a wide variety of settings.

\begin{remark} \label{rmk:choice-of-k}
    In this work, we consider the number of clusters k as fixed and do not address its optimal selection. We conjecture that existing methods such as the Elbow method or general tuning parameter selection techniques may be adapted to our framework. In practice, the choice of $k$ often reflects the investigator's goal: e.g., identifying two highly contrasting subgroups with $k = 2$. Developing data-driven strategies for selecting $k$ remains an important direction for future work.\\
\end{remark}

\section{Semiparametric Estimator}
\label{sec:efficient-causal-clustering}

In this section, we develop estimators that attain faster convergence rates than the plug in estimator in Section \ref{sec:plug-in-estimator} by leveraging semiparametric efficiency theory.

\subsection{Proposed estimator}

For convenience, we introduce the following additional notation
\begin{equation} \label{eqn:notation-IFs}
\begin{aligned}
    \pi_a(X) & = \Pb\left(A=a \mid X\right), \\
    \varphi_{1,a}(Z;\eta_a) &= \frac{\mathbbm{1}(A=a)}{\pi_a(X)}\left\{Y- \mu_A(X)\right\} + \mu_a(X), \\
    \varphi_{2,a}(Z;\eta_a) &= 2\mu_a(X)\frac{\mathbbm{1}(A=a)}{\pi_a(X)}\left\{Y- \mu_A(X)\right\} + \mu_a^2(X), 
\end{aligned}
\end{equation}
where $\eta_a=\{ \pi_a, \mu_a \}$ denotes a set of relevant nuisance functions, $\forall a \in \mathcal{A}$. $\pi_a$ is a conditional probability of receiving the treatment $a$; when $p=2$, $\pi_1$ denotes the propensity score. $\varphi_{1,a}$ and $\varphi_{2,a}$ are the uncentered efficient influence function for the parameters $\E\left\{ \mu_a(X) \right\}$ and $\E\left\{ \mu_a^2(X) \right\}$, respectively. The efficient influence function is important to construct optimal estimators since its variance equals the efficiency bound (in asymptotic minimax sense). By leveraging the efficient influence function, one can obtain desirable properties such as double robustness and reduced second order bias, thereby weakening the nonparametric conditions required for nuisance function estimation. For further background on influence functions and semiparametric efficiency theory, see, for example, \citet{vanderVaart2002semiparametric,tsiatis2007semiparametric,kennedy2016semiparametric,kennedy2022semiparametric}.

Next, for any fixed $C \in \mathcal{C}_k$, define
\begin{equation} \label{eqn:eif-of-risk}
\begin{aligned}
\varphi_{C}(Z;\eta) = \sum_{a\in\mathcal{A}} \left\{\varphi_{2,a}(Z;\eta_a) - 2\varphi_{1,a}(Z;\eta_a)\left[\Pi_C\left(\mu \right)\right]_a + \left[\Pi_C\left(\mu\right)\right]^2_a \right\},
\end{aligned}
\end{equation}
where we let $\eta = \left\{ \eta_a \right\}_{a \in \mathcal{A}}$ denote a set of all nuisance functions collectively, and $\left[\Pi_C\left(\mu \right)\right]_a$ be the $a$-th element of the projection $\Pi_C\left(\mu \right)$. Then $\varphi_{C^*}(Z;\eta)$ is the uncentered efficient influence function for $R(C^*)$ whenever $\Pb$ satisfies the margin condition, as formally stated below. 

\begin{lemma} \label{lem:eif-of-risk}
    Suppose that Assumptions \ref{assumption:A1-boundedness}, \ref{assumption:A2-consistency} hold, and that $\Pb$ satisfies the margin condition with some $\kappa > 0$ and $\alpha > 0$. If, for every optimal codebook $C^* \in \mathcal{C}_k^*$, we let $\phi_{C^*}(z;\Pb) = \varphi_{C^*}(z;\Pb) - \int \varphi_{C^*}(z;\Pb)d\Pb$, then $\phi_{C^*}$ is the efficient influence function for $R(C^*)$.
\end{lemma}

Based on Lemma \ref{lem:eif-of-risk}, we may construct an efficient semiparametric estimator for $R(C)$ by de-biasing the plug-in estimator in \eqref{def:plug-ins-k-means-codebook}. There are two main approaches for constructing such semiparametric estimators; one is based on empirical process conditions, and the other is to use sample
splitting. Following \citet{robins2008higher, zheng2010asymptotic, Chernozhukov17, newey2018cross, kennedy2020optimal} and many others, we use {sample splitting} (or {cross-fitting}) to allow for arbitrarily complex nuisance estimators $\widehat{\eta}$. Specifically with fixed $K$, we split the data into $K$ disjoint groups, each with size $n/K$ approximately, by drawing variables $(B_1,\ldots, B_n)$ independent of the data; $B_i=b$ indicates that subject $i$ was split into group $b \in \{1,\ldots,K\}$. This could be done, for example, by drawing each $B_i$ uniformly from $\{1, \ldots, K\}$. Then we propose our estimator for $R(C)$ as 
\begin{align} \label{eqn:eif-cluster-risk-estimator}
    \widehat{R}(C) &= \sum_{b=1}^K \left\{\frac{1}{n} \sum_{i=1}^n \mathbbm{1}(B_i=b) \right\} \Pn^b\left\{ \varphi_{C}(Z;\widehat{\eta}_{-b})\right\} \nonumber \\
    & \equiv \Pn\left\{  \varphi_{C}(Z;\widehat{\eta}_{-K}) \right\},
\end{align}
where we let $\Pn^b$ denote empirical averages only over the set of units $\{i : B_i=b\}$ in group $b$ and let $\widehat{\eta}_{-b}$ denote the nuisance estimator constructed only using those units $\{i : B_i \neq b\}$. If one is willing to assume that the nuisance function class and corresponding estimators are not too complex (e.g., Donsker or low entropy conditions), then $\eta$ can be estimated on the same sample without sample splitting \citep[e.g.,][]{kennedy2016semiparametric, kennedy2022semiparametric}. 

Subsequently, we estimate the optimal codebook $C^*$ by minimizing $\widehat{R}(C)$:
\begin{align} \label{eqn:eif-codebook-estimator}
    \widehat{C} = \underset{C \in \mathcal{C}_k}{\argmin}\widehat{R}(C).    
\end{align}
Note that the cross-fitting procedure described above is equally applicable to the plug-in estimator \eqref{def:plug-ins-k-means-codebook}. \eqref{eqn:eif-codebook-estimator} can be computed using first-order (e.g., gradient descent) or second-order (e.g., Newton-Raphson) methods based on the derivative formulas \eqref{eqn:varphi-derivative} and \eqref{eqn:varphi-Hessian} specified in the following section. In practice, a generalized Lloyd–type block coordinate descent algorithm may also be effective, as it computes exact per-step minimizers and often yields faster and more stable convergence for the semiparametric objective without requiring explicit gradient calculations.

\subsection{Asymptotic Properties} \label{sec:eff-k-means-theoretical-properties}

In this subsection, we study the asymptotic properties of the proposed semiparametric estimator in \eqref{eqn:eif-codebook-estimator}. For notational convenience, we first define the remainder term that appears in our results:
\begin{align*}    
    R_{2,n} &= \max_a \left\{\Vert \widehat{\mu}_a - {\mu}_a \Vert_{\Pb,2} \left( \Vert \widehat{\mu}_a - {\mu}_a \Vert_{\Pb,2} + \Vert \widehat{\pi}_a - {\pi}_a \Vert_{\Pb,2} \right) \right\} + \max_a \Vert \widehat{\mu}_a - \mu_a \Vert_{\infty}^{\alpha+1}\\
    &\quad + \frac{1}{\kappa}\max_a \left(\Vert \widehat{\mu}_a - \mu_a \Vert_{\infty} \left\Vert \widehat{\mu}_a - \mu_a \right\Vert_{\Pb,1}\right).
\end{align*}

Note that all terms in $R_{2,n}$ are second order, unlike those in $R_{1,n}$ from the previous section.
We impose the following additional assumptions on nuisance estimation.

\begin{assumptionp}{A4} 
\label{assumption:A4-pihat-boundedness}
        $\Pb\left\{ \epsilon \leq \widehat{\pi}_a(X) \leq 1 - \epsilon \right\} = 1$ for some $\epsilon > 0$.
\end{assumptionp}

\begin{assumptionp}{A5} 
\label{assumption:A5-np-consistency-condition}
       $\underset{a}{\max} \left\{ \Vert \widehat{\pi}_a - {\pi}_a \Vert_{\Pb,2} + \Vert \widehat{\mu}_a - {\mu}_a \Vert_{\infty} \right\} = o_\Pb(1)$.
\end{assumptionp}

\begin{assumptionp}{A6} 
\label{assumption:A6-np-remainder-condition}
        $R_{2,n} = o_\Pb(n^{-1/2})$.
\end{assumptionp}

Assumption \ref{assumption:A5-np-consistency-condition} is a mild consistency assumption, with no requirement on rates of convergence. Assumption \ref{assumption:A6-np-remainder-condition} may hold, for example, under standard $n^{-1/4}$-type rate conditions on $\widehat{\eta}$ which can be attained under smoothness, sparsity, or other structural constraints \citep[e.g.,][]{kennedy2016semiparametric}.  

Lemma \ref{lem:eif-of-risk} yields conditions under which $\widehat{R}(C^*)$ is asymptotically normal and efficient for $R(C^*)$ for any $C^*$ satisfying the margin condition, as formalized below.

\begin{lemma} \label{lem:root-n-CAN-Rhat}
    Suppose that the margin condition is satisfied with some $\alpha >0$, $\kappa >0$, and that Assumptions \ref{assumption:A1-boundedness}, \ref{assumption:A4-pihat-boundedness}, \ref{assumption:A5-np-consistency-condition}, and \ref{assumption:A6-np-remainder-condition} hold. Then for any $C^*  \in \mathcal{C}_k^*$,
    \begin{align*}
        \sqrt{n}\left\{ \widehat{R}(C^*) - R(C^*) \right\} \leadsto N\left(0, \var\left(\varphi_{C^*} \right) \right),
    \end{align*}
    where $\varphi_{C}$ is defined in \eqref{eqn:eif-of-risk}.
\end{lemma}
Under the similar conditions as Theorem \ref{thm:plug-in-consistency}, we can show the proposed codebook estimator \eqref{eqn:eif-codebook-estimator} is consistent, as stated in the following corollary.

\begin{corollary} \label{cor:consistency-of-Chat-eif-estimator}
    If Assumptions \ref{assumption:A1-boundedness}, \ref{assumption:A3-uniqueness-of-codebook}, \ref{assumption:A4-pihat-boundedness}, and \ref{assumption:A5-np-consistency-condition} hold, then $\widehat{C}$ computed by the semiparametric estimator \eqref{eqn:eif-codebook-estimator} converges in probability to $C^*$.
\end{corollary}

We now turn to the asymptotic properties of $\widehat{C}$, focusing in particular on conditions that ensure $\sqrt{n}$-consistency and asymptotic normality in large nonparametric models. To this end, let $\upvarphi_1(z;\eta) = [\varphi_{1,1}(z;\eta_1), \ldots, \varphi_{1,p}(z;\eta_p)]^\top$ where each $\varphi_{1,a}$ is defined in \eqref{eqn:notation-IFs}. With a slight abuse of notation, as was done in \citet{bottou1994convergence}, we define the derivative of $\varphi_C$ at any $C^\prime \in \mathcal{C}_k$ for some fixed $\bar{\eta}$ by
\begin{align} \label{eqn:varphi-derivative}
    \upvarphi_{C^\prime}(Z;\bar{\eta}) & \equiv \frac{\partial}{\partial C} \varphi_C(Z;\bar{\eta}) \Big\vert_{C=C^\prime}\\
    &= 2\left[(c^\prime_1 - \upvarphi_1(Z;\bar{\eta})) \mathbbm{1}\{1=d(\bar{\mu},C^\prime)\}, \ldots, (c^\prime_k - \upvarphi_1(Z;\bar{\eta})) \mathbbm{1}\{k=d(\bar{\mu},C^\prime)\} \right]^\top \nonumber
\end{align}
where $d(\bar{\mu},C^\prime) = \underset{j\in\{1,\ldots,k\}}{\argmin} \Vert c^\prime_j - \bar{\mu} \Vert_2^2$, i.e., the subscript for the nearest center to a given $\bar{\mu}$. Similarly, one may compute the derivative matrix of $\Pb\left\{ \upvarphi_{C}(Z;\bar{\eta}) \right\}$ at $C^\prime$:
\begin{align} \label{eqn:varphi-Hessian}
    M(C^\prime, \bar{\eta}) \equiv \frac{\partial}{\partial C}  \Pb\left\{ \upvarphi_{C}(Z;\bar{\eta}) \right\} \Big\vert_{C=C^\prime} = 2\text{diag}\left(\bm{1}_{(p)} p_1(\bar{\eta},C^\prime), \ldots,  \bm{1}_{(p)} p_k(\bar{\eta},C^\prime) \right),
\end{align}
where $p_j(\bar{\eta},C^\prime) = \Pb\{j=d(\bar{\mu},C^\prime)\}$ and $\bm{1}_{(p)}$ is a $p$-dimensional vector of all ones.

Then, up to an \(o_{\Pb}(1/\sqrt{n})\) error, the solutions to the minimization problem in \eqref{eqn:eif-codebook-estimator} can be equivalently characterized as solutions to the following empirical moment condition:
\begin{align*}
    \Pn\left\{\upvarphi_{\widehat{C}}(Z;\widehat{\eta}_{-K}) \right\} = o_\Pb\left( \frac{1}{\sqrt{n}} \right).
\end{align*}

Analyzing the optimal solution is more delicate than analyzing the value, because the codebook depends on the data through Voronoi assignments, which can change discontinuously under small perturbations near cluster boundaries. In the classical k-means setting, asymptotic normality of the estimated codebook was first established by \citet{pollard1982central}. Extending that result to our causal clustering setting is substantially more challenging, since the risk functional \(\E\{\varphi_C(Z)\}\) depends nontrivially on infinite dimensional nuisance parameters. To handle these additional sources of instability, we first introduce the following technical assumption.
\begin{assumptionp}{A7}
\label{assumption:A7-hard-fitted-margin}
There exists a sequence $\rho_n > 0$ with $\rho_n^{-1}=o(n)$ such that, with probability tending to one, the following holds for each fold $b \in \{1,\ldots,K\}$: for every codebook $C$ on the line segment between $\widehat C$ and any leave-one-out perturbation of $\widehat C$ obtained by replacing a single observation in fold $b$, and for every $i$ with $B_i=b$,
\[
\operatorname{dist}\!\bigl(\widehat\mu_{-b}(X_i),\partial C\bigr)\ge \rho_n.
\]
Here $\partial C$ denotes the union of the Voronoi boundaries induced by $C$.
\end{assumptionp}

Assumption~\ref{assumption:A7-hard-fitted-margin} is a sample level separation condition requiring the fitted points $\widehat\mu_{-b}(X_i)$ to remain uniformly away from the Voronoi boundaries of codebooks on the segment between $\widehat C$ and its leave one out perturbations, by more than the perturbation scale $n^{-1}$. It is stronger than the population margin condition in Definition~\ref{def:margin-condition}, but more directly aligned with the local geometric stability needed in our proof. Similar separation conditions have appeared in clustering theory \citep{levrard2015nonasymptotic,le2018notion}. Although stronger, the assumption is still plausible in settings with clear cluster separation, where fitted points are unlikely to lie near Voronoi boundaries.

A simple example where Assumption~\ref{assumption:A7-hard-fitted-margin} holds is a well separated mixture type setting in which the fitted points cluster around distinct centers and remain away from the induced Voronoi boundaries with high probability. This is especially natural in designs with a hard geometric margin, where the fitted features are confined to separated regions of the feature space. More generally, the assumption still allows the fitted margin $\rho_n$ to shrink with $n$, provided it shrinks more slowly than the perturbation scale, that is, $\rho_n^{-1}=o(n)$.

We next introduce the following additional technical assumption.
\begin{assumptionp}{A8}
\label{assumption:A8-fitted-boundary-margin}
There exist constants $\kappa_{\mathrm{fit}} > 0$, $\beta > 0$, and $L < \infty$ such that, for each fold $b \in \{1,\ldots,K\}$, there is a random neighborhood $\mathcal N_{n,b}$ of $C^*$ satisfying the following with probability tending to one:

(i) $\mathcal N_{n,b}$ contains $\widehat C$ and all leave-one-out perturbations of $\widehat C$ arising from replacing a single observation in fold $b$; and

(ii) for every $0 < t \le \kappa_{\mathrm{fit}}$,
\[
\sup_{C \in \mathcal N_{n,b}}
\Pb\!\left\{
\operatorname{dist}\!\bigl(\widehat\mu_{-b}(X),\partial C\bigr)\le t
\,\middle|\,
\{Z_i:B_i\neq b\}
\right\}
\le
L t^\beta.
\]
Here $\partial C$ denotes the union of the Voronoi boundaries induced by $C$.
\end{assumptionp}

Assumption~\ref{assumption:A8-fitted-boundary-margin} requires that the fitted feature map $\widehat\mu_{-b}(X)$ place only limited probability mass near the Voronoi boundaries of codebooks close to $C^*$. In this sense, it is a fitted analogue of the population margin condition in Definition~\ref{def:margin-condition}, and is closely related in spirit to margin conditions used in classification theory \citep[e.g.,][]{audibert2007fast}. Condition~(i) is only a mild localization requirement, ensuring that $\widehat C$ and its leave-one-out perturbations remain in a neighborhood of $C^*$ with high probability.

Assumption~\ref{assumption:A8-fitted-boundary-margin} is satisfied, for example, if conditional on the training sample for fold $b$, the distribution of $\widehat\mu_{-b}(X)$ has a density that is uniformly bounded near the relevant Voronoi boundaries. In that case, the probability mass of a boundary strip of width $t$ is of order $t$, so condition~(ii) holds with $\beta=1$. A stronger case is a hard-margin regime, where $\widehat\mu_{-b}(X)$ remains a fixed positive distance away from those boundaries with high probability, in which case condition (ii) holds for any $\beta>0$.

In the next theorem, our first main result of this section, we compute an asymptotic bound for the excess risk, as well as the rate of convergence for $\widehat{C}$. The main technical challenge stems from the non-trivial coupling between the clustering procedure and the nuisance estimators, induced by the cross-fitting scheme.

\begin{theorem}
\label{thm:Chat-root-nCAN}
Suppose that $\Pb$ satisfies the population margin condition with some $\kappa>0$ and $\alpha>0$, and that Assumptions \ref{assumption:A1-boundedness}, \ref{assumption:A3-uniqueness-of-codebook}, \ref{assumption:A4-pihat-boundedness}, \ref{assumption:A5-np-consistency-condition}, \ref{assumption:A7-hard-fitted-margin}, and \ref{assumption:A8-fitted-boundary-margin} hold. Also assume that $\E\|\widehat\mu(X)\|_2^2<\infty$. Then, if $p_j(\eta,C^*)>0$ for all $j$,
\begin{align*}
\|\widehat C-C^*\|_1
=
O_\Pb\!\left(n^{-\min\{\beta,1\}/2}+R_{2,n}\right)
\, 
\text{ and } \,
R(\widehat C)-R(C^*)
=
o_\Pb\!\left(n^{-\min\{\beta,1\}/2}+R_{2,n}\right).
\end{align*}
In particular, if $\beta\ge 1$, then
\begin{align*}
\|\widehat C-C^*\|_1
=
O_\Pb\!\left(n^{-1/2}+R_{2,n}\right)
\, 
\text{ and } \,
R(\widehat C)-R(C^*)
=
o_\Pb\!\left(n^{-1/2}+R_{2,n}\right).
\end{align*}
\end{theorem}

Theorem \ref{thm:Chat-root-nCAN} shows that the proposed codebook estimator $\widehat{C}$ and the associated excess risk may attain substantially faster rates than its nuisance estimators $\widehat{\eta}$. Specifically if $R_{2,n} = O_\Pb\big(n^{-1/2}\big)$ (weaker than Assumption \ref{assumption:A6-np-remainder-condition}) and $\alpha \geq 1$, we can attain $\sqrt{n}$ rates for $\widehat{C}$ and faster-than-$\sqrt{n}$ rates for excess risk by virtue of the fact that $R_{2,n}$ involves products of nuisance estimation errors. \citet{levrard2015nonasymptotic, levrard2018quantization} provided some instances of the natural classifiers corresponding to the margin exponent $\alpha = 1$. Note that the condition $p_j(\eta,C^*) >0$ is equivalent to $\Pb\{V_i(C^*)\} > 0$, ensuring that there are no vacant Voronoi cells. This condition also guarantees that the derivative matrix $M(C^*, \eta)$ is nonsingular. 

The following corollary provides the second main result of this section by establishing conditions under which the codebook estimator $\widehat{C}$ is $\sqrt{n}$-consistent and asymptotically normal, building upon Theorem~\ref{thm:Chat-root-nCAN}.

\begin{corollary}
\label{thm:clt-new}
Suppose that $\Pb$ satisfies the population margin condition with some $\kappa > 0$ and $\alpha > 0$, and that Assumptions \ref{assumption:A1-boundedness}, \ref{assumption:A3-uniqueness-of-codebook}, \ref{assumption:A4-pihat-boundedness}, \ref{assumption:A5-np-consistency-condition}, \ref{assumption:A7-hard-fitted-margin}, and \ref{assumption:A8-fitted-boundary-margin} hold. Also assume that $\E\|\widehat\mu(X)\|_2^2<\infty$. Then, if $p_j(\eta,C^*) >0$ for all $j$ and $\beta \ge 1$,
\begin{align*}
\widehat{C} - C^*
=
-M(C^*, \eta)^{-1}(\Pn - \Pb)\left\{ \upvarphi_{C^*}(Z;\eta) \right\}
+
O_\Pb\!\left(R_{2,n} + o_\Pb\!\left( \frac{1}{\sqrt{n}} \right)\right).
\end{align*}
\end{corollary}

Corollary \ref{thm:clt-new} follows when the fitted boundary mass near the relevant Voronoi boundaries decays at least linearly, that is, when $\beta \ge 1$, while the original population margin condition enters only through the bias remainder term $R_{2,n}$. It implies that $\widehat C$ is not only $\sqrt n$ consistent but also asymptotically normal under Assumption \ref{assumption:A6-np-remainder-condition}, which may still hold for generic, flexibly estimated nuisances. This in turn enables inference for $C^*$, for example via asymptotically valid bootstrap confidence intervals under mild additional regularity conditions.

\begin{remark}\label{rmk:empty-margin-donsker}
An alternative route to Corollary~\ref{thm:clt-new}, mathematically cleaner but more restrictive, is to replace Assumptions~\ref{assumption:A7-hard-fitted-margin} and~\ref{assumption:A8-fitted-boundary-margin} by an \emph{empty-margin} condition, namely that for some $\kappa>0$, $\Pb\{\mu \in N_{C^*}(\kappa)\}=0$. This condition rules out any population mass in a neighborhood of the relevant Voronoi boundaries, so the Voronoi labels are locally stable around $C^*$ and the class $\{\upvarphi_C(\cdot;\bar\eta): C \in \mathcal C_k \text{ near } C^*\}$ behaves as a piecewise smooth parametric class for any fixed $\bar\eta$, yielding a Donsker type empirical process bound; see Remark~\ref{app:rmk:empty-margin-donsker} in Web Appendix~\ref{appendix-proofs}. We do not pursue this route here, since the empty-margin condition is quite stringent and offers little practical advantage over our separate stability and boundary-mass assumptions in realistic settings.\\
\end{remark}

\section{Illustration}

\subsection{Simulation study}

We assess finite–sample performance in a simplified six-cluster design. Let \(C^{*}=\{c_1,\dots,c_6\}\subset[-6,6]^2\) be the vertices of a wide regular hexagon; the minimum pairwise separation is \(6\), so the distance from any center to its closest Voronoi boundary is \(3\). We draw covariates \(X=(X_1,\ldots,X_6)\stackrel{\text{i.i.d.}}{\sim}\mathrm{Unif}[-1,1]\) and define the sector index \(S\in\{1,\ldots,6\}\) by the polar angle \(\theta=\operatorname{atan2}(X_2,X_1)\) using breakpoints \(-\pi,-\tfrac{2\pi}{3},-\tfrac{\pi}{3},0,\tfrac{\pi}{3},\tfrac{2\pi}{3},\pi\). To keep points tightly concentrated around their sector center, introduce small `jitters' \(j_0(X)=\delta\,[\sin(\pi X_3)+0.5\cos(\pi X_4)]\) and \(j_1(X)=\delta\,[\cos(\pi X_5)+0.5\sin(\pi X_6)]\) with \(\delta=0.01\), and set \(\mu_0(X)=c_{S,1}+j_0(X)\) and \(\mu_1(X)=c_{S,2}+j_1(X)\). Thus \(\mu(X)=(\mu_0(X),\mu_1(X))\) lies within radius \(\delta\) of \(c_S\), yielding a hard margin of at least \(3-\delta\). In particular, this design automatically satisfies the margin condition in Definition \ref{def:margin-condition}.

The treatment mechanism is \(\pi_1(X)=\mathrm{expit}(-0.2+0.6X_1-0.25X_2+0.2X_3)\in[0.05,0.95]\); we draw \(A\mid X\sim\mathrm{Bernoulli}(\pi_1(X))\) and generate outcomes \(Y=\mu_A(X)+\varepsilon\) with \(\varepsilon\sim\mathcal N(0,\sigma^2)\) and \(\sigma=0.15\). Under this design, the term $\Vert \widehat{\mu}_a - \mu_a \Vert_{\infty}^{\alpha+1}$ is negligible. Because \(\Pi_{C^{*}}(\mu(X))=c_S\) almost surely, the oracle risk equals \(R(C^{*})=\mathbb{E}[\,\|\mu(X)-\Pi_{C^{*}}(\mu(X))\|_2^2\,]=\mathbb{E}[\,j_0(X)^2+j_1(X)^2\,]=\delta^{2}\!\left(\tfrac{1}{2}+\tfrac{1}{8}+\tfrac{1}{2}+\tfrac{1}{8}\right)=\tfrac{5}{4}\,\delta^{2}\).

The nuisance functions are estimated using 5-fold cross-fitting: in each fold, we fit deliberately misspecified outcome regressions $\hat{\mu}_a(x)$ using simple linear models based only on $(X_1,X_2)$, while the propensity score $\hat{\pi}_a(x)$ is estimated via a correctly specified logistic regression. This configuration preserves the validity of the EIF scores despite outcome-model misspecification. To estimate $\hat C$, we apply the same gradient–descent algorithm for both the semiparametric and plug-in approaches.

The results in Figure~\ref{fig:excess-risk-codebook-error} closely align with the theoretical findings established in Sections~\ref{sec:plug-in-estimator} and~\ref{sec:efficient-causal-clustering}. 
The semiparametric estimator (Semi/EIF) attains uniformly lower excess risk than the plug–in approach, where the noticeably steeper slope on the log–log scale indicates a faster convergence rate. 
The semiparametric estimator consistently achieves more accurate recovery of the true codebook as well, with a pronounced reduction in the per–center \(L_{1}\) error as \(n\) increases.

\begin{figure}[!ht]
\centering
\begin{minipage}[t]{0.49\textwidth}
    \centering
    \includegraphics[width=\linewidth]{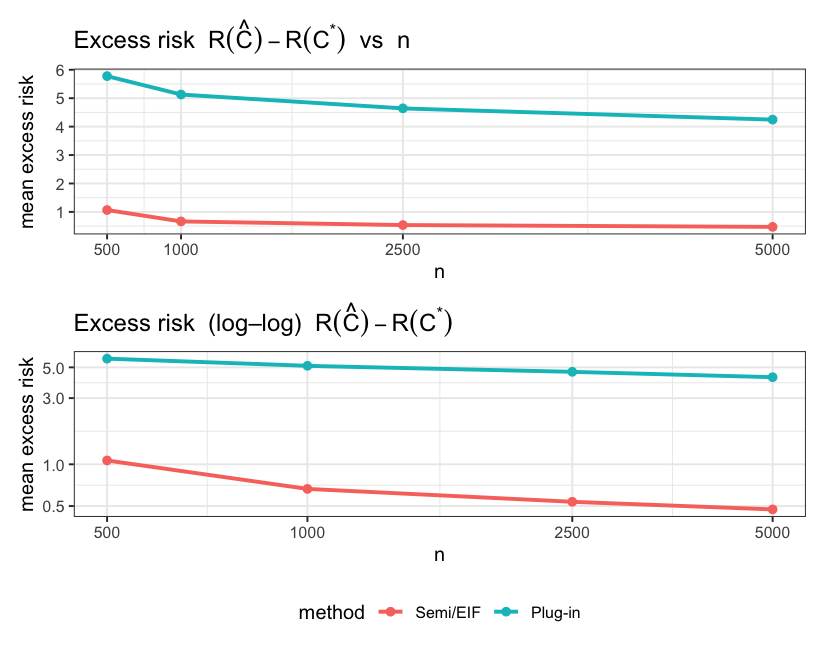}    
\end{minipage}\hfill%
\begin{minipage}[t]{0.49\textwidth}
    \centering
    \includegraphics[width=\linewidth]{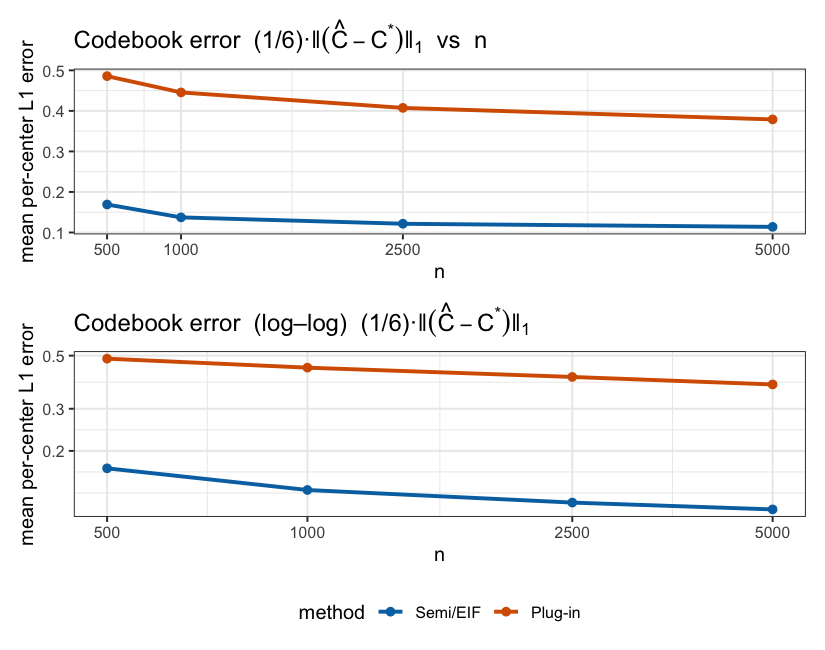}    
\end{minipage}
\caption{(Left) Excess risk $R(\hat C)-R(C^\star)$ across sample sizes. (Right) Codebook error $(1/6)\,\|\hat C - C^\star\|_{1}$ across sample sizes.}
\label{fig:excess-risk-codebook-error}
\end{figure}


\subsection{Case Study: PROPEL Chronic Low Back Pain Trial}

We illustrate the causal clustering approach using data from the Problem-Solving Pain to Enhance Living Well (PROPEL) clinical trial, which was designed to evaluate the effectiveness of mobile-supported self-management interventions for patients with chronic low back pain \citep{hong2022feasibility, kim2022propelcris}. The original dataset consists of linked adult medical records and survey responses from participants assigned to one of three 12 week active intervention arms: (1) weekly yoga sessions, (2) weekly yoga plus nurse led self management counseling, or (3) weekly yoga plus counseling supplemented with digital self management support through the PROPEL mobile application. Our analysis focuses on revealing treatment effect heterogeneity by examining how these interventions differentially affect end-of-study pain intensity across individuals. We define treatment by exposure duration, more than 8 weeks versus fewer than 4 weeks, thereby collapsing the data into a single timepoint setting and ignoring time varying effects. We use ten pre treatment covariates spanning demographic, behavioral, physiological, and genetic factors. For methodological illustration, we estimate the joint distribution of the observed data and resample from it to construct a balanced analysis set with about 200 individuals per treatment arm.

We implement the proposed semiparametric estimator with $K=2$ sample splits. 
Nuisance functions are estimated using a cross-validated Super Learner ensemble 
\citep{van2007super}, combining regression splines, support vector machine 
regression, and random forests. The elbow method suggests that $k=6$ is a 
reasonable choice for the number of clusters. Figure~\ref{fig:application}(a) 
plots the six clusters in the counterfactual mean space, revealing clear 
patterns of treatment-effect heterogeneity. Figure~\ref{fig:application}(b) 
shows density plots of the pairwise CATE estimates $\widehat{\tau}_{2,1}$ 
and $\widehat{\tau}_{3,1}$ across clusters, and Figure~\ref{fig:application}(c) 
presents a heatmap of standardized cluster-level baseline covariate means. 
Together, these visualizations illustrate how units in each cluster exhibit distinct responses to the active regimens, while simultaneously revealing the characteristic covariate profiles that represent each cluster.

For example, clusters C5 and C2, characterized by older individuals with relatively low BMI and low smoking prevalence, exhibit positive CATEs for both $\widehat{\tau}_{2,1}$ and $\widehat{\tau}_{3,1}$. This suggests that these groups benefit most from both active interventions, especially treatment~(3), which additionally includes counseling and mobile application support. In contrast, cluster C6, characterized by high BMI 
and the highest smoking rates, shows strongly negative pairwise CATE values, 
suggesting that metabolically and behaviorally high-risk patients may not 
tolerate or respond well to either active treatment. This case study illustrates the value of the proposed framework for systematically characterizing treatment effect heterogeneity. Further findings and supporting analyses, including detailed cluster-level effect heterogeneity and covariate profile summaries, are provided in Web Appendix~\ref{sec:app-case-study}.
\\

\begin{figure}[t!]
\centering
\subfigure[]{\includegraphics[width=0.32\textwidth]{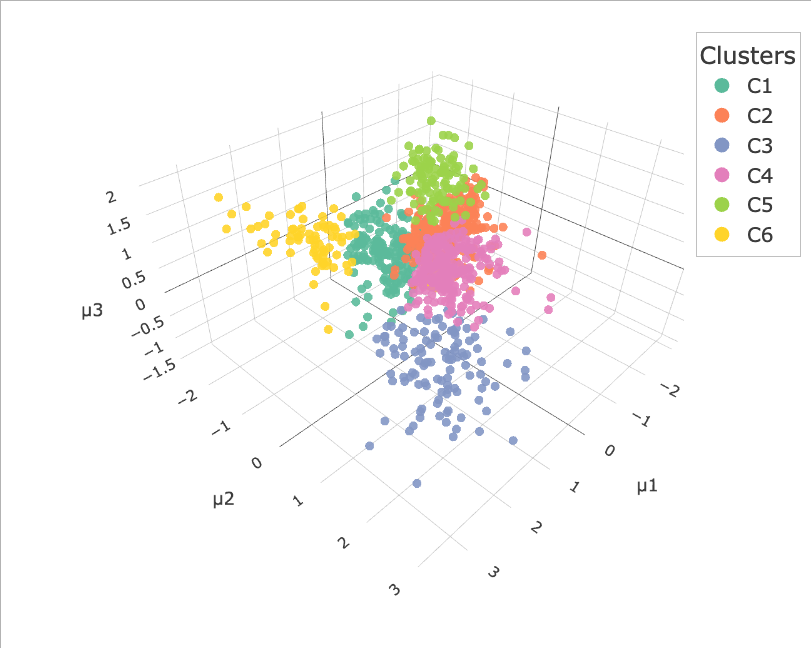}}
\hfill
\subfigure[]{\includegraphics[width=0.321\textwidth]{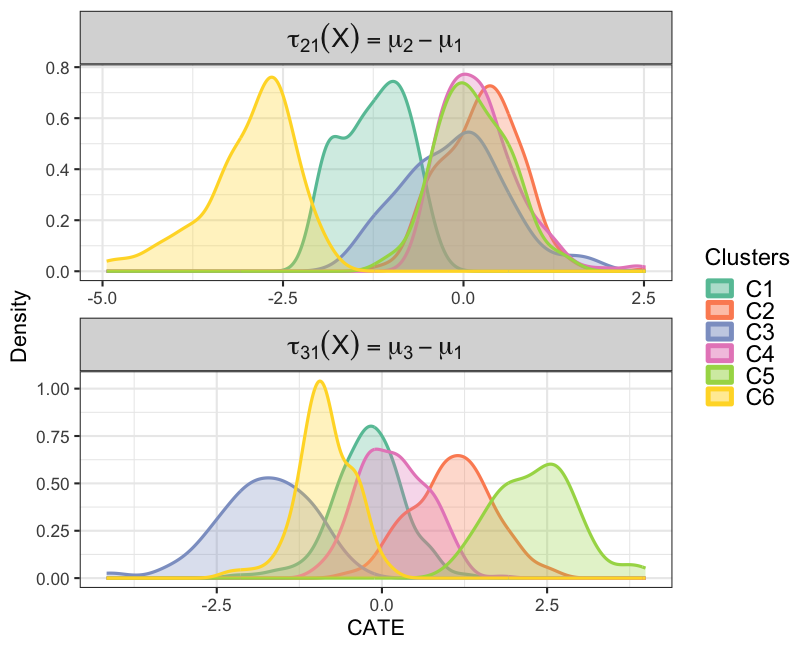}}
\hfill
\subfigure[]{\includegraphics[width=0.335\textwidth]{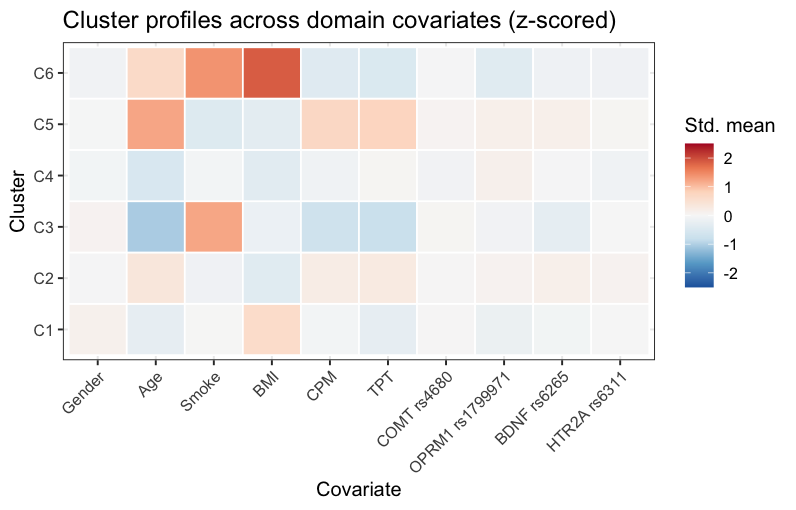}}
\caption{Visualization of estimated causal clusters. 
Panel~(a) displays the six clusters C1–C6 in the three–dimensional counterfactual mean space $(\mu_1,\mu_2,\mu_3)$. 
Panel~(b) shows kernel density estimates of the pairwise CATEs $\widehat{\tau}_{21}(X)=\widehat{\mu}_2(X)-\widehat{\mu}_1(X)$ (top) and $\widehat{\tau}_{31}(X)=\widehat{\mu}_3(X)-\widehat{\mu}_1(X)$ (bottom) for each cluster, illustrating systematic differences in treatment contrasts. 
Panel~(c) presents a heatmap of standardized cluster–level means for the baseline covariates, summarizing how each causal cluster aligns with background risk factors.}
\label{fig:application}
\end{figure}

\section{Discussion}
\label{sec:discussion}

In this paper, we propose a new framework for analyzing treatment effect heterogeneity by leveraging tools in cluster analysis. We provide flexible nonparametric estimators for a wide class of models. The proposed methods are easily implemented with off-the-shelf algorithms, and enable the discovery of subgroup structures in studies with multiple treatments or outcomes. In particular, the plug in estimator extends naturally to clustering problems based on general regression functions, whereas extension of the semiparametric estimator requires additional problem specific efficiency analysis. More broadly, the framework accommodates generic pseudo-outcomes, including settings with partially observed outcomes or unknown counterfactual functions.

Our findings open up a plethora of intriguing opportunities for future work. In an upcoming companion paper, we consider kernel-based undersmoothing approaches for causal k-means clustering, which do not require the margin condition. Much more work is required to expand causal clustering to other widely-used clustering algorithms, such as density-based clustering and hierarchical clustering; each rests on different assumptions and requires its own analysis. Another important direction is to integrate this framework with prescriptive methods by estimating soft cluster memberships, e.g., estimates of $\Pb(S_j = 1 \mid X)$, where $S_j = \mathbbm{1}(X \in R_j)$, via kernel or mixture-based density estimation. This enables assigning new individuals to clusters probabilistically based on their observed covariates and, in turn, recommending the treatment associated with the most likely cluster, thereby facilitating the construction of individualized optimal treatment regimes. Other settings involving, for example, time-varying treatments, instrumental variables, or mediation would also be promising directions for future research.
\\



\section*{Acknowledgements}
\label{sec:acknowledgements}
This work was supported by the National Research Foundation grant funded by the Korean government (MSIT) (Nos. RS-2022-NR068754, RS-2024-00335008, and RS-2025-24534596), and by the Samsung Science and Technology Foundation under Project Number SSTF-BA2502-01. The work was also supported by the National Library of Medicine, \#1R01LM013361-01A1 and NSF CAREER Award 2047444.\\

\section*{Data and code availability}
R source code is publicly available at \url{https://github.com/kwangho-joshua-kim/causal-k-means} and reproduces all simulation results and the case study analysis. Because the PROPEL dataset cannot be shared, we provide a fully synthetic dataset designed to closely mirror the original data structure.

\pagebreak

\bibliographystyle{agsm}
\bibliography{reference}


\pagebreak
\setcounter{page}{1}
\appendix
\newpage
\begin{center}
    {\LARGE Web Appendix\par}
    \vspace{0.4em}
    {\large for\par}
    \vspace{0.4em}
    {\Large Causal K-Means Clustering\par}
    \vspace{0.5em}
    {\normalsize Kwangho Kim, Jisu Kim, and Edward H Kennedy\par}
    \vspace{0.8em}
    \rule{0.6\textwidth}{0.4pt}
\end{center}
\vspace{1em}

\setcounter{equation}{0}
\renewcommand{\theequation}{A.\arabic{equation}}
\setcounter{figure}{0}
\renewcommand{\thefigure}{A.\arabic{figure}}
\setcounter{theorem}{0}
\renewcommand{\thetheorem}{A.\arabic{theorem}}
\setcounter{remark}{0}
\renewcommand{\theremark}{A.\arabic{remark}}

\section{Expanded Case Study Findings} \label{sec:app-case-study}

Because the original PROPEL data cannot be shared for privacy reasons, we work with a semi synthetic dataset constructed to closely mirror its structure. Although the original dataset contains hundreds of recorded covariates, we focus here on 10 baseline variables that capture representative demographic, behavioral, physiological, and genetic characteristics. We estimate their joint distribution using kernel density estimation and then generate synthetic covariate vectors by resampling from the fitted distribution. The remaining variables are subsequently generated from these synthetic covariates so that the resulting dataset preserves the main dependence patterns and scale of the original study while remaining fully artificial at the individual level. This study was conducted in accordance with Korea University Institutional Review Board requirements\footnote{https://irb.korea.ac.kr/}. In particular, only data collected under approved consent and protocol conditions were included in the analysis, and genetic information from participants who did not agree to its reuse was excluded prior to analysis and data generation.

\textbf{Choosing $k$.} To determine the number of causal clusters, we evaluated the geometry of the estimated counterfactual mean vectors $\hat{\mu} = (\hat{\mu}_1,\hat{\mu}_2,\hat{\mu}_3)$ and the estimated centers $\widehat{C}$, using the total within-cluster sum of squares (WCSS) and the incremental gains in fit as $k$ increases. As shown in Figure~\ref{fig:elbow-gain}, the WCSS curve computed from the $\hat{\mu}$ features exhibits a clear elbow at $k=6$, and the relative-gain plot indicates that the improvement from $k=5$ to $k=6$ is the final substantial increase before additional clusters provide only marginal benefits. Taken together, these diagnostics suggest that $k=6$ provides a parsimonious yet sufficiently expressive representation of the heterogeneity in the estimated counterfactual response surfaces. Nonetheless, we acknowledge that this choice of $k$ is based on empirical heuristics rather than formal theory, and that principled methods for selecting the number of causal clusters remain an open direction for future research.

\begin{figure}[h!]
\centering
\subfigure[]{\includegraphics[width=0.48\textwidth]{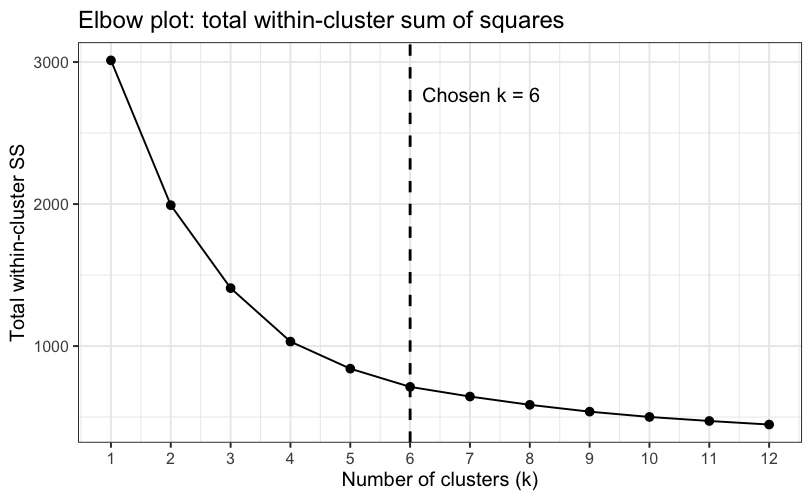}}
\hfill
\subfigure[]{\includegraphics[width=0.48\textwidth]{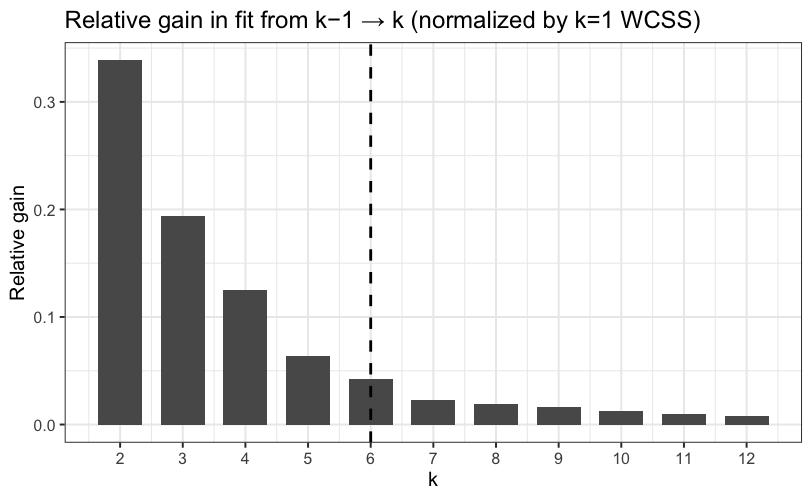}}
\caption{
(a) Elbow plot of total within-cluster sum of squares (WCSS) computed from the estimated counterfactual mean vectors $\hat{\mu}$, showing a pronounced flattening after $k=6$. 
(b) Relative gain in fit from $k{-}1$ to $k$, normalized by the $k=1$ WCSS, indicating that the improvement at $k=6$ is the last meaningful increase prior to diminishing returns. 
Both diagnostics jointly support selecting $k=6$ as an appropriate number of causal clusters.}
\label{fig:elbow-gain}
\end{figure}

\textbf{Cluster‐level effect heterogeneity and covariate profiles.}
In this analysis, we select ten baseline covariates spanning demographic, behavioral, physiological, and genetic domains. Demographic and behavioral factors include gender, age, smoking status, and BMI. Physiological measures consist of Conditioned Pain Modulation (CPM), a summary of endogenous pain-inhibition capacity (higher values indicate more effective modulation), and Thermal Pain Thresholds (TPT), the cold temperature at which pain is first perceived. Genetic markers include COMT rs4680, OPRM1 rs1799971, BDNF rs6265, and HTR2A rs6311, where greater mutant-allele load is generally associated with reduced neurotransmitter or receptor function and heightened pain sensitivity.

Figure~\ref{fig:application}(a) displays the six clusters obtained by applying
$k$-means to the estimated counterfactual mean vectors
$(\hat\mu_1,\hat\mu_2,\hat\mu_3)$, and Figure~\ref{fig:application}(b) presents
the empirical densities of the pairwise CATEs
$\widehat{\tau}_{2,1}(X)=\hat\mu_2-\hat\mu_1$ and
$\widehat{\tau}_{3,1}(X)=\hat\mu_3-\hat\mu_1$ within each cluster.  
Taken together, these plots reveal pronounced and interpretable patterns of
treatment-effect heterogeneity across the PROPEL population.

\begin{figure}[h!]
\centering
\subfigure[]{\includegraphics[width=0.48\textwidth]{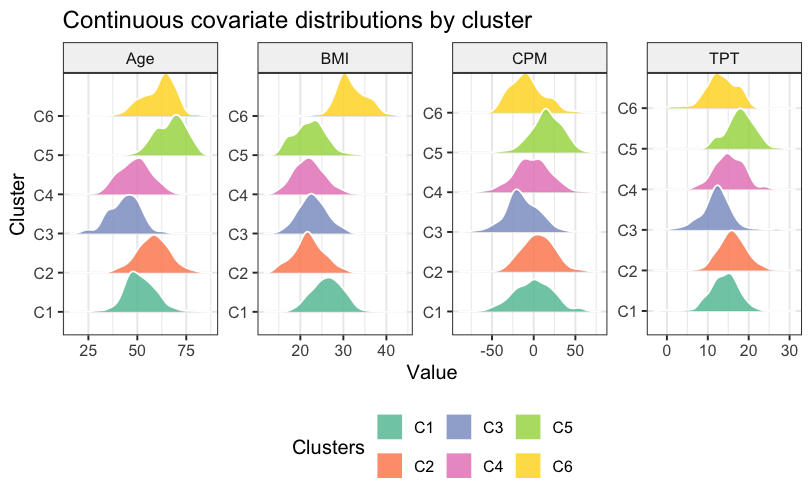}}
\hfill
\subfigure[]{\includegraphics[width=0.48\textwidth]{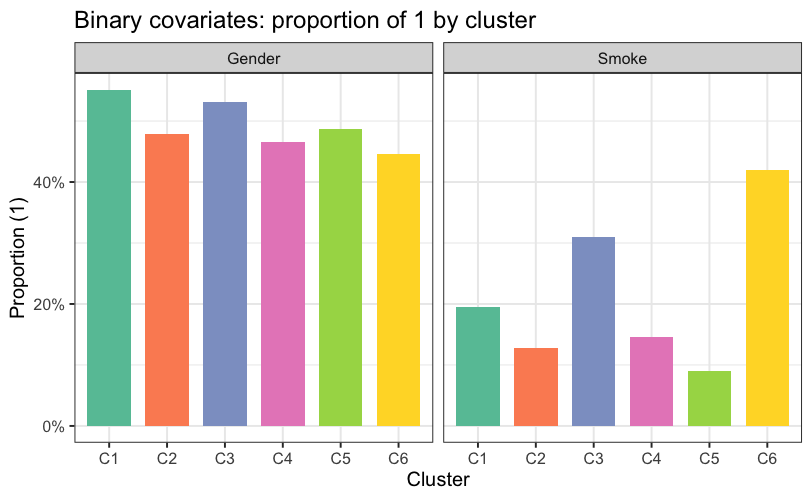}}
\caption{(a) Ridgeline plots of continuous covariates by cluster.
(b) Proportions of binary covariates (Gender and Smoke) across clusters.}
\label{fig:ridge-bar}
\end{figure}

\begin{figure}[h!]
\centering
\subfigure[]{\includegraphics[width=0.48\textwidth]{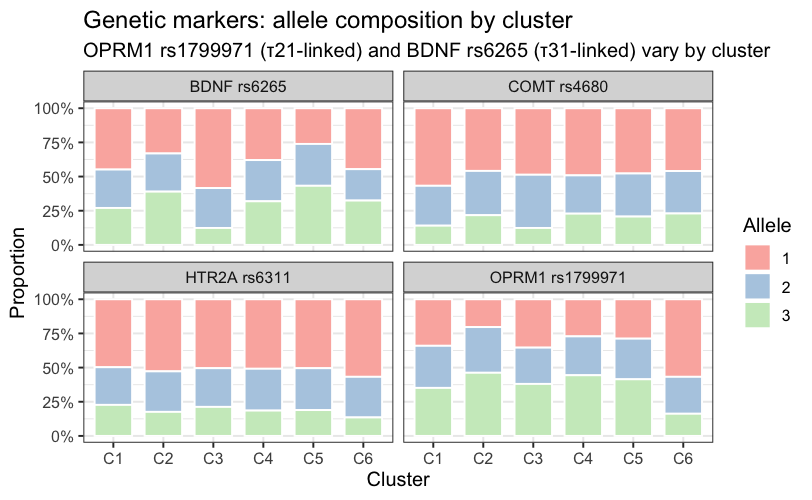}}
\hfill
\subfigure[]{\includegraphics[width=0.48\textwidth]{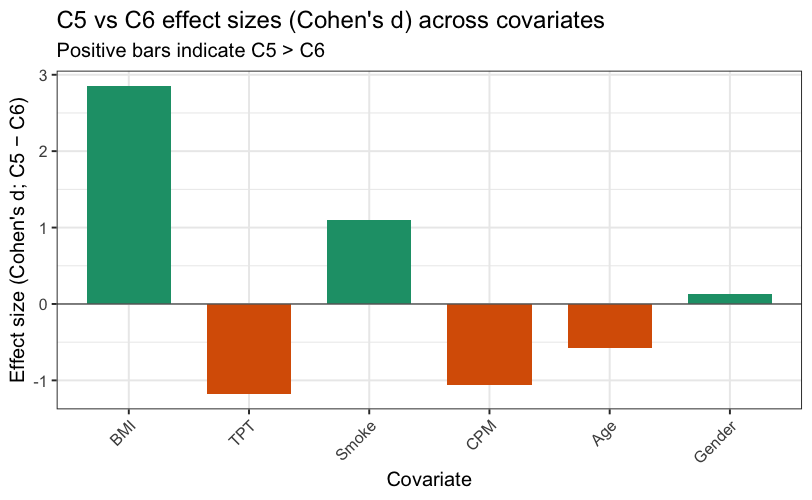}}
\caption{(a) Cluster-specific allele compositions for four genetic markers.
(b) Standardized effect sizes (Cohen's $d$) comparing clusters C5 and C6.}
\label{fig:snp-effectsize}
\end{figure}

\emph{Benefit clusters (C2 and C5).}
Clusters C2 and~C5 consistently exhibit positive treatment effects.  
C2 shows moderately positive effects for both treatments~2 and~3 versus~1,
while C5 stands out as the group with the largest positive
$\widehat{\tau}_{3,1}(X)$ among all clusters.  
These favorable responses align closely with their baseline profiles.
Figure~\ref{fig:ridge-bar}(a) shows that both clusters contain some of the
oldest participants, and Figure~\ref{fig:ridge-bar}(b) confirms that they also
exhibit among the lowest smoking prevalence.  
Cluster~C5 additionally displays lower BMI and more favorable pain sensitivity
patterns (higher CPM and TPT), indicating a physiologically resilient
subpopulation.  
The SNP distributions in Figure~\ref{fig:snp-effectsize}(a) further reveal that
C5 is enriched with potentially advantageous alleles (e.g., certain OPRM1 and
BDNF variants), providing additional biological plausibility for their more
favorable response patterns.

\emph{Harm or nonresponse clusters (C6 and C3).}
Cluster~C6 is a clear ``harm'' or ``nonresponder'' group, with both
$\widehat{\tau}_{2,1}$ and $\widehat{\tau}_{3,1}$ densities shifted far to the
left in Figure~\ref{fig:application}(b).  
This aligns with its adverse baseline profile: C6 contains the youngest
participants, has the highest smoking rate (Figure~\ref{fig:ridge-bar}(b)),
and also shows the highest BMI and the lowest CPM values, reflecting a
higher metabolic and inflammatory burden.  
Cluster~C3 shows somewhat less extreme but still unfavorable patterns,
including younger age and elevated smoking prevalence, which correspond to its
negative effects for treatment~3 versus~1 in particular.

\emph{Near-average clusters (C1 and C4).}
Clusters C1 and C4 exhibit mild or near-zero CATEs for both contrasts.
Their covariate distributions in Figure~\ref{fig:ridge-bar}(a)-(b) show that
they occupy ``middle-of-the-road'' ranges for age, BMI, CPM, and smoking, and
their SNP patterns in Figure~\ref{fig:snp-effectsize}(a) closely resemble the
overall study population.  
These observations are consistent with their more modest causal effects.

To identify which baseline characteristics most strongly distinguish the two clusters with the most extreme treatment responses to the mobile-application intervention, Figure~\ref{fig:snp-effectsize}(b) summarizes standardized effect sizes (Cohen’s $d$) comparing C5 and C6.
Age, BMI, CPM, and smoking status produce the largest contrasts, reinforcing
the interpretation that a combination of demographic, behavioral, and
physiological characteristics drives substantial treatment effect
heterogeneity.

Finally, Figure~\ref{fig:app-benefitproxy} displays the distribution of the constructed effect-modification score
\[
\text{BenefitProxy}
=0.70\!\left(\frac{X_{\mathrm{Age}}-40}{20}\right)^{+}
-0.35\,\text{Smoke}
-0.30\!\left(\frac{X_{\mathrm{BMI}}-27}{8}\right)^{+},
\]
which provides a simple, interpretable summary of baseline factors that qualitatively align with the observed CATE patterns. The structure of the proxy is motivated by the diagnostics: age is centered at 40 and scaled by 20 so that only meaningfully older individuals contribute positively (reflecting the strong age gradient between high- and low-response clusters); smoking receives a moderate penalty consistent with its sharp separation across clusters; and BMI enters through a positive-part transformation centered at 27 to capture the adverse effect of elevated adiposity. The weights $(0.70,\, -0.35,\, -0.30)$ are chosen to roughly reflect the relative magnitudes seen in the empirical covariate contrasts. As shown in Figure~\ref{fig:app-benefitproxy}, the resulting score ranks clusters in a manner consistent with the CATE estimates, linking observed characteristics to the discovered heterogeneity structure. The BenefitProxy distributions reveal
substantial differences in within-cluster heterogeneity: high-benefit groups
(C2, C5) exhibit tightly concentrated scores, whereas lower-benefit groups
(C3, C6) display broader, more dispersed distributions.  
This separation highlights the stability of the discovered clusters and shows
that the risk profiles underlying treatment benefit are not only distinct
across clusters but also internally coherent.

\begin{figure}[h!]
\centering
\includegraphics[width=0.6\textwidth]{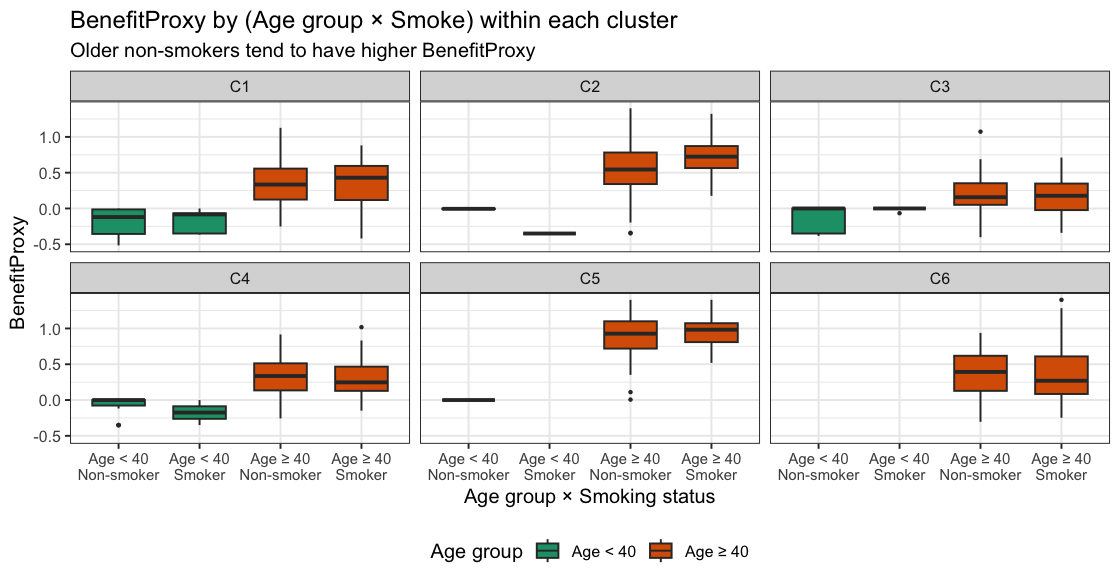}
\caption{Distribution of the constructed \textit{BenefitProxy}
effect-modification score across the six clusters.}
\label{fig:app-benefitproxy}
\end{figure}

\textbf{Limitations and future work Plan.} 
The original dataset includes multiple longitudinal timepoints and exhibits non-compliance with the assigned interventions. It also contains thousands of additional variables, including survey responses, medical history, and extended genetic information, which are not yet accessible due to the lengthy approval process. In a future application paper, we plan to generalize our framework to handle time-varying treatment effects and apply it to this fully expanded dataset once access is granted.

\section{Proofs} \label{appendix-proofs}

\textbf{Notation Guide.}
Hereafter, we let $\Vert f \Vert$ denote the $L_{2}(\Pb)$-norm in order to simplify notation and avoid any confusion with the Euclidean norm $\Vert \cdot \Vert_2$, as the $L_{2}(\Pb)$-norm is used most frequently in the proofs. For simplicity, we drop the dependence on $Z$ if the context is clear. Also, for any fixed $C$, we let $f_{C}(x) = \Vert x-\Pi_{C}(x)\Vert_2^2$ for $x \in \R^p$ so that $R(C) = \E\{f_{C}(\mu)\}$, and let 
\begin{align*}
        f_{c_j}(\mu) &= \Vert \mu - c_j \Vert_2^2, \\
        \varphi_{c_j}(\eta) &= \sum_{a} \left\{\varphi_{2,a}(\eta) - 2\varphi_{1,a}(\eta)c_{ja} + c_{ja}^2 \right\}, \quad j=1,\ldots, k.
\end{align*}
Further, we let $\zeta_j(\mu;C) = \underset{i\neq j}{\min}  \Vert \mu - c_i \Vert_2 - \Vert \mu - c_j \Vert_2$ so that $\Pb\left\{ \vert \zeta_j(\mu;C^*) \vert \leq t \bigm\vert 0 \leq t \leq \kappa \right\} \lesssim t^\alpha$ under the margin condition for any $\alpha >0$, $\kappa >0$. With a slight abuse of notation, we write $\Vert \widehat{C} - C^* \Vert_{1} = \sum_{j=1}^k \Vert \widehat{c}_j - c^*_j \Vert_{1}$. 
\vspace*{.2in}

\subsection{Proof of Theorem \ref{thm:k-means}}
\label{proof:k-means}
Before proceeding to the proof of Theorem~\ref{thm:k-means}, we present a sequence of supporting lemmas. The first lemma refines \citet[][Lemma 1]{kennedy2018sharp}, providing a slightly stronger statement.

\begin{lemma}\label{lem:indicator-bound}
Let the functions $\hat f,\,f$ take any real values. Then
\begin{align*}
\bigl| \mathbbm 1\{\hat f>0\}-\mathbbm 1\{f>0\} \bigr|
        &\;\le\;
          \mathbbm 1\!\bigl(\lvert f\rvert\le \lvert \hat f-f\rvert \,\&\, f>0\bigr)
          +\mathbbm 1\!\bigl(\lvert \hat f\rvert\le \lvert \hat f-f\rvert \,\&\, \hat f>0\bigr).
\end{align*}
\end{lemma}

\begin{proof}
We begin by noting that the left-hand side is either $0$ or $1$, and is nonzero if and only if $\hat{f}$ and $f$ lie on opposite sides of $0$:
\[
\left| \mathbbm{1}\{\hat{f} > 0\} - \mathbbm{1}\{f > 0\} \right| = 1 
\iff 
\operatorname{sign}(\hat{f}) \neq \operatorname{sign}(f).
\]
Thus it suffices to analyze the two cases in which the signs differ.

\noindent{Case 1.} $\hat{f} > 0$ and $f \le 0$. Then $|\hat{f}| = \hat{f}$ and $|f| = -f$. Therefore,
\[
|\hat{f}| + |f| = \hat{f} - f = |\hat{f} - f| \quad \Rightarrow \quad \max\{ |\hat{f}|, |f| \} \leq |\hat{f} - f|.
\]
In particular, $|\hat{f}| \leq |\hat{f} - f|$, and since $\hat{f} > 0$, the second indicator on the RHS evaluates to $1$.

\noindent{Case 2.} $f > 0$ and $\hat{f} \le 0$. Then the same argument holds, and again,
$
\max\{ |\hat{f}|, |f| \} \leq |\hat{f} - f|.
$
In particular, $|f| \leq |\hat{f} - f|$, and since $f > 0$, the first indicator on the RHS evaluates to $1$.

Hence, in either case, one of the two indicators on the right-hand side is $1$, ensuring the inequality holds.
\end{proof}

The following lemma establishes that the projection error arising from perturbations in $\mu$ can be controlled under the margin condition.

\begin{lemma}\label{lem:Pi-bias-alpha} Suppose that Assumption \ref{assumption:A1-boundedness} holds, and $\Pb$ satisfies the margin condition with some $\kappa>0$, $\alpha>0$. Then we have
    \begin{align*}
        \Pb\left\vert \Pi_{C^*,a}(\widehat{\mu}) - \Pi_{C^*,a}(\mu) \right\vert &\lesssim \max_j\left\Vert \mathbbm{1}\left\{\zeta_j(\widehat{\mu};C^*) > 0 \right\} - \mathbbm{1}\left\{\zeta_j(\mu;C^*) > 0 \right\} \right\Vert_{\infty} \sum_a\Vert \widehat{\mu}_a - \mu_a \Vert_{\Pb,1}\\
        & \quad + \sum_a \Vert \widehat{\mu}_a - \mu_a \Vert_{\infty}^\alpha,    
    \end{align*}
where $\Pi_{C^*,a}(\cdot)$ denotes the $a$-th coordinate of $\Pi_{C^*}(\cdot)$.   
\end{lemma}
\begin{proof}      
Recall that $\zeta_j(\bar{\mu};C) = \underset{i \neq j}{\min}  \Vert \bar{\mu} - c_i \Vert_2 - \Vert \bar{\mu} - c_j \Vert_2$ for any $\bar{\mu}$, and that $\left\{ \bar{\mu} \in V_j(C^*)  \Bigm\vert \left\vert \zeta_j(\bar{\mu}, C^*) \right\vert  \leq t \right\} \subseteq N_{C^*}(t)$, $\forall j$.
Letting $\zeta_j \equiv \zeta_j(\mu;C^*)$ and $\widehat{\zeta}_j \equiv \zeta_j(\widehat{\mu};C^*)$, we have
\begin{align*}
    & \Pb\left\vert \Pi_{C^*,a}(\widehat{\mu}) - \Pi_{C^*,a}(\mu) \right\vert\\
    &= \Pb\left( \sum_j c^*_{j,a} \left\vert\mathbbm{1}\left\{\widehat{\zeta}_j > 0 \right\} - \mathbbm{1}\left\{\zeta_j > 0 \right\} \right\vert \right) \\
    & = \sum_j c^*_{j,a}  \Pb \left( \left\vert \mathbbm{1}\left\{\widehat{\zeta}_j > 0 \right\} - \mathbbm{1}\left\{\zeta_j > 0 \right\}\right\vert \left[\mathbbm{1}\left\{\left\vert \widehat{\zeta}_j - \zeta_j \right\vert \leq \kappa \right\} + \mathbbm{1}\left\{\left\vert \widehat{\zeta}_j - \zeta_j \right\vert > \kappa \right\} \right] \right).
\end{align*}
From the last display, we observe that, on one hand,
\begin{align*}
        & \sum_j c^*_{j,a} \Pb \left[ \left\vert \mathbbm{1}\left\{\widehat{\zeta}_j > 0 \right\} - \mathbbm{1}\left\{\zeta_j > 0 \right\}\right\vert \mathbbm{1}\left\{\left\vert \widehat{\zeta}_j - \zeta_j \right\vert \leq \kappa \right\} \right]\\      
        &\leq \sum_j c^*_{j,a}  \Pb\left[\Pb\left(\left\vert \mathbbm{1}\left\{\widehat{\zeta}_j > 0 \right\} - \mathbbm{1}\left\{\zeta_j > 0 \right\}\right\vert \Bigm\vert \left\vert \widehat{\zeta}_j - \zeta_j \right\vert \leq \kappa \right) \mathbbm{1}\left\{\left\vert \widehat{\zeta}_j - \zeta_j \right\vert \leq \kappa \right\}  \right] \\
        &\leq \sum_j c^*_{j,a}  \Pb\left[\Pb\left(\mathbbm{1}\left\{ \vert \zeta_j \vert \leq \left\vert \widehat{\zeta}_j - \zeta_j \right\vert \, \&  \, \mu \in V_j(C^*) \right\} \Bigm\vert \left\vert \widehat{\zeta}_j - \zeta_j \right\vert \leq \kappa \right) \mathbbm{1}\left\{\left\vert \widehat{\zeta}_j - \zeta_j \right\vert \leq \kappa \right\}  \right] \\
        & \quad + \sum_j c^*_{j,a}  \Pb\left[\Pb\left(\mathbbm{1}\left\{ \vert \widehat{\zeta}_j \vert \leq \left\vert \widehat{\zeta}_j - \zeta_j \right\vert \, \&  \, \widehat{\mu} \in V_j(C^*) \right\} \Bigm\vert \left\vert \widehat{\zeta}_j - \zeta_j \right\vert \leq \kappa \right) \mathbbm{1}\left\{\left\vert \widehat{\zeta}_j - \zeta_j \right\vert \leq \kappa \right\}  \right] \\
        &\leq \sum_j c^*_{j,a}  \Pb\left[\Pb\left\{ \mu \in N_{C^*}\left(\left\vert \widehat{\zeta}_j - \zeta_j \right\vert \right) \Bigm\vert \left\vert \widehat{\zeta}_j - \zeta_j \right\vert \leq \kappa \right\} \mathbbm{1}\left\{\left\vert \widehat{\zeta}_j - \zeta_j \right\vert \leq \kappa \right\} \right] \\
        & \quad + \sum_j c^*_{j,a}  \Pb\left[\Pb\left\{ \widehat{\mu} \in N_{C^*}\left(\left\vert \widehat{\zeta}_j - \zeta_j \right\vert \right) \Bigm\vert \left\vert \widehat{\zeta}_j - \zeta_j \right\vert \leq \kappa \right\} \mathbbm{1}\left\{\left\vert \widehat{\zeta}_j - \zeta_j \right\vert \leq \kappa \right\} \right] \\
        &\lesssim \sum_j  \Vert \widehat{\zeta}_j - \zeta_j \Vert_{\infty}^\alpha 
\end{align*}  
where the first and second inequalities follow by the iterated expectation and Lemma \ref{lem:indicator-bound}, respectively, and the last by the margin condition. On the other hand, it also follows that
\begin{align*}
        & \sum_j c^*_{j,a} \Pb \left[ \left\vert \mathbbm{1}\left\{\widehat{\zeta}_j > 0 \right\} - \mathbbm{1}\left\{\zeta_j > 0 \right\}\right\vert \mathbbm{1}\left\{\left\vert \widehat{\zeta}_j - \zeta_j \right\vert > \kappa \right\} \right]\\              
        & \leq \sum_j c^*_{j,a} \left\Vert \mathbbm{1}\left\{\widehat{\zeta}_j > 0 \right\} - \mathbbm{1}\left\{\zeta_j > 0 \right\}\right\Vert_\infty \Pb\left[ \mathbbm{1}\left\{\left\vert \widehat{\zeta}_j - \zeta_j \right\vert > \kappa \right\}\right] \\
        & \lesssim \sum_j c^*_{j,a} \left\Vert \mathbbm{1}\left\{\widehat{\zeta}_j > 0 \right\} - \mathbbm{1}\left\{\zeta_j > 0 \right\}\right\Vert_\infty \left(\sum_a\Vert \widehat{\mu}_a - \mu_a \Vert_{\Pb,1}\right),
\end{align*}  
where the first and second inequalities follow by Hölder's and Markov's inequalities, respectively, and the fact that each $\zeta_j$ is Lipschitz at $\mu$.

Putting the two pieces together, we finally obtain that
\begin{align*}
    &\Pb\left\vert \Pi_{C^*,a}(\widehat{\mu}) - \Pi_{C^*,a}(\mu) \right\vert \\    
    &\lesssim \max_j\left\Vert \mathbbm{1}\left\{\zeta_j(\widehat{\mu};C^*) > 0 \right\} - \mathbbm{1}\left\{\zeta_j(\mu;C^*) > 0 \right\} \right\Vert_{\infty}\sum_a\Vert \widehat{\mu}_a - \mu_a \Vert_{\Pb,1} + \sum_a \Vert \widehat{\mu}_a - \mu_a \Vert_{\infty}^\alpha.
\end{align*}
\end{proof}

The next lemma shows that one may achieve faster rates for the bias of $f_{C^*}(\widehat{\mu})$.
\begin{lemma}\label{lem:f_C-bias-alphap1}
    Suppose that Assumption \ref{assumption:A1-boundedness} holds and $\Pb$ satisfies the margin condition with some $\kappa > 0$, $\alpha>0$. Then we have
    \begin{align*}
        & \left\vert \Pb\left\{ f_{C^*}(\widehat{\mu}) - f_{C^*}(\mu) \right\} \right\vert\\
        & \leq \max_a \left\Vert  \widehat{\mu}_a - \mu_a \right\Vert_{\Pb,1} + \max_a \Vert \widehat{\mu}_a - \mu_a \Vert_{\infty}^{\alpha + 1}  + \frac{1}{\kappa}\max_a \left(\Vert \widehat{\mu}_a - \mu_a \Vert_{\infty} \left\Vert \widehat{\mu}_a - \mu_a \right\Vert_{\Pb,1}\right).
    \end{align*}
\end{lemma}

\begin{proof}
    The proof follows the same logic that we develop in greater detail in the subsequent proof of Lemma \ref{lem:phi-bias-alphap1} (see Remark \ref{rmk:proof-of-f_C-bias-alphap1}).
\end{proof}

The following lemma computes the bias of our plug-in risk estimator $\widehat{R}_{n}$.
\begin{lemma} \label{lem:bias-plug-in-risk}
    Suppose $\Pb$ satisfies the margin condition for some $\kappa > 0$, $\alpha > 0$. Then under Assumptions \ref{assumption:A1-boundedness}, \ref{assumption:A2-consistency}, we have 
    \begin{align*}        
    & \widehat{R}_{n}(C^{*})-R(C^{*})\\
    &= O_\Pb\left( \frac{1}{\sqrt{n}} + \max_a \left\Vert  \widehat{\mu}_a - \mu_a \right\Vert_{\Pb,1} + \max_a \Vert \widehat{\mu}_a - \mu_a \Vert_{\infty}^{\alpha + 1}  + \frac{1}{\kappa}\max_a \left(\Vert \widehat{\mu}_a - \mu_a \Vert_{\infty} \left\Vert \widehat{\mu}_a - \mu_a \right\Vert_{\Pb,1}\right) \right),  
    \end{align*}
    whenever $\widehat{\mu}$ is constructed from a separate independent sample.
\end{lemma}

\begin{proof}
    It is immediate to see that
\begin{equation} \label{eqn:plug-in-risk-decomposition}
    \begin{aligned}
    \widehat{R}_{n}(C^{*})-R(C^{*}) &= \Pn\left\{f_{C^*}(\widehat{\mu})\right\} - \E\left\{f_{C^*}(\mu)\right\} \\
    &= (\Pn - \Pb)\left\{f_{C^*}(\widehat{\mu}) - f_{C^*}(\mu) \right\} \\
    & \quad +  (\Pn - \Pb)f_{C^*}(\mu) + \Pb\left\{ f_{C^*}(\widehat{\mu}) - f_{C^*}(\mu) \right\},
    \end{aligned}
\end{equation}
where $f_{C}(x) = \Vert x-\Pi_{C}(x)\Vert_2^2$, $\forall x \in \R^p$. 
The central limit theorem implies $(\Pn - \Pb)f_{C^*}(\mu) = O_\Pb(n^{-1/2})$. Also, it follows by Lemma \ref{lem:f_C-bias-alphap1} that
\begin{align*}
        & \Pb\left\{ f_{C^*}(\widehat{\mu}) - f_{C^*}(\mu) \right\} \\
        & \leq \max_a \left\Vert  \widehat{\mu}_a - \mu_a \right\Vert_{\Pb,1} + \max_a \Vert \widehat{\mu}_a - \mu_a \Vert_{\infty}^{\alpha + 1}  + \frac{1}{\kappa}\max_a \left(\Vert \widehat{\mu}_a - \mu_a \Vert_{\infty} \left\Vert \widehat{\mu}_a - \mu_a \right\Vert_{\Pb,1}\right).
\end{align*}

Further, under Assumption \ref{assumption:A1-boundedness}, it follows that
\begin{align} 
    f_{C^*}(\widehat{\mu}) - f_{C^*}(\mu) &= \Vert \widehat{\mu}-\Pi_{C}(\widehat{\mu})\Vert_2^2 - \Vert \mu-\Pi_{C}(\mu)\Vert_2^2 \nonumber \\
    & = \Vert \widehat{\mu}\Vert_2^2 + \widehat{\mu}^\top\Pi_{C^*}(\widehat{\mu}) +  \Vert \Pi_{C^*}(\widehat{\mu})\Vert_2^2 \nonumber\\
    & \quad - \left\{ \Vert \mu \Vert_2^2 + \mu^\top\Pi_{C^*}(\mu) +  \Vert \Pi_{C^*}(\mu)\Vert_2^2 \right\} \nonumber\\
    & = \Vert \widehat{\mu} \Vert_2^2 - \widehat{\mu}^\top\Pi_{C^*}(\widehat{\mu}) +\widehat{\mu}^\top\Pi_{C^*}(\mu) +  \Vert \Pi_{C^*}(\widehat{\mu})\Vert_2^2 \nonumber\\
    & \quad - \left\{ \Vert \mu \Vert_2^2 - \mu^\top\Pi_{C^*}(\mu) +\widehat{\mu}^\top\Pi_{C^*}(\mu) + \Vert \Pi_{C^*}(\mu)\Vert_2^2 \right\} \nonumber\\
    &\leq 3B \sum_a \left\{\left\vert \widehat{\mu}_a - \mu_a \right\vert + \left\vert  \Pi_{C^*,a}(\widehat{\mu}) - \Pi_{C^*,a}(\mu) \right\vert \right\}, \label{eqn:phi-mu-hat-bias-1}
\end{align}
which, by the triangle inequality, leads to
\begin{align*}
    \Vert f_{C^*}(\widehat{\mu}) - f_{C^*}(\mu) \Vert \lesssim \sum_a \left(\Vert \widehat{\mu}_a - \mu_a \Vert + \left\Vert \Pi_{C^*,a}(\widehat{\mu}) - \Pi_{C^*,a}(\mu) \right\Vert \right).
\end{align*}
For the second term in the last display, note that
\begin{align} \label{eqn:Pi_Cstar-muhat-sqr-bias}
    \Pb\left\{ \Pi_{C^*,a}(\widehat{\mu}) - \Pi_{C^*,a}(\mu) \right\}^2 &= \Pb\left( \sum_j c^*_{j,a} \left[\mathbbm{1}\left\{\zeta_j(\widehat{\mu}) > 0 \right\} - \mathbbm{1}\left\{\zeta_j(\mu) > 0 \right\} \right] \right)^2 \nonumber \\
    &\leq \Pb\left\{ \sum_j {c^*}_{j,a}^2 \sum_j \left[\mathbbm{1}\left\{\zeta_j(\widehat{\mu}) > 0 \right\} - \mathbbm{1}\left\{\zeta_j(\mu) > 0 \right\} \right]^2 \right\} \nonumber\\
    & \lesssim \sum_j c^*_{j,a} \Pb \left\vert \mathbbm{1}\left\{\zeta_j(\widehat{\mu}) > 0 \right\} - \mathbbm{1}\left\{\zeta_j(\mu) > 0 \right\} \right\vert \nonumber\\    
    & \lesssim \max_j\left\Vert \mathbbm{1}\left\{\zeta_j(\widehat{\mu};C^*) > 0 \right\} - \mathbbm{1}\left\{\zeta_j(\mu;C^*) > 0 \right\} \right\Vert_{\infty} \sum_a\Vert \widehat{\mu}_a - \mu_a \Vert_{\Pb,1} \nonumber\\
    & \quad + \sum_a \Vert \widehat{\mu}_a - \mu_a \Vert_{\infty}^\alpha,
\end{align}
where the last inequality follows by Lemma \ref{lem:Pi-bias-alpha}. Hence, by the given consistency condition in Assumption \ref{assumption:A2-consistency}, we get $\Vert \widehat{\mu}_a - \mu_a \Vert_{\Pb,1} = o_\Pb(1)$, $\Vert \widehat{\mu}_a - \mu_a \Vert_{\infty}^\alpha = o_\Pb(1)$, $\forall \alpha > 0$, and thereby conclude that $\left\Vert \Pi_{C^*,a}(\widehat{\mu}) - \Pi_{C^*,a}(\mu) \right\Vert = o_\Pb(1)$. Hence, $\Vert f_{C^*}(\widehat{\mu}) - f_{C^*}(\mu) \Vert = o_\Pb(1)$, and by the sample splitting lemma \cite[][Lemma 2]{kennedy2018sharp}, we obtain $(\Pn - \Pb)\left\{f_{C^*}(\widehat{\mu}) - f_{C^*}(\mu) \right\} = o_\Pb(n^{-1/2})$.

Putting the three pieces back together into \eqref{eqn:plug-in-risk-decomposition}, we obtain the desired bias bound as
\begin{align*} 
    & \widehat{R}_{n}(C^{*})-R(C^{*})\\
    &= O_\Pb\left( \frac{1}{\sqrt{n}} + \max_a \left\Vert  \widehat{\mu}_a - \mu_a \right\Vert_{\Pb,1} + \max_a \Vert \widehat{\mu}_a - \mu_a \Vert_{\infty}^{\alpha + 1}  + \frac{1}{\kappa}\max_a \left(\Vert \widehat{\mu}_a - \mu_a \Vert_{\infty} \left\Vert \widehat{\mu}_a - \mu_a \right\Vert_{\Pb,1}\right) \right).
\end{align*}
\end{proof}

The proof of Theorem \ref{thm:k-means} is established through Lemma \ref{lem:bias-plug-in-risk} and the auxiliary results utilized for the proof of Theorem \ref{thm:plug-in-consistency} (presented in Appendix \ref{appsec:proof-thm-3.2}).

\begin{proof}[Proof of Theorem \ref{thm:k-means}]
Notice that
\begin{align} \label{eq:k-means_factorization}
R(\widehat{C})-R(C^{*}) 
& = R(\widehat{C}) - R_{n}(\widehat{C}) + R_{n}(\widehat{C}) -\widehat{R}_{n}(\widehat{C})+\widehat{R}_{n}(\widehat{C})-R(C^{*}) \nonumber \\
& \leq R(\widehat{C}) - R_{n}(\widehat{C}) + R_{n}(\widehat{C}) -\widehat{R}_{n}(\widehat{C})+\widehat{R}_{n}(C^{*})-R(C^{*}) \nonumber\\
& \leq \underset{C\in\mathcal{C}_{k}}{\sup}\left|R(C)-R_{n}(C)\right| + R_{n}(\widehat{C}) - \widehat{R}_{n}(\widehat{C}) + \widehat{R}_{n}(C^{*})-R(C^{*}).
\end{align}
Since
$\left\Vert \mu  \right\Vert _{2} < \infty$ a.s., \citet[][Theorem 1]{linder1994rates} implies the following bound for the first term in \eqref{eq:k-means_factorization}: 
\begin{align}
\underset{C\in\mathcal{C}_{k}}{\sup}\left\vert R(C)-R_{n}(C)\right\vert = O_\Pb \left(\sqrt{\frac{\log n}{n}} \right) \label{eq:k-means_first} 
\end{align}

For the second term in \eqref{eq:k-means_factorization}, we observe that
\begin{align*}
    \widehat{R}_{n}(\widehat{C}) - R_{n}(\widehat{C}) &= \Pn\{f_{\widehat{C}}(\widehat{\mu})\} - \Pn\{f_{\widehat{C}}(\mu)\} \\
    &= (\Pn - \Pb)\left\{f_{\widehat{C}}(\widehat{\mu}) - f_{\widehat{C}}(\mu) \right\} + \Pb\left\{f_{\widehat{C}}(\widehat{\mu}) - f_{\widehat{C}}(\mu) \right\}.
\end{align*}
The terms on the RHS of the last display can be bounded using techniques that will be developed in the proof of Theorem~\ref{thm:plug-in-consistency}. First, by Lemma \ref{lem:f_C-supremum-bias}, we get
\begin{align*}
    \left\vert \Pb\left\{f_{\widehat{C}}(\widehat{\mu}) - f_{\widehat{C}}(\mu) \right\} \right\vert = O_\Pb\left( \max_a \Vert \widehat{\mu}_a - \mu_a \Vert_{\Pb,1} \right).
\end{align*}
Moreover, by the arguments used in analyzing terms (i) and (iii) in the proof of Theorem~\ref{thm:plug-in-consistency}, the class $\mathscr{F}_\mu=\{f_{C} \circ \mu:C\in\mathcal C_k\}$
is a VC-subgraph class and therefore $\Pb$-Donsker. Because composition by a measurable map preserves the VC index, the class
$\mathscr{F}_{\widehat\mu} =\{f_{C} \circ \widehat\mu : C\in\mathcal C_k\}$, is VC-subgraph with the same envelope as $\mathscr{F}_\mu$ as well. Consequently we get 
\begin{align*}
    (\Pn - \Pb)\left\{f_{\widehat{C}}(\widehat{\mu}) - f_{\widehat{C}}(\mu) \right\} \leq \underset{C\in\mathcal{C}_{k}}{\sup} (\Pn - \Pb)\left\{f_{C}(\widehat{\mu}) - f_{C}(\mu)\right\}  = O_\Pb\left(\frac{1}{\sqrt{n}}\right),
\end{align*}
and thus,
\begin{align*}
    \widehat{R}_{n}(\widehat{C}) - R_{n}(\widehat{C})  = O_\Pb\left(\frac{1}{\sqrt{n}} +  \max_a \Vert \widehat{\mu}_a - \mu_a \Vert_{\Pb,1} \right).
\end{align*}

For the third term in \eqref{eq:k-means_factorization}, we have
\begin{align}\label{eq:k-means_second}
    & \widehat{R}_{n}(C^{*})-R(C^{*}) \nonumber \\
    &= O_\Pb\left( \frac{1}{\sqrt{n}} + \max_a \left\Vert  \widehat{\mu}_a - \mu_a \right\Vert_{\Pb,1} + \max_a \Vert \widehat{\mu}_a - \mu_a \Vert_{\infty}^{\alpha + 1}  + \frac{1}{\kappa}\max_a \left(\Vert \widehat{\mu}_a - \mu_a \Vert_{\infty} \left\Vert \widehat{\mu}_a - \mu_a \right\Vert_{\Pb,1}\right) \right),
\end{align}
due to Lemma \ref{lem:bias-plug-in-risk}.

Assembling the preceding results, we establish that
\begin{align*}
    & R(\widehat{C})-R(C^{*}) \\
    & = O_\Pb\left( \sqrt{\frac{\log n}{n}} + \max_a \left\Vert  \widehat{\mu}_a - \mu_a \right\Vert_{\Pb,1} + \max_a \Vert \widehat{\mu}_a - \mu_a \Vert_{\infty}^{\alpha + 1}  + \frac{1}{\kappa}\max_a \left(\Vert \widehat{\mu}_a - \mu_a \Vert_{\infty} \left\Vert \widehat{\mu}_a - \mu_a \right\Vert_{\Pb,1}\right) \right).
\end{align*}

The same argument as in the preceding proof can be used to compute the rate of convergence in expectation as well. Specifically, when $\Vert \mu \Vert_\infty \leq B < \infty$ a.s., \citet[][Theorem 2.1]{biau2008performance} implies that
    \begin{align*}
        \Pb\left\{\underset{C\in\mathcal{C}_{k}}{\sup}\left\vert R(C)-R_{n}(C)\right\vert \right\} &\leq \frac{12B^2k}{\sqrt{n}}.
    \end{align*}
    Next, by Lemma \ref{lem:f_C-supremum-bias}, it follows that
    \begin{align*}
    \Pb\left\{ \widehat{R}_{n}(\widehat{C}) - R_{n}(\widehat{C}) \right\} &= \Pb\left\{f_{\widehat{C}}(\widehat{\mu}) - f_{\widehat{C}}(\mu) \right\} 
    = O_\Pb\left( \max_a \Vert \widehat{\mu}_a - \mu_a \Vert_{\Pb,1} \right).
    \end{align*}
    Also, by virtue of Lemma \ref{lem:f_C-bias-alphap1} one may deduce that
    \begin{align*}
        \Pb\left\{ \widehat{R}_{n}(C^{*})-R(C^{*}) \right\}  &= \Pb\left\{ f_{C^*}(\widehat{\mu}) - f_{C^*}(\mu) \right\} \\
        & \lesssim \max_a \left\Vert  \widehat{\mu}_a - \mu_a \right\Vert_{\Pb,1} + \max_a \Vert \widehat{\mu}_a - \mu_a \Vert_{\infty}^{\alpha + 1} \\
        & \quad + \frac{1}{\kappa}\max_a \left(\Vert \widehat{\mu}_a - \mu_a \Vert_{\infty} \left\Vert \widehat{\mu}_a - \mu_a \right\Vert_{\Pb,1}\right).
    \end{align*}
    With the incorporation of these bounds, \eqref{eq:k-means_factorization} reduces to:
    \begin{align*}        
    & \Pb\left\{R(\widehat{C})-R(C^{*}) \right\}\\
    & \lesssim \frac{1}{\sqrt{n}}+\max_a \left\Vert  \widehat{\mu}_a - \mu_a \right\Vert_{\Pb,1} + \max_a \Vert \widehat{\mu}_a - \mu_a \Vert_{\infty}^{\alpha + 1}  + \frac{1}{\kappa}\max_a \left(\Vert \widehat{\mu}_a - \mu_a \Vert_{\infty} \left\Vert \widehat{\mu}_a - \mu_a \right\Vert_{\Pb,1}\right).
    \end{align*} 
\end{proof}
\vspace*{.2in}

\subsection{Proof of Theorem \ref{thm:plug-in-consistency}} \label{appsec:proof-thm-3.2}

\begin{lemma} \label{lem:f_C-supremum-bias}
    For any $C \in \mathcal{C}_k$, under Assumption \ref{assumption:A1-boundedness}, we have
    \begin{align*}
        \underset{C \in \mathcal{C}_k}{\sup} \left\vert \Pb\left\{ f_{C}(\widehat{\mu}) - f_{C}(\mu) \right\} \right\vert \lesssim \max_a \Vert \widehat{\mu}_a - \mu_a \Vert_{\Pb,1}.
    \end{align*}
\end{lemma}

\begin{proof}
   We defer the proof until Lemma \ref{lem:phi_C-supremum-bias} (see Remark \ref{rmk:proof-of-f_C-supremum-bias}).
\end{proof}

\begin{proof}[Proof of Theorem \ref{thm:plug-in-consistency}]
First, we aim to show
\begin{align*}
    \underset{C\in\mathcal{C}_{k}}{\sup}\left\vert R_n(C) -  R(C)\right\vert = o_\Pb(1).
\end{align*}

To this end, consider the following decomposition for any $C \in \mathcal{C}_k$:
\begin{align*}
    R_n(C) - R(C) &= \underbrace{(\Pn - \Pb) \left\{f_{C}(\widehat{\mu}) - f_{C}(\mu) \right\}}_{(i)}  \\
    & \quad + \underbrace{\Pb \left\{ f_{C}(\widehat{\mu}) - f_{C}(\mu) \right\}}_{(ii)} \\
    & \quad + \underbrace{(\Pn - \Pb)\left\{f_{C}(\mu) \right\}}_{(iii)}.
\end{align*}
We will analyze the terms in the following order: (iii) $\rightarrow$ (ii) $\rightarrow$ (i).

\textbf{(iii)}  Consider sets $\mathscrsfs{G}$ of the subgraph $\{f_{C}(x) > u : (x,u) \in \R^{p} \times \R \}$.  The shattering number of $\mathscrsfs{G}$ is $s(\mathscrsfs{G},n) \leq n^{k(p+1)}$, which follows by the fact that each $\{f_{C}(x) > u \}$ is represented as a union of the complements of $k$ spheres. Hence the function class $\widetilde{\mathcal{F}} = \{f_{C}(\cdot) : C \in \mathcal{C}_k\}$ is a VC-class. For any fixed $\mu:\R^p \to \R$ and $f_C(\mu) = \Vert \mu - \Pi_C(\mu) \Vert_2^2 = f_C \circ \mu$, by the stability property \citep[e.g.,][Lemma 2.6.17]{van1996weak} the function class $\mathcal{F}_{\mu} = \{f_C(\mu(\cdot)) : C \in \mathcal{C}_k\}$ is also a VC-class. Taking $F_{\mu} = \underset{C \in \mathcal{C}_k}{\sup}\left\vert f_C(\mu) \right\vert$ as the envelope function, we have $\Pb\{F_{\mu}\} \leq 4B^2$ under the given boundedness condition. Thus, $\mathcal{F}_{\mu}$ is $\Pb$-Glivenko-Cantelli, yielding $\underset{C \in \mathcal{C}_k}{\sup}\left\vert (\Pn - \Pb)\left\{f_{C}(\mu)\right\} \right\vert = o_\Pb(1)$.

\textbf{(ii)} Under Assumption \ref{assumption:A1-boundedness}, by Lemma \ref{lem:f_C-supremum-bias} we have
\begin{align*}
        \underset{C \in \mathcal{C}_k}{\sup} \left\vert \Pb\left\{ f_{C}(\widehat{\mu}) - f_{C}(\mu) \right\} \right\vert \lesssim \max_a \Vert \widehat{\mu}_a - \mu_a \Vert_{\infty},
\end{align*}
which is $o_\Pb(1)$ under the consistency condition in Assumptions \ref{assumption:A2-consistency}.

\textbf{(i)} Let $\mathcal{F}_n = \mathcal{F}_{\hat{\mu}} - \mathcal{F}_{\mu}$ for the function class $\mathcal{F}_{\bar{\mu}} = \{f_C(\bar{\mu}(\cdot)) : C \in \mathcal{C}_k\}$ from before. Then,
\begin{align*}
    \left\Vert \frac{1}{\sqrt{n}} \Gn\left\{f_{C}(\widehat{\mu}) - f_{C}(\mu)\right\} \right\Vert_{\mathcal{C}_k} &= \underset{C \in \mathcal{C}_k}{\sup} \left\vert \frac{1}{\sqrt{n}} \Gn\left\{f_{C}(\widehat{\mu}) - f_{C}(\mu)\right\} \right\vert \\
    &= \frac{1}{\sqrt{n}}  \underset{f \in \mathcal{F}_n}{\sup} \left\vert\Gn(f)\right\vert.
\end{align*}

One may view the nuisance functions $\widehat{\mu}$ as fixed given the training data $D_0$. Since $\mathcal{F}_{\mu}$ is a VC-subgraph for any fixed $\mu$, so is $\mathcal{F}_n$ given $D_0$. Let the VC index of $\mathcal{F}_n$ be $\nu' < \infty$. Then we have
\begin{align*}
        \underset{Q}{\sup} N\left(\epsilon\Vert F_n \Vert_{Q,2}, \mathcal{F}_n, L_2(Q) \right) \lesssim \left( \frac{c_1}{\epsilon} \right)^{c_2\nu'}
\end{align*}
for some universal constants $c_1, c_2 > 0$. Hence applying \citet[][Theorem 3.5.4]{gine2021mathematical}, we obtain that 
\begin{align*}
    \Pb\left\{ \underset{f \in \mathcal{F}_n}{\sup} \left\vert\Gn(f)\right\vert \right\} &\lesssim \left\Vert F_n \right\Vert \underset{Q}{\sup}  \int_{0}^1 \sqrt{1 + \log N\left(\epsilon\Vert F_n \Vert_{Q,2}, \mathcal{F}_n, L_2(Q) \right)} d\epsilon\\
    &\lesssim \left\Vert F_n \right\Vert \int_{0}^1 \sqrt{1+\nu' \log(1/\epsilon)}d\epsilon.
\end{align*}
Taking the envelope $F_n = \underset{C \in \mathcal{C}_k}{\sup} \left\vert f_{C}(\widehat{\mu}) - f_{C}(\mu) \right\vert$ which is bounded, it is immediate to show that $\Pb\left\{ \underset{f \in \mathcal{F}_n^b}{\sup} \left\vert\Gn(f)\right\vert \right\} = O_\Pb(1)$ as the integral in the last display is finite. Consequently we get $\left\Vert (\Pn - \Pb)\left\{f_{C}(\widehat{\mu}) - f_{C}(\mu)\right\} \right\Vert_{\mathcal{C}_k} = O_\Pb(\frac{1}{\sqrt{n}}) = o_\Pb(1)$.

Now that we have shown
$
    \underset{C\in\mathcal{C}_{k}}{\sup}\left\vert \widehat{R}(C) -  R(C)\right\vert = o_\Pb(1),
$
the desired consistency $\widehat{C} \xrightarrow{p} C^*$ follows by \citet[][Theorem 5.7]{van2000asymptotic}, noting that $R(\cdot)$ is a continuous, bounded function whose domain $\mathcal{C}_k$ is compact, and that $C^*$ is unique (Assumption \ref{assumption:A3-uniqueness-of-codebook}).
\end{proof}
\vspace*{.2in}

\subsection{Proof of Lemma \ref{lem:eif-of-risk}}

To facilitate the proof of the main result, we begin by introducing several key lemmas.

\begin{lemma}\label{lem:E-bias-psi-2}  Under Assumptions \ref{assumption:A1-boundedness} and \ref{assumption:A4-pihat-boundedness}, we have that for any $a \in \mathcal{A}$,
    \begin{align*}
        \Pb\left\{ \varphi_{2,a}(\widehat{\eta}) - \varphi_{2,a}(\eta) \right\} \lesssim \Vert \widehat{\mu}_a - {\mu}_a \Vert \left( \Vert \widehat{\mu}_a - {\mu}_a \Vert + \Vert \widehat{\pi}_a - {\pi}_a \Vert \right).
    \end{align*}
\end{lemma}
\begin{proof}
Since
$
    \Pb\left\{\varphi_{2,a}(\eta)\right\} = \Pb\left\{\mu_{a}^2(X)\right\}
$, it follows
\begin{align*}
    \Pb\left\{ \varphi_{2,a}(\widehat{\eta}) - \varphi_{2,a}(\eta) \right\} &= \Pb\left\{ 2\widehat{\mu}_a\frac{\mathbbm{1}(A=a)}{\widehat{\pi}_a}\left\{Y- \widehat{\mu}_A\right\} + \widehat{\mu}_a^2 - \mu_{a}^2\right\} \\
    &= \Pb\left\{2\widehat{\mu}_a \frac{\pi_a}{\widehat{\pi}_a}(\mu_{a} - \widehat{\mu}_a) + (\widehat{\mu}_a - \mu_{a})(\widehat{\mu}_a + \mu_{a}) \right\}\\
    &= \Pb\left[ \left(\mu_{a} - \widehat{\mu}_a\right)\left\{4\widehat{\mu}_a\left(\frac{\pi_a - \widehat{\pi}_a}{\widehat{\pi}_a}\right) + \widehat{\mu}_a - {\mu}_a \right\} \right] \\
    &\leq \Pb\left\{ \left\vert \widehat{\mu}_a - {\mu}_a \right\vert \left( \left\vert \widehat{\mu}_a - {\mu}_a \right\vert + \frac{4B}{\epsilon} \left\vert \widehat{\pi}_a - {\pi}_a \right\vert \right) \right\} \\
    &\lesssim \left\Vert \widehat{\mu}_a - {\mu}_a \right\Vert\left( \left\Vert \widehat{\mu}_a - {\mu}_a \right\Vert + \left\Vert \widehat{\pi}_a - {\pi}_a \right\Vert \right)
\end{align*}    
\end{proof}

\begin{remark} \label{rmk:E-bias-psi-1} (\citet[][Example 2]{kennedy2022semiparametric})
    For $\varphi_{1,a}(\eta)$, it is well known that
    \begin{align*}
        \Pb\left\{ \varphi_{1,a}(\widehat{\eta}) - \varphi_{1,a}(\eta) \right\} \lesssim \Vert \widehat{\mu}_a - {\mu}_a \Vert  \Vert \widehat{\pi}_a - {\pi}_a \Vert.
    \end{align*}    
\end{remark}

\begin{lemma}\label{lem:phi-bias-alphap1}
    Suppose that Assumptions \ref{assumption:A1-boundedness}, \ref{assumption:A4-pihat-boundedness} hold and $\Pb$ satisfies the margin condition with some $\kappa > 0 $, $\alpha>0$. Then we have
    \begin{align*}
        & \left\vert \Pb\left\{ \varphi_{C^*}(\widehat{\eta}) - \varphi_{C^*}(\eta) \right\} \right\vert\\
        & \lesssim \max_a \left\Vert  \widehat{\mu}_a - {\mu}_a \right\Vert\left( \left\Vert \widehat{\mu}_a - {\mu}_a \right\Vert + \left\Vert \widehat{\pi}_a - {\pi}_a \right\Vert \right) + \max_a \Vert \widehat{\mu}_a - \mu_a \Vert_{\infty}^{\alpha + 1} \\
        & \quad + \frac{1}{\kappa}\max_a \left(\Vert \widehat{\mu}_a - \mu_a \Vert_{\infty} \left\Vert \widehat{\mu}_a - \mu_a \right\Vert_{\Pb,1}\right).
    \end{align*}
\end{lemma}

\begin{proof}
    Letting
    \begin{align*}
        f_{c^*_j}(\mu) &= \Vert \mu - c^*_j \Vert_2^2, \\
        \varphi_{c^*_j}(\eta) &= \sum_{a} \left\{\varphi_{2,a}(\eta) - 2\varphi_{1,a}(\eta)c^*_{ja} + {c^*}_{ja}^2 \right\}, 
    \end{align*}
    and 
    \begin{align*}
        d \equiv d(\mu;C^*) &= \argmin_j f_{c^*_j}(\mu), \\
        \widehat{d} \equiv d(\widehat{\mu};C^*) &= \argmin_j f_{c^*_j}(\widehat{\mu}),
    \end{align*}    
    one may write
    \begin{align*}
        \Pb\left\{ f_{C^*}(\mu) \right\} & = \Pb\left[ \min_{j\in\{1,\ldots,k\}} f_{c^*_j}(\mu) \right] = \sum_{j=1}^k \Pb\left\{\mathbbm{1}\{d = j\} f_{c^*_j}(\mu) \right\}, \\
        \Pb\{\varphi_{C^*}(\eta)\} & = \sum_{j=1}^k \Pb\left\{\mathbbm{1}\{d = j\} \varphi_{c^*_j}(\eta) \right\}, \\
        \Pb\{\varphi_{C^*}(\widehat{\eta})\} & = \sum_{j=1}^k \Pb\left\{\mathbbm{1}\{\widehat{d} = j\} \varphi_{c^*_j}(\widehat{\eta}) \right\}.
    \end{align*}        
    Now note that
    \begin{align}
        & \Pb\left\{ \varphi_{C^*}(\widehat{\eta}) - \varphi_{C^*}(\eta) \right\} \nonumber \\
        &= \sum_{j=1}^k \left( \Pb\left[ \mathbbm{1}(\widehat{d} = j)\left\{\varphi_{c^*_j}(\widehat{\eta}) - \varphi_{c^*_j}(\eta) \right\} \right] +  \Pb\left[ \left\{\mathbbm{1}(\widehat{d} = j) - \mathbbm{1}(d = j) \right\}\varphi_{c^*_j}(\eta) \right]\right) \nonumber \\
        &= \sum_{j=1}^k \left( \Pb\left[ \mathbbm{1}(\widehat{d} = j)\left\{\varphi_{c^*_j}(\widehat{\eta}) - \varphi_{c^*_j}(\eta) \right\} \right] +  \Pb\left[ \left\{\mathbbm{1}(\widehat{d} = j) - \mathbbm{1}(d = j) \right\}f_{c^*_j}(\mu) \right]\right), \label{eqn:phi_Cstar-bias-decomp}
    \end{align}
    where the last equality follows by the fact that $\Pb\left\{ f_{C^*}(\mu) \right\} = \Pb\{\varphi_{C^*}(\eta)\}$. For the first term in the last display,
    it is immediate to see by Lemma \ref{lem:E-bias-psi-2} and Remark \ref{rmk:E-bias-psi-1} that 
    \begin{align} \label{eqn:app-phi_bias_1}
        \sum_{j=1}^k \left\vert\Pb\left[ \mathbbm{1}(\widehat{d} = j)\left\{\varphi_{c^*_j}(\widehat{\eta}) - \varphi_{c^*_j}(\eta) \right\} \right] \right\vert
        \lesssim \max_a \left\Vert  \widehat{\mu}_a - {\mu}_a \right\Vert\left( \left\Vert \widehat{\mu}_a - {\mu}_a \right\Vert + \left\Vert \widehat{\pi}_a - {\pi}_a \right\Vert \right).
    \end{align}
    Next, let us rewrite the second term in \eqref{eqn:phi_Cstar-bias-decomp} by
    \begin{align*}
        & \sum_j \Pb\left[ \left\{\mathbbm{1}(\widehat{d} = j) - \mathbbm{1}(d = j) \right\}f_{c^*_j}(\mu) \right] \\
        &= \sum_j \Pb\Bigg( \left\{\mathbbm{1}(\widehat{d} = j) - \mathbbm{1}(d = j) \right\}f_{c^*_j}(\mu) \\
        & \qquad \qquad \times \left[\mathbbm{1}\left\{2 \max_j \left\vert f_{c^*_j}(\widehat{\mu}) - f_{c^*_j}(\mu) \right\vert \leq \kappa \right\} + \mathbbm{1}\left\{2 \max_j \left\vert f_{c^*_j}(\widehat{\mu}) - f_{c^*_j}(\mu) \right\vert > \kappa \right\}\right] \Bigg).
    \end{align*}    
    By mimicking the proof of Theorem 2 of \citet{levis2023covariate}, we have that
    \begin{align} \label{eqn:app-phi_bias_2}
        & \left\vert \sum_j \Pb\left[ \left\{\mathbbm{1}(\widehat{d} = j) - \mathbbm{1}(d = j) \right\}f_{c^*_j}(\mu)\mathbbm{1}\left\{2 \max_j \left\vert f_{c^*_j}(\widehat{\mu}) - f_{c^*_j}(\mu) \right\vert \leq \kappa \right\} \right] \right\vert \nonumber \\
        &= \Pb\left[\mathbbm{1}\left\{ f_{c^*_d}(\mu) < f_{c^*_{\widehat{d}}}(\mu) \right\} \left\{f_{c^*_{\widehat{d}}}(\mu) - f_{c^*_d}(\mu)\right\}\mathbbm{1}\left\{2 \max_j \left\vert f_{c^*_j}(\widehat{\mu}) - f_{c^*_j}(\mu) \right\vert \leq \kappa \right\} \right] \nonumber \\
        &\leq \Pb\Bigg( \mathbbm{1}\left[ \min_{j \neq d} \left\{ f_{c^*_j}(\mu) - f_{c^*_d}(\mu) \right\} \leq  f_{c^*_{\widehat{d}}}(\mu) - f_{c^*_d}(\mu) + f_{c^*_d}(\widehat{\mu}) - f_{c^*_{\widehat{d}}}(\widehat{\mu}) \right] \nonumber \\
        & \qquad \quad \times \left\{ f_{c^*_{\widehat{d}}}(\mu) - f_{c^*_d}(\mu) + f_{c^*_d}(\widehat{\mu}) - f_{c^*_{\widehat{d}}}(\widehat{\mu}) \right\} \mathbbm{1}\left\{2 \max_j \left\vert f_{c^*_j}(\widehat{\mu}) - f_{c^*_j}(\mu) \right\vert \leq \kappa \right\} \Bigg) \nonumber \\
        & \leq 2 \max_j \Vert f_{c^*_j}(\widehat{\mu}) - f_{c^*_j}(\mu) \Vert_{\infty}\nonumber \\
        & \quad \times \Pb\left[\zeta_j(\mu;C^*) \leq 2 \max_j \left\vert f_{c^*_j}(\widehat{\mu}) - f_{c^*_j}(\mu) \right\vert \Bigm\vert 2 \max_j \left\vert f_{c^*_j}(\widehat{\mu}) - f_{c^*_j}(\mu) \right\vert \leq \kappa \right] \nonumber \\
        & \lesssim \max_j \Vert f_{c^*_j}(\widehat{\mu}) - f_{c^*_j}(\mu) \Vert_{\infty}^{\alpha + 1}  \nonumber \\
        & \lesssim \max_a \Vert \widehat{\mu}_a - \mu_a \Vert_{\infty}^{\alpha + 1},
    \end{align} 
    where the first inequality follows by the fact that $f_{c^*_d}(\widehat{\mu}) \geq f_{c^*_{\widehat{d}}}(\widehat{\mu})$ and $f_{c^*_{\widehat{d}}}(\mu) \geq f_{c^*_d}(\mu)$, the third by the margin condition, and the last by local Lipschitz continuity of each $f_{c^*_j}$ at $\mu$ under Assumption \ref{assumption:A1-boundedness}. 
    
    Similarly as above, we also note that 
    \begin{align} \label{eqn:app-phi_bias_2-2}
        & \left\vert \sum_j \Pb\left[ \left\{\mathbbm{1}(\widehat{d} = j) - \mathbbm{1}(d = j) \right\}f_{C^*}(\mu) \mathbbm{1}\left\{2 \max_j \left\vert f_{c^*_j}(\widehat{\mu}) - f_{c^*_j}(\mu) \right\vert > \kappa \right\} \right] \right\vert \nonumber\\
        & = \Pb\left[ \mathbbm{1}\left\{ f_{c^*_d}(\mu) < f_{c^*_{\widehat{d}}}(\mu) \right\} \left\{f_{c^*_{\widehat{d}}}(\mu) - f_{c^*_d}(\mu) \right\} \mathbbm{1}\left\{2 \max_j \left\vert f_{c^*_j}(\widehat{\mu}) - f_{c^*_j}(\mu) \right\vert > \kappa \right\} \right] \nonumber\\
        & \leq \Pb\left[ \mathbbm{1}\left\{ f_{c^*_d}(\mu) < f_{c^*_{\widehat{d}}}(\mu) \right\} \left\{f_{c^*_{\widehat{d}}}(\mu) - f_{c^*_d}(\mu) + f_{c^*_d}(\widehat{\mu}) - f_{c^*_{\widehat{d}}}(\widehat{\mu}) \right\} \mathbbm{1}\left\{2 \max_j \left\vert f_{c^*_j}(\widehat{\mu}) - f_{c^*_j}(\mu) \right\vert > \kappa \right\} \right] \nonumber\\
        & \leq 2\max_j \Vert f_{c^*_j}(\widehat{\mu}) - f_{c^*_j}(\mu) \Vert_{\infty} \Pb\left\{ \max_j \left\vert f_{c^*_j}(\widehat{\mu}) - f_{c^*_j}(\mu) \right\vert > \kappa/2 \right\} \nonumber\\
        & \leq \frac{4}{\kappa} \max_j \Vert f_{c^*_j}(\widehat{\mu}) - f_{c^*_j}(\mu) \Vert_{\infty} \max_j \Pb\left\vert f_{c^*_j}(\widehat{\mu}) - f_{c^*_j}(\mu) \right\vert \nonumber\\
        & \lesssim \frac{1}{\kappa}\max_a \Vert \widehat{\mu}_a - \mu_a \Vert_{\infty} \left\Vert \widehat{\mu}_a - \mu_a \right\Vert_{\Pb,1},
    \end{align}
    which the first inequality follow by Hölder's inequality, the second by Markov's inequality. Putting these together, we finally obtain that
    \begin{align*} 
        & \left\vert \Pb\left\{ \varphi_{C^*}(\widehat{\eta}) - \varphi_{C^*}(\eta) \right\} \right\vert  \\
        & \lesssim \max_a \left\Vert  \widehat{\mu}_a - {\mu}_a \right\Vert\left( \left\Vert \widehat{\mu}_a - {\mu}_a \right\Vert + \left\Vert \widehat{\pi}_a - {\pi}_a \right\Vert \right) + \max_a \Vert \widehat{\mu}_a - \mu_a \Vert_{\infty}^{\alpha + 1}  \\
        & \quad + \frac{1}{\kappa}\max_a \Vert \widehat{\mu}_a - \mu_a \Vert_{\infty} \left\Vert \widehat{\mu}_a - \mu_a \right\Vert_{\Pb,1}. 
    \end{align*}    
\end{proof} 

\begin{remark}[Proof of Lemma \ref{lem:f_C-bias-alphap1}] \label{rmk:proof-of-f_C-bias-alphap1}
    The proof of Lemma \ref{lem:f_C-bias-alphap1} parallels the proof of Lemma \ref{lem:phi-bias-alphap1} provided above. Indeed, since we have the counterpart of \eqref{eqn:phi_Cstar-bias-decomp} as
    \begin{align*}         
        & \Pb\left\{ f_{C^*}(\widehat{\mu}) - f_{C^*}(\mu) \right\}\\        
        &= \sum_{j=1}^k \left( \Pb\left[ \mathbbm{1}(\widehat{d} = j)\left\{f_{c^*_j}(\widehat{\mu}) - f_{c^*_j}(\mu) \right\} \right] +  \Pb\left[ \left\{\mathbbm{1}(\widehat{d} = j) - \mathbbm{1}(d = j) \right\}f_{c^*_j}(\mu) \right]\right),    
    \end{align*}
    the only difference is to replace \eqref{eqn:app-phi_bias_1} with
    \begin{align*}
        \sum_{j=1}^k \left\vert\Pb\left[ \mathbbm{1}(\widehat{d} = j)\left\{f_{c^*_j}(\widehat{\mu}) - f_{c^*_j}(\mu) \right\} \right] \right\vert
        \lesssim \max_a \left\Vert  \widehat{\mu}_a - {\mu}_a \right\Vert_{\Pb,1},
    \end{align*}
    which gives the result.        
\end{remark}

Using the same logic as in the proof of Lemma \ref{lem:phi-bias-alphap1}, we may obtain the following uniform bound.
\begin{lemma} \label{lem:phi_C-supremum-bias}
    For any $C \in \mathcal{C}_k$, under Assumptions \ref{assumption:A1-boundedness}, \ref{assumption:A4-pihat-boundedness}, we have
    \begin{align*}
        \underset{C \in \mathcal{C}_k}{\sup} \left\vert \Pb\left\{ \varphi_{C}(\widehat{\eta}) - \varphi_{C}(\eta) \right\} \right\vert &\lesssim \max_a \Vert \widehat{\mu}_a - \mu_a \Vert_{\Pb,1} + \max_a \left\Vert  \widehat{\mu}_a - {\mu}_a \right\Vert\left( \left\Vert \widehat{\mu}_a - {\mu}_a \right\Vert + \left\Vert \widehat{\pi}_a - {\pi}_a \right\Vert \right).        
    \end{align*}
\end{lemma}

\begin{proof}
    Notice that \eqref{eqn:phi_Cstar-bias-decomp} and \eqref{eqn:app-phi_bias_1} hold for any $C \in \mathcal{C}_k$, i.e.,
     \begin{align*}
        & \Pb\left\{ \varphi_{C}(\widehat{\eta}) - \varphi_{C}(\eta) \right\} \\        
        &= \sum_{j=1}^k \left( \Pb\left[ \mathbbm{1}(\widehat{d} = j)\left\{\varphi_{c_j}(\widehat{\eta}) - \varphi_{c_j}(\eta) \right\} \right] +  \Pb\left[ \left\{\mathbbm{1}(\widehat{d} = j) - \mathbbm{1}(d = j) \right\}f_{c_j}(\mu) \right]\right), 
    \end{align*}
    where 
    \begin{align*} 
        \sum_{j=1}^k \left\vert\Pb\left[ \mathbbm{1}(\widehat{d} = j)\left\{\varphi_{c_j}(\widehat{\eta}) - \varphi_{c_j}(\eta) \right\} \right] \right\vert
        \lesssim \max_a \left\Vert  \widehat{\mu}_a - {\mu}_a \right\Vert\left( \left\Vert \widehat{\mu}_a - {\mu}_a \right\Vert + \left\Vert \widehat{\pi}_a - {\pi}_a \right\Vert \right).
    \end{align*}
    Further, proceeding similarly to \eqref{eqn:app-phi_bias_2}, we may get
    \begin{align*}
        & \left\vert \sum_j \Pb\left[ \left\{\mathbbm{1}(\widehat{d} = j) - \mathbbm{1}(d = j) \right\}f_{c_j}(\mu) \right] \right\vert\\
        &= \Pb\left[ \mathbbm{1}\left\{ f_{c_d}(\mu) < f_{c_{\widehat{d}}}(\mu) \right\} \left\{f_{c_{\widehat{d}}}(\mu) - f_{c_d}(\mu) + f_{c_d}(\widehat{\mu}) - f_{c_{\widehat{d}}}(\widehat{\mu}) \right\} \right] \\
        & \leq 2\max_j \Vert f_{c_j}(\widehat{\mu}) - f_{c_j}(\mu) \Vert_{\Pb,1} \\        
        & \lesssim \max_a \Vert \widehat{\mu}_a - \mu_a \Vert_{\Pb,1},
    \end{align*}
    which follows by Hölder's inequality and local Lipchitz continuity of $f_{c_j}$ at $\mu$.
    Hence, we conclude that for any $C \in \mathcal{C}_k$,
    \begin{align*}
        \left\vert \Pb\left\{ \varphi_{C}(\widehat{\eta}) - \varphi_{C}(\eta) \right\} \right\vert &\lesssim \max_a \Vert \widehat{\mu}_a - \mu_a \Vert_{\Pb,1} + \max_a \left\Vert  \widehat{\mu}_a - {\mu}_a \right\Vert\left( \left\Vert \widehat{\mu}_a - {\mu}_a \right\Vert + \left\Vert \widehat{\pi}_a - {\pi}_a \right\Vert \right).
    \end{align*}
    The result arises from the fact that the RHS is independent of $C$.
\end{proof}

\begin{remark} [Proof of Lemma \ref{lem:f_C-supremum-bias}] \label{rmk:proof-of-f_C-supremum-bias}
    The proof of Lemma \ref{lem:f_C-supremum-bias} parallels that of Lemma \ref{lem:phi_C-supremum-bias} given above. Indeed, for $\left\vert \Pb\left\{ f_{C}(\widehat{\mu}) - f_{C}(\mu) \right\} \right\vert$, both
    \begin{align*}
    & \left\vert\Pb\left[ \left\{\mathbbm{1}(\widehat{d} = j) - \mathbbm{1}(d = j) \right\}f_{c_j}(\mu) \right]\right\vert,\\
    & \left\vert\Pb\left[ \mathbbm{1}(\widehat{d} = j)\left\{f_{c_j}(\widehat{\eta}) - f_{c_j}(\eta) \right\} \right] \right\vert
    \end{align*}    
    are $O\left(\max_a \Vert \widehat{\mu}_a - \mu_a \Vert_{\Pb,1}\right)$.
\end{remark}

We are now in a position to prove Lemma~\ref{lem:eif-of-risk}.

\begin{proof}[Proof of Lemma \ref{lem:eif-of-risk}]
Recall that $\phi_{C^*}(z;\Pb) = \varphi_{C^*}(z;\Pb) - \int \varphi_{C^*}(z;\Pb)d\Pb$ and $R(C^*) = \Pb\{\varphi_{C^*}(\eta)\}   \equiv \int \varphi_{C^*}(z;\Pb) d\Pb$. For two distributions $\bar{\Pb}, \Pb$, the second-order remainder term in the von Mises expansion is given by
\begin{equation} \label{eqn:von-mises-expansion-1}
\begin{aligned}
    R_2(\bar{\Pb}, \Pb) &= \bar{R}(C^*) - R(C^*) + \int f_{C^*}(z;\bar{\Pb}) d\Pb \\
    &= \int \left\{\varphi_{C^*}(z;\Pb) - \varphi_{C^*}(z;\bar{\Pb}) \right\} d\Pb.
\end{aligned}
\end{equation}
By Lemma \ref{lem:phi-bias-alphap1},
the last term in \eqref{eqn:von-mises-expansion-1} is further bounded as
\begin{equation*} 
    \begin{aligned}
    \left\vert \Pb\left\{ \varphi_{C^*}(\bar{\eta}) - \varphi_{C^*}(\eta) \right\} \right\vert 
    & \lesssim \max_a \left\Vert  \bar{\mu}_a - {\mu}_a \right\Vert\left( \left\Vert \bar{\mu}_a - {\mu}_a \right\Vert + \left\Vert \bar{\pi}_a - {\pi}_a \right\Vert \right) + \max_a \Vert \bar{\mu}_a - \mu_a \Vert_{\infty}^{\alpha + 1} \\
        & \quad + \frac{1}{\kappa}\max_a \left(\Vert \bar{\mu}_a - \mu_a \Vert_{\infty} \left\Vert \bar{\mu}_a - \mu_a \right\Vert_{\Pb,1}\right).
    \end{aligned}
\end{equation*}

Hence for a submodel $\Pb_\varepsilon$, we have
\[
    \frac{d}{d\varepsilon}R_2(\Pb, \Pb_\varepsilon) \Bigm\vert_{\varepsilon=0} = 0,
\]
by virtue of the fact that the remainder $R_2(\Pb, \Pb_\varepsilon)$ essentially consists of only second-order products of errors between $\Pb, \Pb_\varepsilon$. Since there is at most one efficient influence function in nonparametric models, now we can apply Lemma 2 of \citet[][]{kennedy2023semiparametric} and conclude that $\phi_{C^*}$ is the efficient influence function.
\end{proof}
\vspace*{.2in}

\subsection{Proof of Lemma \ref{lem:root-n-CAN-Rhat}}

\begin{proof}[Proof of Lemma \ref{lem:root-n-CAN-Rhat}]
For any $C^* \in \mathcal{C}_k^*$, one may write
\begin{align*} 
    \widehat{R}(C^*) &= \sum_{b=1}^K \Pn\left\{ \varphi_{C^*}(\widehat{\eta}_{-b}) \mathbbm{1}(B=b)\right\} \nonumber \\
    R(C^*) &= \E(\varphi_{C^*}) = \sum_{b=1}^K \Pb\left\{ \varphi_{C^*}(\eta)\mathbbm{1}(B=b) \right\},
\end{align*}
where we drop the dependence on $Z$ in $\varphi_{C^*}$ for simplicity. Then consider the following decomposition:
\begin{align*}
    \sqrt{n}\left\{ \widehat{R}(C^*) - R(C^*) \right\} &= \sum_{b=1}^K \underbrace{\Gn\left[\left\{\varphi_{C^*}(\widehat{\eta}_{-b}) - \varphi_{C^*}(\eta)  \right\} \mathbbm{1}(B=b) \right] }_\text{\clap{(i)~}} \\
    & \quad + \sqrt{n}\sum_{b=1}^K\underbrace{\Pb\left[ \left\{ \varphi_{C^*}(\widehat{\eta}_{-b}) - \varphi_{C^*}(\eta) \right\}\mathbbm{1}(B=b) \right]}_\text{\clap{(ii)~}} \\
    & \quad + \Gn\left\{ \varphi_{C^*}(\eta) \right\}.
\end{align*}

It suffices to show that the terms $(i)$ and $(ii)$ are negligible, as the last term converges to $N\left(0, \var\left(\varphi_{C^*} \right) \right)$ by the central limit theorem.

\textbf{(i)} Noting $n \lesssim n/K$ with fixed $K$, we have
\begin{align*}
    & \left\Vert \left\{ \varphi_{C^*}(\widehat{\eta}_{-b}) - \varphi_{C^*}(\eta) \right\} \mathbbm{1}(B=b)\right\Vert 
    \lesssim \left\Vert \varphi_{C^*}(\widehat{\eta}) - \varphi_{C^*}(\eta) \right\Vert \\
    & \quad \leq \sum_a \left\Vert \varphi_{2,a}(\widehat{\eta}) - \varphi_{2,a}(\eta) + 2\Pi_{C^*,a}(\mu)(\varphi_{1,a}(\eta) - \varphi_{1,a}(\widehat{\eta})) + \left\{\widehat{\pi}_a + {\pi}_a - 2\varphi_{1,a}(\widehat{\eta})\right\}(\widehat{\pi}_a - {\pi}_a) \right\Vert \\
    & \quad \lesssim \sum_a\left( \left\Vert \varphi_{2,a}(\widehat{\eta}) - \varphi_{2,a}(\eta) \right\Vert + \left\Vert \varphi_{1,a}(\widehat{\eta}) - \varphi_{1,a}(\eta) \right\Vert + \left\Vert \Pi_{C^*,a}(\widehat{\mu}) - \Pi_{C^*,a}(\mu) \right\Vert \right).
\end{align*}

By adding and subtracting terms, it is straightforward to show
\begin{align*}
    &\left\Vert \varphi_{2,a}(\widehat{\eta}) - \varphi_{2,a}(\eta) \right\Vert \\ 
    & \quad \leq \left\Vert \widehat{\mu}_a\frac{\mathbbm{1}(A=a)}{\widehat{\pi}_a}\left(\mu_A- \widehat{\mu}_A\right) + \mathbbm{1}(A=a)(Y-\mu_A)\left( \frac{\widehat{\mu}_a}{\widehat{\pi}_a} - \frac{\mu_a}{\widehat{\pi}_a} + \frac{\mu_a}{\widehat{\pi}_a} - \frac{\mu_a}{\pi_a} \right) + (\widehat{\mu}_a - \mu_a)(\widehat{\pi}_a - \pi_a) \right\Vert \\
    & \quad \lesssim \left\Vert \widehat{\mu}_a - \mu_a \right\Vert + \left\Vert \widehat{\pi}_a - \pi_a \right\Vert.
\end{align*}
Similarly, one may get
\[
\left\Vert \varphi_{1,a}(\widehat{\eta}) - \varphi_{1,a}(\eta) \right\Vert \lesssim \left\Vert \widehat{\mu}_a - \mu_a \right\Vert + \left\Vert \widehat{\pi}_a - \pi_a \right\Vert.
\]

Further, we showed in \eqref{eqn:Pi_Cstar-muhat-sqr-bias} that $\left\Vert \Pi_{C^*,a}(\widehat{\mu}) - \Pi_{C^*,a}(\mu) \right\Vert = o_\Pb(1)$ if $\max_a \Vert \widehat{\mu}_a - \mu_a \Vert_{\infty} = o_\Pb(1)$.

Putting the three pieces together, we conclude that $\left\Vert \varphi_{C^*}(\widehat{\eta}) - \varphi_{C^*}(\eta) \right\Vert =  o_\Pb(1)$ under the consistency condition in Assumption \ref{assumption:A5-np-consistency-condition}.
Hence, we conclude 
\begin{align*}
    \Gn\left[\left\{\varphi_{C^*}(\widehat{\eta}_{-b}) - \varphi_{C^*}(\eta)  \right\} \mathbbm{1}(B=b) \right] = o_\Pb\left(\frac{1}{\sqrt{n}}\right),
\end{align*}
which follows by the sample splitting lemma \cite[][Lemma 2]{kennedy2018sharp}.

\textbf{(ii)} Noting that
\begin{equation*} 
    \begin{aligned}
    \left\vert \Pb\left[ \left\{ \varphi_{C^*}(\widehat{\eta}_{-b}) - \varphi_{C^*}(\eta) \right\}\mathbbm{1}(B=b) \right] \right\vert &\lesssim \left\vert \Pb\left\{ \varphi_{C^*}(\widehat{\eta}) - \varphi_{C^*}(\eta) \right\} \right\vert,
    \end{aligned}
\end{equation*}
by Lemma \ref{lem:phi-bias-alphap1} we get
\begin{align*}
    \left\vert \Pb\left\{ \varphi_{C^*}(\widehat{\eta}) - \varphi_{C^*}(\eta) \right\} \right\vert & \lesssim \max_a \left\Vert  \widehat{\mu}_a - {\mu}_a \right\Vert\left( \left\Vert \widehat{\mu}_a - {\mu}_a \right\Vert + \left\Vert \widehat{\pi}_a - {\pi}_a \right\Vert \right) + \max_a \Vert \widehat{\mu}_a - \mu_a \Vert_{\infty}^{\alpha + 1} \\
    & \quad + \frac{1}{\kappa}\max_a \left(\Vert \widehat{\mu}_a - \mu_a \Vert_{\infty} \left\Vert \widehat{\mu}_a - \mu_a \right\Vert_{\Pb,1}\right).
\end{align*}
which is $o_\Pb(\frac{1}{\sqrt{n}})$ by the given nonparametric condition $R_{2,n} = o_\Pb(n^{-1/2})$. 

Finally, the desired result follows by Slutsky's theorem.
\end{proof}

\vspace*{.2in}

\subsection{Proof of Corollary \ref{cor:consistency-of-Chat-eif-estimator}}
\begin{proof}[Proof of Corollary \ref{cor:consistency-of-Chat-eif-estimator}]
    The proof follows the exact same logic as that of Theorem \ref{thm:plug-in-consistency}. It boils down to show $\underset{C\in\mathcal{C}_{k}}{\sup} \left\vert  \Pb \left\{ \varphi_{C}(\widehat{\eta}) - \varphi_{C}(\eta) \right\} \right\vert = o_\Pb(1)$. This follows under the consistency condition in Assumption \ref{assumption:A5-np-consistency-condition} since    
    \begin{align*}
        \underset{C \in \mathcal{C}_k}{\sup} \left\vert \Pb\left\{ \varphi_{C}(\widehat{\eta}) - \varphi_{C}(\eta) \right\} \right\vert \lesssim \max_a \Vert \widehat{\mu}_a - \mu_a \Vert_{\infty}
    \end{align*}
    due to Lemma \ref{lem:phi_C-supremum-bias}. 
\end{proof} 
\vspace*{.2in}

\subsection{Proof of Theorem \ref{thm:Chat-root-nCAN}} \label{app-sec:proof-thm-Chat-root-nCAN}

First, we introduce technical lemmas showing how sample perturbations translate into stability of the empirical objective. 

\begin{lemma}[Primary-sample perturbation]
\label{lem:primary-perturbation-eif-objective}
Suppose that the nuisance estimator $\widehat\eta$ is fitted on an auxiliary block
$Z_{n+1:N}=\{Z_{n+1},\dots,Z_N\}$, independent of the primary sample
$Z_{1:n}=\{Z_1,\dots,Z_n\}$. For each $r \in \{1,\ldots,n\}$, let
\[
Z_i^{(r)} =
\begin{cases}
Z_i, & i \neq r,\\
Z_r', & i=r,
\end{cases}
\]
where $Z_r'$ is an independent copy of $Z_r$, and let
\[
\mathbb P_n^{(r)} f = \frac{1}{n}\sum_{i=1}^n f\!\left(Z_i^{(r)}\right)
\]
denote the empirical measure based on the modified primary sample
$\{Z_1,\ldots,Z_{r-1},Z_r',Z_{r+1},\ldots,Z_n\}$. Define
\[
\widehat R_n(C;\widehat\eta)
=
\mathbb P_n \varphi_C(\,\cdot\,;\widehat\eta),
\qquad
\widehat R_n^{(r)}(C;\widehat\eta)
=
\mathbb P_n^{(r)} \varphi_C(\,\cdot\,;\widehat\eta).
\]
Then, under Assumptions
\ref{assumption:A1-boundedness}
and \ref{assumption:A4-pihat-boundedness},
\[
\sup_{C \in \mathcal C_k}
\left|
\widehat R_n(C;\widehat\eta)
-
\widehat R_n^{(r)}(C;\widehat\eta)
\right|
=
O_\Pb(n^{-1}),
\]
uniformly in $r \in \{1,\ldots,n\}$.
\end{lemma}

\begin{proof}
Since $\widehat\eta$ is fitted on the auxiliary block, replacing $Z_r$ with
$r \le n$ does not affect $\widehat\eta$. Hence, for any $C \in \mathcal C_k$,
\[
\widehat R_n(C;\widehat\eta)
-
\widehat R_n^{(r)}(C;\widehat\eta)
=
\frac{1}{n}
\left\{
\varphi_C(Z_r;\widehat\eta)
-
\varphi_C(Z_r';\widehat\eta)
\right\}.
\]
Therefore,
\[
\sup_{C \in \mathcal C_k}
\left|
\widehat R_n(C;\widehat\eta)
-
\widehat R_n^{(r)}(C;\widehat\eta)
\right|
\le
\frac{1}{n}
\sup_{C \in \mathcal C_k}
\left|
\varphi_C(Z_r;\widehat\eta)
-
\varphi_C(Z_r';\widehat\eta)
\right|.
\]
Under Assumptions \ref{assumption:A1-boundedness} and
\ref{assumption:A4-pihat-boundedness}, the criterion
$\varphi_C(z;\eta)$ is uniformly bounded over admissible $z$, $\eta$, and
$C \in \mathcal C_k$. Hence the right-hand side is $O_\Pb(n^{-1})$,
uniformly in $r$.
\end{proof}

Next, we formalize a simple geometric fact: a change in Voronoi label can occur only near the Voronoi boundary.

\begin{lemma}[Voronoi label stability under small codebook perturbations]
\label{lem:voronoi-label-stability}
Let $C=\{c_1,\ldots,c_k\}$ and $C'=\{c_1',\ldots,c_k'\}$ be two codebooks in $\R^p$
such that
\[
\max_{1\le j\le k}\|c_j-c_j'\|_2 \le \delta.
\]
Then for any $u \in \R^p$, if
\[
\ell_C(u)\neq \ell_{C'}(u),
\]
it must hold that
\[
\operatorname{dist}(u,\partial C)\le 2\delta.
\]
\end{lemma}

\begin{proof}
Suppose $\operatorname{dist}(u,\partial C)>2\delta$, and let $j=\ell_C(u)$.
Then for every $m\neq j$,
\[
\|u-c_m\|_2-\|u-c_j\|_2 > 2\delta.
\]
Using the triangle inequality and $\|c_j-c_j'\|_2,\|c_m-c_m'\|_2\le\delta$,
\[
\|u-c_m'\|_2
\ge \|u-c_m\|_2-\delta
>
\|u-c_j\|_2+\delta
\ge
\|u-c_j'\|_2.
\]
Hence $\ell_{C'}(u)=j=\ell_C(u)$, a contradiction.
\end{proof}

For the subsequent proof, we adopt a sample-splitting representation of cross-fitting and restate Assumption~\ref{assumption:A7-hard-fitted-margin} in an equivalent but more transparent form. Specifically, we condition on one training fold, so that the corresponding validation fold serves as the primary sample and its complement serves as the auxiliary sample, as stated below.

\begin{assumptionp}{A7$^\prime$}
\label{assumption:A7'-hard-fitted-margin}
Let $\Pn$ be the empirical measure based on the primary sample
$Z_{1:n}=\{Z_1,\dots,Z_n\}$, and let $\Pb$ denote the population distribution.
Suppose that the nuisance estimator $\widehat\eta$ is fitted on an auxiliary block
$Z_{n+1:N}=\{Z_{n+1},\dots,Z_N\}$, independent of $Z_{1:n}$. Then there exists a sequence
$\rho_n > 0$ such that $\rho_n^{-1}=o(n)$ and, with probability tending to one, for every
codebook $C$ lying on the line segment between $\widehat C$ and a leave-one-out perturbation
of $\widehat C$ obtained by replacing a single primary-sample observation, and for every
$i \in \{1,\dots,n\}$,
\[
\operatorname{dist}\!\bigl(\widehat\mu(X_i),\partial C\bigr)\ge \rho_n.
\]
Here $\partial C$ denotes the union of the Voronoi boundaries induced by $C$.
\end{assumptionp}

In what follows, we use Assumption~\ref{assumption:A7'-hard-fitted-margin} as the proof specific version of Assumption~\ref{assumption:A7-hard-fitted-margin}, as they are essentially equivalent to each other. The next lemma shows that the estimated codebook is stable under leave one out perturbations of the primary sample under this segmentwise separation condition.

\begin{lemma}
\label{lem:primary-codebook-stability}
Let $\Pn$ be the empirical measure based on the primary sample
$Z_{1:n}=\{Z_1,\dots,Z_n\}$, and let the nuisance estimator $\widehat\eta$ be fitted on an auxiliary block
$Z_{n+1:N}=\{Z_{n+1},\dots,Z_N\}$, independent of $Z_{1:n}$, with $N \asymp n$.
For each $r \in \{1,\ldots,n\}$, let $Z_r'$ be an independent copy of $Z_r$, let
$\Pn^{(r)}$ denote the empirical measure obtained by replacing $Z_r$ by $Z_r'$, and let
$\widehat C^{(r)}$ denote the empirical minimizer recomputed from the perturbed primary sample,
keeping the auxiliary block fixed. Assume Assumptions
\ref{assumption:A1-boundedness},
\ref{assumption:A3-uniqueness-of-codebook},
\ref{assumption:A4-pihat-boundedness},
\ref{assumption:A5-np-consistency-condition}, and
\ref{assumption:A7-hard-fitted-margin}. Then
\[
\widetilde\delta_n
:=
\max_{1 \le r \le n}
\min_{\sigma \in \mathfrak S_k}
\max_{1 \le j \le k}
\bigl\|
\widehat c_j-\widehat c_{\sigma(j)}^{(r)}
\bigr\|_2
=
O_\Pb(n^{-1})
=
O_\Pb(N^{-1}).
\]
\end{lemma}

\begin{proof}
Since $\widehat\eta$ is fitted on the auxiliary block, it is unchanged under replacement of
a single primary-sample observation. Hence, by Lemma
\ref{lem:primary-perturbation-eif-objective}, and by the same argument with $\varphi_C$
replaced by $\upvarphi_C$,
\[
\sup_{C \in \mathcal C_k}
\left|
\widehat R_n(C;\widehat\eta)-\widehat R_n^{(r)}(C;\widehat\eta)
\right|
=
O_\Pb(n^{-1}),
\quad
\sup_{C \in \mathcal C_k}
\left\|
\Psi_n(C;\widehat\eta)-\Psi_n^{(r)}(C;\widehat\eta)
\right\|_2
=
O_\Pb(n^{-1}),
\]
uniformly in $r \in \{1,\ldots,n\}$, where
\[
\Psi_n(C;\bar\eta)=\Pn\upvarphi_C(\,\cdot\,;\bar\eta),
\qquad
\Psi_n^{(r)}(C;\bar\eta)=\Pn^{(r)}\upvarphi_C(\,\cdot\,;\bar\eta).
\]

Since $\widehat C$ and $\widehat C^{(r)}$ are local minimizers of their respective
empirical criteria, they satisfy the approximate first-order empirical moment conditions
\[
\Psi_n(\widehat C;\widehat\eta)=o_\Pb(n^{-1/2}),
\qquad
\Psi_n^{(r)}(\widehat C^{(r)};\widehat\eta)=o_\Pb(n^{-1/2}).
\]
Subtracting these two displays gives
\[
\Psi_n(\widehat C;\widehat\eta)-\Psi_n^{(r)}(\widehat C^{(r)};\widehat\eta)
=
o_\Pb(n^{-1/2}).
\]

By Corollary \ref{cor:consistency-of-Chat-eif-estimator}, $\widehat C \xrightarrow[]{p} C^*$.
The same argument applies to each $\widehat C^{(r)}$, since the perturbed empirical
criterion differs from the original one by only a one-observation perturbation. Hence
\[
\max_{1\le r\le n}
\min_{\sigma\in\mathfrak S_k}
\|\widehat C^{(r)}-\sigma(C^*)\|_1
=
o_\Pb(1).
\]

By assumption, with probability tending to one uniformly in $r$, no fitted point
$\widehat\mu(X_i)$ lies on a Voronoi boundary for any codebook along the segment joining
$\widehat C$ and $\widehat C^{(r)}$. On this event, the Voronoi labels remain fixed along
that segment, and therefore, by \eqref{eqn:varphi-derivative}, the map
\[
C \mapsto \Psi_n(C;\widehat\eta)
\]
is affine on the segment. Thus
\[
\Psi_n(\widehat C;\widehat\eta)-\Psi_n(\widehat C^{(r)};\widehat\eta)
=
\widehat M_{n,r}\,(\widehat C-\widehat C^{(r)}),
\]
where $\widehat M_{n,r}$ is the corresponding block-diagonal empirical derivative matrix.

Moreover, since $M(C^*,\eta)$ is nonsingular, with eigenvalues bounded away from zero by
the condition $p_j^*>0$, and since $\widehat C$ and $\widehat C^{(r)}$ both converge to
$C^*$ uniformly in $r$, while the corresponding empirical cell proportions converge
uniformly to $p_j^*$, we have
\[
\sup_{1\le r\le n}\|\widehat M_{n,r}-M(C^*,\eta)\| \xrightarrow[]{p} 0.
\]
Therefore, the smallest eigenvalues of $\widehat M_{n,r}$ are bounded away from zero with
probability tending to one, and so
\[
\|\widehat M_{n,r}^{-1}\|=O_\Pb(1)
\]
uniformly in $r$.

Hence, on this event, we get the first-order expansion
\[
\widehat C-\widehat C^{(r)}
=
-\widehat M_{n,r}^{-1}
\Bigl[
\Psi_n(\widehat C^{(r)};\widehat\eta)
-
\Psi_n^{(r)}(\widehat C^{(r)};\widehat\eta)
\Bigr]
+
o_\Pb(n^{-1}),
\]
and therefore
\[
\max_{1\le r\le n}
\|\widehat C-\widehat C^{(r)}\|_2
=
O_\Pb(n^{-1}).
\]
Since $N \asymp n$, we conclude that
\[
\widetilde\delta_n
=
\max_{1 \le r \le n}
\min_{\sigma \in \mathfrak S_k}
\max_{1 \le j \le k}
\bigl\|
\widehat c_j-\widehat c_{\sigma(j)}^{(r)}
\bigr\|_2
=
O_\Pb(n^{-1})
=
O_\Pb(N^{-1}).
\]
\end{proof}

Similarly, we adopt a sample-splitting representation of cross-fitting and restate Assumption~\ref{assumption:A8-fitted-boundary-margin} in an equivalent form.

\begin{assumptionp}{A8$^\prime$}
\label{assumption:A8'-fitted-boundary-margin}
Let $\Pn$ be the empirical measure based on the primary sample
$Z_{1:n}=\{Z_1,\dots,Z_n\}$, and let $\Pb$ denote the population distribution.
Suppose that the nuisance estimator $\widehat\eta$ is fitted on an auxiliary block
$Z_{n+1:N}=\{Z_{n+1},\dots,Z_N\}$, independent of $Z_{1:n}$. Then there exist constants
$\kappa_{\mathrm{fit}} > 0$, $\beta > 0$, and $L < \infty$, and a random neighborhood
$\mathcal N_n$ of $C^*$ such that, with probability tending to one,

(i) $\widehat C \in \mathcal N_n$ and $\widehat C^{(r)} \in \mathcal N_n$ for all
$1 \le r \le n$, where $\widehat C^{(r)}$ denotes the codebook estimator recomputed after
replacing the $r$th primary sample observation $Z_r$ by an independent copy $Z_r'$; and

(ii) for every $0 < t \le \kappa_{\mathrm{fit}}$,
\[
\sup_{C \in \mathcal N_n}
\Pb\!\left\{
\operatorname{dist}\!\bigl(\widehat\mu(X),\partial C\bigr) \le t
\,\middle|\,
Z_{n+1:N}
\right\}
\le
L t^\beta.
\]
Here $\partial C$ denotes the union of the Voronoi boundaries induced by $C$.
\end{assumptionp}

We then prove the following lemma, a main ingredient to prove Theorem \ref{thm:Chat-root-nCAN}.
In what follows, we use Assumption~\ref{assumption:A8'-fitted-boundary-margin} as the proof specific version of Assumption~\ref{assumption:A8-fitted-boundary-margin}. Together with the leave one out stability established under Assumption~\ref{assumption:A7-hard-fitted-margin} in Lemma~\ref{lem:primary-codebook-stability}, this yields control of the empirical process cross term, as formalized in the next lemma.

\begin{lemma} \label{lem:cross-term}
Let $\Pn$ be the empirical measure based on the primary sample
$Z_{1:n}=\{Z_1,\dots,Z_n\}$, and let $\Pb$ denote the population distribution.
The nuisance estimator $\widehat\eta$ is fitted on an auxiliary block
$Z_{n+1:N}=\{Z_{n+1},\dots,Z_N\}$, independent of $Z_{1:n}$, with $N \asymp n$,
and the codebook estimator $\widehat C$ is computed using the full sample of size $N$.
Assume Assumptions \ref{assumption:A1-boundedness},
\ref{assumption:A3-uniqueness-of-codebook},
\ref{assumption:A4-pihat-boundedness},
\ref{assumption:A5-np-consistency-condition},
\ref{assumption:A7-hard-fitted-margin}, and
\ref{assumption:A8-fitted-boundary-margin}. Then
\[
\left\|
(\Pn-\Pb)\Bigl\{
\upvarphi_{\widehat C}(Z;\widehat\eta)
-
\upvarphi_{C^*}(Z;\widehat\eta)
\Bigr\}
\right\|_2
=
O_\Pb\!\left(n^{-\min\{\beta/2,1\}}\right).
\]
In particular, if $\beta \ge 1$, then
\[
\left\|
(\Pn-\Pb)\Bigl\{
\upvarphi_{\widehat C}(Z;\widehat\eta)
-
\upvarphi_{C^*}(Z;\widehat\eta)
\Bigr\}
\right\|_2
=
O_\Pb\!\left(n^{-1/2}\right).
\]
\end{lemma}

\begin{proof}
Throughout, condition on the auxiliary block $Z_{n+1:N}$. Under this conditioning,
$\widehat\eta$ is fixed, and
\[
\Delta_n
:=
(\Pn-\Pb)\Bigl\{
\upvarphi_{\widehat C}(Z;\widehat\eta)
-
\upvarphi_{C^*}(Z;\widehat\eta)
\Bigr\}
\]
is a functional only of the primary sample $Z_{1:n}$.

For each $r \in \{1,\ldots,n\}$, let $Z_r'$ be an independent copy of $Z_r$, and let
$\widehat C^{(r)}$ denote the empirical minimizer recomputed after replacing $Z_r$ by
$Z_r'$, keeping the auxiliary block fixed. Since codebooks are identified only up to
permutation, define
\[
\widetilde\delta_n
=
\max_{1 \le r \le n}
\min_{\sigma \in \mathfrak S_k}
\max_{1 \le j \le k}
\bigl\|
\widehat c_j-\widehat c_{\sigma(j)}^{(r)}
\bigr\|_2,
\]
where $\mathfrak S_k$ denotes the set of all permutations of $\{1,\ldots,k\}$.

By Lemma \ref{lem:primary-codebook-stability},
\[
\widetilde\delta_n
=
O_\Pb(n^{-1})
=
O_\Pb(N^{-1}),
\]
as $N \asymp n$.

Since the dimension $kp$ is fixed, it suffices to bound an arbitrary coordinate of
$\Delta_n$. Fix $\ell \in \{1,\dots,kp\}$, and define
\[
h(z)
=
\bigl[\upvarphi_{\widehat C}(z;\widehat\eta)
-
\upvarphi_{C^*}(z;\widehat\eta)\bigr]_\ell,
\quad
h^{(r)}(z)
=
\bigl[\upvarphi_{\widehat C^{(r)}}(z;\widehat\eta)
-
\upvarphi_{C^*}(z;\widehat\eta)\bigr]_\ell.
\]
Thus for $\Delta_n$, we write
\[
\Delta_{n,\ell}
=
(\Pn-\Pb)\{h(Z)\}.
\]

Let $\ell_C(x)$ denote the Voronoi label assigned by codebook $C$ to the fitted feature
vector $\widehat\mu(x)$, and define
\[
\mathfrak A_r(z)
=
\mathbbm 1\bigl\{
\ell_{\widehat C}(x)\neq \ell_{\widehat C^{(r)}}(x)
\bigr\},
\qquad z=(x,a,y).
\]
By Lemma \ref{lem:voronoi-label-stability}, on the event
$\{\widehat C,\widehat C^{(r)}\in\mathcal N_n\}$,
\[
\mathfrak A_r(z)=1
\quad\Longrightarrow\quad
\operatorname{dist}\!\bigl(\widehat\mu(x),\partial\widehat C\bigr)\le 2\widetilde\delta_n.
\]
Therefore, on the event
$\{\widehat C,\widehat C^{(r)}\in\mathcal N_n,\ \widetilde\delta_n<\kappa_{\mathrm{fit}}/2\}$,
Assumption \ref{assumption:A8-fitted-boundary-margin} yields
\[
\Pb\bigl\{\mathfrak A_r(Z)=1 \,\big|\, Z_{n+1:N}\bigr\}
\le
L(2\widetilde\delta_n)^\beta
=
O_\Pb(n^{-\beta}).
\]
Since $\widehat C,\widehat C^{(r)}\in\mathcal N_n$ with probability tending to one and
$\widetilde\delta_n=O_\Pb(n^{-1})$, we conclude that
\begin{equation}
\label{eq:Ar-bound-revised}
\Pb\bigl\{\mathfrak A_r(Z)=1 \,\big|\, Z_{n+1:N}\bigr\}
=
O_\Pb\!\bigl(n^{-\beta}\bigr).
\end{equation}

Next we bound $|h(z)-h^{(r)}(z)|$.
On the event $\{\mathfrak A_r(z)=0\}$, the same Voronoi cell is active under
$\widehat C$ and $\widehat C^{(r)}$. By the explicit form of $\upvarphi_C$ in
\eqref{eqn:varphi-derivative}, the difference between
$\upvarphi_{\widehat C}(z;\widehat\eta)$ and
$\upvarphi_{\widehat C^{(r)}}(z;\widehat\eta)$ is then due only to the shift in the
active center. Therefore,
\[
|h(z)-h^{(r)}(z)| \le 2\widetilde\delta_n
\quad\text{on}\quad\{\mathfrak A_r(z)=0\}.
\]
On the event $\{\mathfrak A_r(z)=1\}$, boundedness of $\upvarphi_C$ under
Assumptions \ref{assumption:A1-boundedness} and
\ref{assumption:A4-pihat-boundedness} implies that there exists $M<\infty$ such that
\[
|h(z)-h^{(r)}(z)| \le 2M.
\]
Putting these together, it follows that
\begin{equation}
\label{eq:h-diff-bound-revised}
|h(z)-h^{(r)}(z)|
\le
2M\,\mathfrak A_r(z) + 2\widetilde\delta_n.
\end{equation}

Using $(a+b)^2 \le 2a^2+2b^2$, together with \eqref{eq:Ar-bound-revised} and
$\widetilde\delta_n^2=O_\Pb(n^{-2})$, we obtain
\begin{equation}
\label{eq:h-second-moment-revised}
\E\bigl[(h(Z)-h^{(r)}(Z))^2 \,\big|\, Z_{n+1:N}\bigr]
=
O_\Pb\!\bigl(n^{-\beta}\bigr).
\end{equation}

Let $\Delta_{n,\ell}^{(r)}$ denote the quantity obtained from $\Delta_{n,\ell}$ after
replacing $Z_r$ by $Z_r'$. Since only the primary sample is perturbed,
\[
\Delta_{n,\ell}-\Delta_{n,\ell}^{(r)}
=
\frac{1}{n}\sum_{i=1}^n
\Bigl\{
h(Z_i)-h^{(r)}(Z_i)
\Bigr\}
-
\E\bigl[h(Z)-h^{(r)}(Z)\,\big|\, Z_{n+1:N}\bigr].
\]
Therefore,
\[
\E\bigl[(\Delta_{n,\ell}-\Delta_{n,\ell}^{(r)})^2 \,\big|\, Z_{n+1:N}\bigr]
\lesssim
\frac{1}{n}
\E\bigl[(h(Z)-h^{(r)}(Z))^2 \,\big|\, Z_{n+1:N}\bigr]
=
O_\Pb\!\bigl(n^{-1-\beta}\bigr).
\]

Applying the (conditional) Efron--Stein inequality to the scalar functional
$\Delta_{n,\ell}$ of the primary sample $Z_{1:n}$ yields
\[
\Var(\Delta_{n,\ell} \mid Z_{n+1:N})
\le
\frac12\sum_{r=1}^n
\E\bigl[(\Delta_{n,\ell}-\Delta_{n,\ell}^{(r)})^2 \,\big|\, Z_{n+1:N}\bigr]
=
O_\Pb\!\bigl(n^{-\beta}\bigr).
\]
Hence
\[
\Delta_{n,\ell}-\E(\Delta_{n,\ell} \mid Z_{n+1:N})
=
O_\Pb\!\bigl(n^{-\beta/2}\bigr).
\]

It remains to bound the conditional bias.
By the usual leave-one-out identity,
\[
\E(\Delta_{n,\ell} \mid Z_{n+1:N})
=
\frac1n\sum_{i=1}^{n}
\E\bigl\{h(Z_i)-h^{(i)}(Z_i)\,\big|\, Z_{n+1:N}\bigr\},
\]
since, conditional on the auxiliary block, \(Z_i\) is independent of the leave-one-out estimator \(\widehat C^{(i)}\).
Using \eqref{eq:h-diff-bound-revised} and \eqref{eq:Ar-bound-revised},
\[
\bigl|\E(\Delta_{n,\ell} \mid Z_{n+1:N})\bigr|
\lesssim
n^{-\beta} + n^{-1}
=
O_\Pb\!\bigl(n^{-\min\{\beta,1\}}\bigr).
\]

Using the decomposition
\[
\Delta_{n,\ell}
=
\Bigl(\Delta_{n,\ell}-\E(\Delta_{n,\ell}\mid Z_{n+1:N})\Bigr)
+
\E(\Delta_{n,\ell}\mid Z_{n+1:N}),
\]
together with the fluctuation and conditional bias bounds derived above, we obtain
\[
\Delta_{n,\ell}
=
O_\Pb\!\bigl(n^{-\beta/2}\bigr)
+
O_\Pb\!\bigl(n^{-\min\{\beta,1\}}\bigr)
=
O_\Pb\!\bigl(n^{-\min\{\beta/2,1\}}\bigr).
\]

Because the coordinate index $\ell$ was arbitrary and the dimension $kp$ is fixed,
the same rate holds for the full Euclidean norm:
\[
\left\|
(\Pn-\Pb)\Bigl\{
\upvarphi_{\widehat C}(Z;\widehat\eta)
-
\upvarphi_{C^*}(Z;\widehat\eta)
\Bigr\}
\right\|_2
=
O_\Pb\!\bigl(n^{-\min\{\beta/2,1\}}\bigr).
\]
In particular, if $\beta \ge 1$, then
\[
\left\|
(\Pn-\Pb)\Bigl\{
\upvarphi_{\widehat C}(Z;\widehat\eta)
-
\upvarphi_{C^*}(Z;\widehat\eta)
\Bigr\}
\right\|_2
=
O_\Pb\!\left(n^{-1/2}\right).
\]
This completes the proof.
\end{proof}

Lemma~\ref{lem:cross-term} enables control of the cross-term without invoking strong empirical process conditions. We are now prepared to prove Theorem~\ref{thm:Chat-root-nCAN}.

\begin{proof}[Proof of Theorem~\ref{thm:Chat-root-nCAN}]
The population first-order condition for the optimal codebook is
\[
\Pb\left\{\nabla \varphi_{C^*}(Z;\eta)\right\}
=
\Pb\left\{\upvarphi_{C^*}(Z;\eta)\right\}
=
0,
\]
where $\upvarphi_C$ is defined in \eqref{eqn:varphi-derivative}. Also,
\eqref{eqn:eif-codebook-estimator} is equivalent to minimizing $\widehat R(C)$ with
\begin{equation} \label{eqn:eif-of-risk-alternative}
\varphi_C(Z;\eta)
=
\sum_{a\in\mathcal A}
\left\{
\varphi_{1,a}^2(Z;\eta)
-
2\varphi_{1,a}(Z;\eta)\bigl[\Pi_C(\mu)\bigr]_a
+
\bigl[\Pi_C(\mu)\bigr]_a^2
\right\}.
\end{equation}
We proceed using \eqref{eqn:eif-of-risk-alternative}.

The argument parallels the proof of Theorem 3 of \citet{kennedy2023semiparametric}. By abuse of notation, we rewrite the empirical moment condition as
\begin{align}
o_\Pb\!\left(n^{-1/2}\right)
&=
\Pn\left\{\upvarphi_{\widehat C}(Z;\widehat\eta)\right\}
-
\Pb\left\{\upvarphi_{C^*}(Z;\eta)\right\}
\label{eqn:clt-proof-moment-condition}\\
&=
(\Pn-\Pb)\left\{\upvarphi_{C^*}(Z;\eta)\right\}
+
(\Pn-\Pb)\left\{
\upvarphi_{\widehat C}(Z;\widehat\eta)-\upvarphi_{C^*}(Z;\widehat\eta)
\right\}
\label{eqn:clt-proof-1}\\
&\quad+
(\Pn-\Pb)\left\{
\upvarphi_{C^*}(Z;\widehat\eta)-\upvarphi_{C^*}(Z;\eta)
\right\}
\label{eqn:clt-proof-2}\\
&\quad+
\Pb\left\{
\upvarphi_{\widehat C}(Z;\widehat\eta)-\upvarphi_{C^*}(Z;\widehat\eta)
\right\}
+
\Pb\left\{
\upvarphi_{C^*}(Z;\widehat\eta)-\upvarphi_{C^*}(Z;\eta)
\right\},
\label{eqn:clt-proof-3}
\end{align}
which is obtained by simply adding and subtracting terms. This is a system of \(kp\) equations. We omit the fold indicator \(\mathbbm{1}(B=b)\) for simplicity.

The first term in \eqref{eqn:clt-proof-1} is asymptotically multivariate Gaussian by the multivariate central limit theorem, and hence is \(O_\Pb(n^{-1/2})\). Also, by Lemma \ref{lem:cross-term},
\begin{equation}
(\Pn-\Pb)\left\{
\upvarphi_{\widehat C}(Z;\widehat\eta)-\upvarphi_{C^*}(Z;\widehat\eta)
\right\}
=
O_\Pb\!\left(n^{-\min\{\beta,1\}/2}\right).
\label{eqn:upvarphi-empirical-process}
\end{equation}
Therefore, the two terms in \eqref{eqn:clt-proof-1} together are
\[
O_\Pb\!\left(n^{-\min\{\beta,1\}/2}\right).
\]

Under Assumption \ref{assumption:A5-np-consistency-condition}, the term in \eqref{eqn:clt-proof-2} is \(o_\Pb(n^{-1/2})\) by \citet[][Lemma 2]{kennedy2018sharp}.

We next consider the second term in \eqref{eqn:clt-proof-3}. It suffices to study the \(j\)-th block of \(\upvarphi_C\). Adding and subtracting terms gives
\begin{align*}
&\Pb\left[
\left(c_j-\upvarphi_1(Z;\widehat\eta)\right)\mathbbm{1}\{j=d(\widehat\mu,C^*)\}
-
\left(c_j-\upvarphi_1(Z;\eta)\right)\mathbbm{1}\{j=d(\mu,C^*)\}
\right]\\
&=
\Pb\left[
\left\{\upvarphi_1(Z;\widehat\eta)-\upvarphi_1(Z;\eta)\right\}
\mathbbm{1}\{j=d(\widehat\mu,C^*)\}
\right]\\
&\quad+
\Pb\left[
\left\{c_j-\upvarphi_1(Z;\eta)\right\}
\left(
\mathbbm{1}\{j=d(\widehat\mu,C^*)\}
-
\mathbbm{1}\{j=d(\mu,C^*)\}
\right)
\right].
\end{align*}
The first term is bounded by
\[
\left|
\Pb\left[
\left\{\upvarphi_1(Z;\widehat\eta)-\upvarphi_1(Z;\eta)\right\}
\mathbbm{1}\{j=d(\widehat\mu,C^*)\}
\right]
\right|
\lesssim
\max_a
\|\widehat\mu_a-\mu_a\|
\,
\|\widehat\pi_a-\pi_a\|
\,\bm 1_{(p)},
\]
see Remark \ref{rmk:E-bias-psi-1}.

For the second term, note that
\begin{align*}
   & \Pb\left[
\left\{c_j-\upvarphi_1(Z;\eta)\right\}
\left(
\mathbbm{1}\{j=d(\widehat\mu,C^*)\}
-
\mathbbm{1}\{j=d(\mu,C^*)\}
\right)
\right]\\
&=
\Pb\left[
(c_j-\mu)
\left(
\mathbbm{1}\{j=d(\widehat\mu,C^*)\}
-
\mathbbm{1}\{j=d(\mu,C^*)\}
\right)
\right].
\end{align*}
Let \(d=d(\mu,C^*)\) and \(\widehat d=d(\widehat\mu,C^*)\). Then
\begin{align*}
&\left|
\Pb\left[
(c_j^*-\mu)
\left(
\mathbbm{1}\{j=\widehat d\}-\mathbbm{1}\{j=d\}
\right)
\right]
\right|\\
&\le
\bm 1_{(p)}
\Pb\left[
\mathbbm{1}\left\{
\sqrt{f_{c_d^*}(\mu)}<\sqrt{f_{c_{\widehat d}^*}(\mu)}
\right\}
\left\{
\sqrt{f_{c_{\widehat d}^*}(\mu)}-\sqrt{f_{c_d^*}(\mu)}
\right\}
\right],
\end{align*}
where \(f_{c_j^*}(\mu)=\|\mu-c_j^*\|_2^2\). By the same argument used to derive
\eqref{eqn:app-phi_bias_2} and \eqref{eqn:app-phi_bias_2-2} in the proof of Lemma
\ref{lem:phi-bias-alphap1}, this is bounded by
\[
\max_a \|\widehat\mu_a-\mu_a\|_\infty^{\alpha+1}
+
\frac{1}{\kappa}
\max_a \|\widehat\mu_a-\mu_a\|_\infty
\|\widehat\mu_a-\mu_a\|_{\Pb,1}.
\]
Therefore,
\begin{align*}
\Pb\left\{
\upvarphi_{C^*}(Z;\widehat\eta)-\upvarphi_{C^*}(Z;\eta)
\right\}
\lesssim
\Bigg(
&\max_a \|\widehat\mu_a-\mu_a\|\,\|\widehat\pi_a-\pi_a\|\\
&+
\max_a \|\widehat\mu_a-\mu_a\|_\infty^{\alpha+1}\\
&+
\frac{1}{\kappa}
\max_a \|\widehat\mu_a-\mu_a\|_\infty
\|\widehat\mu_a-\mu_a\|_{\Pb,1}
\Bigg)\bm 1_{(p)}.
\end{align*}

Finally, consider the first term in \eqref{eqn:clt-proof-3}. The derivative of
\(C \mapsto \Pb\{\upvarphi_C(Z;\eta)\}\) at \(C=C^*\) is
\[
\frac{\partial}{\partial C}\Pb\left\{\upvarphi_C(Z;\eta)\right\}\Big|_{C=C^*}
=
2\operatorname{diag}\bigl(\bm 1_{(p)}p_1^*,\ldots,\bm 1_{(p)}p_k^*\bigr)
\equiv
M(C^*,\eta),
\]
where \(p_j^*=\Pb\{j=d(\mu,C^*)\}\). Since each \(p_j^*>0\), the matrix \(M(C^*,\eta)\) is nonsingular.

Also, \(\widehat C \to_\Pb C^*\) by Corollary \ref{cor:consistency-of-Chat-eif-estimator}. Therefore, by a mean value expansion in \(C\),
\[
\Pb\left\{
\upvarphi_{\widehat C}(Z;\widehat\eta)-\upvarphi_{C^*}(Z;\widehat\eta)
\right\}
=
M(C^*,\widehat\eta)(\widehat C-C^*)
+
o_\Pb\!\left(\|\widehat C-C^*\|_1\right).
\]
Moreover, \(M(C^*,\widehat\eta)=M(C^*,\eta)+o_\Pb(1)\) under Assumption
\ref{assumption:A5-np-consistency-condition}, so
\[
\Pb\left\{
\upvarphi_{\widehat C}(Z;\widehat\eta)-\upvarphi_{C^*}(Z;\widehat\eta)
\right\}
=
M(C^*,\eta)(\widehat C-C^*)
+
o_\Pb\!\left(\|\widehat C-C^*\|_1\right).
\]

Substituting the preceding bounds into \eqref{eqn:clt-proof-moment-condition}, we obtain
\begin{align*}
o_\Pb\!\left(n^{-1/2}\right)
&=
(\Pn-\Pb)\left\{\upvarphi_{C^*}(Z;\eta)\right\}
+
M(C^*,\eta)(\widehat C-C^*)
+
R_{2,n}\bm 1_{(p)}\\
&\quad+
O_\Pb\!\left(n^{-\min\{\beta,1\}/2}\right)
+
o_\Pb\!\left(\|\widehat C-C^*\|_1\right).
\end{align*}
Equivalently,
\begin{align*}
\widehat C-C^*
&=
-M(C^*,\eta)^{-1}(\Pn-\Pb)\left\{\upvarphi_{C^*}(Z;\eta)\right\}\\
&\quad+
O_\Pb(R_{2,n})
+
O_\Pb\!\left(n^{-\min\{\beta,1\}/2}\right)
+
o_\Pb\!\left(\|\widehat C-C^*\|_1\right).
\end{align*}
Since \(M(C^*,\eta)\) is nonsingular, this implies
\[
\|\widehat C-C^*\|_1
=
O_\Pb\!\left(R_{2,n}+n^{-\min\{\beta,1\}/2}\right).
\]

Finally, by \citet[][Lemma A]{pollard1982central}, under Assumption
\ref{assumption:A1-boundedness}, the map \(C \mapsto \Pb\{f_C(\mu)\}\) is differentiable at
\(C=C^*\). Hence
\begin{align*}
R(\widehat C)-R(C^*)
&=
\Pb\left\{f_{\widehat C}(\mu)-f_{C^*}(\mu)\right\}\\
&=
(\widehat C-C^*)^\top \gamma_{C^*}(\mu)
+
o_\Pb\!\left(\|\widehat C-C^*\|_1\right).
\end{align*}
Because \(C^*\) is a stationary point of \(R(C)\), the linear term vanishes, and therefore
\[
R(\widehat C)-R(C^*)
=
o_\Pb\!\left(\|\widehat C-C^*\|_1\right)
=
o_\Pb\!\left(R_{2,n}+n^{-\min\{\beta,1\}/2}\right).
\]
This completes the proof.
\end{proof}

\subsection{Proof of Corollary \ref{thm:clt-new}}

\begin{proof}[Proof of Corollary \ref{thm:clt-new}]
    The result immediately follows by noticing
    \begin{align*}
        (\Pn - \Pb)\left\{\upvarphi_{\widehat{C}}(Z;\widehat{\eta}) - \upvarphi_{C^*}(Z;\widehat{\eta}) \right\} = O_{\Pb}\left(n^{-1/2}\right),
    \end{align*}
    when $\alpha > 1$. Applying this to the original moment condition \eqref{eqn:clt-proof-moment-condition}, along with the other results in Section \ref{app-sec:proof-thm-Chat-root-nCAN}, we obtain
    \begin{align} \label{eqn:Chat-C-expansion}
        \widehat{C} - C^* = -M(C^*, \eta)^{-1}(\Pn - \Pb)\left\{ \upvarphi_{C^*}(Z;\eta) \right\} + O_\Pb(R_{2,n})+ o_\Pb\left(\Vert \widehat{C} - C^*\Vert_1 \right)+o_\Pb\left( \frac{1}{\sqrt{n}} \right).
    \end{align}
   Substituting the result of Theorem \ref{thm:Chat-root-nCAN} into \eqref{eqn:Chat-C-expansion} gives
    \begin{align*}
        \widehat{C} - C^* = -M(C^*, \eta)^{-1}(\Pn - \Pb)\left\{ \upvarphi_{C^*}(Z;\eta) \right\} +O_\Pb(R_{2,n})+o_\Pb\left( \frac{1}{\sqrt{n}} \right).
    \end{align*}
\end{proof}

The next remark outlines an alternative, more restrictive route to the empirical process bound based on an empty-margin condition.

\begin{remark}\label{app:rmk:empty-margin-donsker}
Fix $\bar\eta$, and let $\bar\mu$ denote its regression component. Suppose there exist $\kappa>0$ and a neighborhood $\mathcal N$ of $C^*$ such that
\[
\Pb\{\bar\mu(X)\in N_C(\kappa)\}=0
\qquad \text{for all } C\in\mathcal N.
\]
Then, for every $C\in\mathcal N$, the Voronoi label map
\[
x \mapsto d(\bar\mu(x),C)
\]
is locally constant almost surely, since no population mass lies in a $\kappa$-neighborhood of the relevant Voronoi boundaries. Hence the indicators
\[
\mathbbm 1\{j=d(\bar\mu,C)\}, \qquad j=1,\ldots,k,
\]
do not fluctuate locally in $C$. Consequently, on $\mathcal N$ the class
\[
\mathcal F_{\mathcal N}
:=
\{\upvarphi_C(\cdot;\bar\eta): C\in\mathcal N\}
\]
reduces to a finite union of piecewise linear, hence locally Lipschitz, parametric subclasses indexed by $C$. Since the parameter space is finite dimensional and can be taken compact locally around $C^*$, each such subclass has finite bracketing entropy, and therefore $\mathcal F_{\mathcal N}$ is Donsker; see, e.g., Lemma 19.24 of \citet{van2000asymptotic}. It follows that, whenever $\widehat C \in \mathcal N$ with probability tending to one,
\[
(\Pn-\Pb)\{\upvarphi_{\widehat C}(Z;\bar\eta)-\upvarphi_{C^*}(Z;\bar\eta)\}
=
O_\Pb(n^{-1/2}).
\]
\end{remark}


\end{document}